%% file: 0.main.tex
\def\llncs{0}
\def\fullpage{1}
\def\anonymous{0}
\def\draft{0}
\def\submission{0}

\ifnum\submission=1
\def\llncs{1}
\def\draft{0}
\def\anonymous{1}
\def\fullpage{0}
\fi

\ifnum\llncs=1
    \documentclass{llncs}
    \ifnum\fullpage=1
    \usepackage{fullpage}
    \fi
\else
    \documentclass[letterpaper]{article}
    \ifnum\fullpage=1
    \usepackage{fullpage}
    \fi
    \usepackage{microtype}
\fi

\usepackage{authblk}
\usepackage{xcolor} %
\usepackage{amsmath}
\usepackage{amssymb}
\usepackage{mathrsfs}

 \usepackage{amsthm}
\usepackage{enumitem} %
\usepackage[colorlinks=true,linkcolor=magenta,citecolor=blue,pagebackref=true,hypertexnames=false,pdftex,pdfpagelabels,bookmarks,hyperindex,hyperfigures]{hyperref}
\IfFileExists{breakcites.sty}{\usepackage{breakcites}}{}
\let\subparagraph\paragraph
\usepackage{titlesec}
\titleformat*{\paragraph}{\normalsize\bfseries}

\usepackage[capitalise,nameinlink]{cleveref}

\input{macro}

\crefname{algorithm}{Algorithm}{Algorithms}
\Crefname{algorithm}{Algorithm}{Algorithms}
\newenvironment{algorithmblock}[1]{%
  \refstepcounter{algorithm}%
  \par\addvspace{0.8em}%
  \noindent\begin{minipage}{\linewidth}%
  \hrule\smallskip
  \textbf{Algorithm \thealgorithm: #1}\par\smallskip
  \begin{enumerate}[
    leftmargin=1.62em,
    label=\arabic*.,
    itemsep=0.13em,
    topsep=0.15em
  ]
}{%
  \end{enumerate}\smallskip\hrule
  \end{minipage}\par\addvspace{0.8em}%
}

\usepackage{adjustbox}
\usepackage{placeins}
\usepackage{graphicx}
\usepackage{tikz}
\usetikzlibrary{positioning,arrows.meta,positioning,calc}
\title{Finding a Shortest Vector and More in $2^{n/2+o(n)}$ Time\\using $q$-ary Coset Difference Tree}
\hypersetup{
  pdftitle={Finding a Shortest Vector and More in 2^{n/2+o(n)} Time using q-ary Coset Difference Tree}
}
\ifnum\anonymous=1
  \hypersetup{pdfauthor={}}
\else
  \hypersetup{pdfauthor={Minki Hhan}}
\fi
\date{}

\ifnum\llncs=1
    \ifnum\anonymous=1
        \author{}
        \institute{}
    \else
        \author{
        Minki Hhan\inst{1} 
        }
        \institute{
        KAIST, Daejeon, Korea
        \\\email{minkihhan@kaist.ac.kr}
        }
    \fi
\else
    \author{Minki Hhan}
    \affil{{\small KAIST, Daejeon, Korea}
    \authorcr{\small minkihhan@kaist.ac.kr}}
    
\fi

\begin{document}

\maketitle

\begin{abstract}
This paper presents a new randomized algorithm for solving the exact shortest vector problem. For the $n$-dimensional lattice $\mathcal L$, our algorithm runs in time and space $2^{n/2+o(n)}$. 

Our algorithm can be viewed as a $q$-ary analogue of the midpoint Hessian for an odd prime $q$; more precisely, we use the fact that, for a shortest vector $v$, the gradient (rather than Hessian) of the periodic Gaussian function at $v/q$ is nearly proportional to $v$ (up to sign), even after aggregation over a relatively large random affine coset. We compute the relevant coset gradient along a chain of intermediate lattices using a combinatorial procedure inspired by Wagner's generalized birthday algorithm, yielding the $2^{n/2+o(n)}$ time and space complexity.

A variant of the algorithm solves the
exact closest vector problem on every input $(y,\mathcal L)$ with a distance guarantee $\operatorname{dist}(y,\mathcal L)\le 1.039\lambda_1(\mathcal L)$ within the same time and space complexity.
This guarantee holds for a random target and a random lattice drawn according to the Haar-Siegel measure. Thus, this algorithm solves a closest vector problem on such random instances in time and space $2^{n/2+o(n)}$.
\end{abstract}

\vfill
\paragraph{Author note \& AI use disclosure.}
The algorithm presented in this paper was suggested by ChatGPT using GPT-5.6 Sol in Ultra mode, building on a number of unsuccessful approaches previously explored by the author.
All of the technical analyses were carried out by ChatGPT using GPT-5.6 Sol and verified by the author.
The manuscript was written mostly by the author, based on the initial draft generated by AI.
The author takes full responsibility for the content;
any remaining errors are probably due to human error on the author's part.

The author is pleased that GPT-5.6 Sol's astonishing capabilities made it possible to obtain the present results through a combinatorial approach. This is particularly gratifying to him, since he had been exploring such approaches to lattice problems, including those based on generalized birthday and $k$-SUM algorithms, since reading \cite{AS17} as a graduate student.

\clearpage

\ifnum\submission=1
\else
    \newpage
    \setcounter{tocdepth}{2}
    \tableofcontents
    \newpage
\fi

\input{1.introduction}

\input{2.prel}
\input{3.derivatives}
\input{4.svp}
\input{5.cvp}

\bibliographystyle{plain}
\bibliography{ref,abbrev3,crypto}

\end{document}

%% file: macro.tex
\ifnum\draft=1
	\newcommand{\minki}[1]{\textcolor{blue}{$\langle\langle$Minki: #1$\rangle\rangle$}}
\else
	\newcommand{\minki}[1]{}
\fi

\ifnum\llncs=0
	\newtheorem{theorem}{Theorem}[section]
	\newtheorem{lemma}[theorem]{Lemma}
	\newtheorem{corollary}[theorem]{Corollary}
	
	\newtheorem{definition}[theorem]{Definition}

	\newtheorem{claim}[theorem]{Claim}

	\theoremstyle{remark}
	\newtheorem{remark}[theorem]{Remark}

\else
	\spnewtheorem{algorithm}{Algorithm}{\bfseries}{\rmfamily}

	\spnewtheorem{claim}{Claim}{\bfseries}{\itshape}
\fi
\crefname{appendix}{Appendix}{Appendices}
\Crefname{appendix}{Appendix}{Appendices}

\crefname{claim}{Claim}{Claims}
\Crefname{claim}{Claim}{Claims}

\newcommand{\N}{\mathbb{N}}

\renewcommand{\epsilon}{\varepsilon}

\newcommand{\R}{\mathbb{R}}
\renewcommand{\Re}{\mathsf{Re}}
\renewcommand{\Im}{\mathsf{Im}}
\newcommand{\Z}{\mathbb{Z}}
\newcommand{\F}{\mathbb{F}}
\newcommand{\E}{\mathbb{E}}
\newcommand{\Prob}{\mathbb{P}}

\newcommand{\cL}{\mathcal{L}}

\providecommand{\ip}[2]{}
\renewcommand{\ip}[2]{\left\langle #1,#2\right\rangle}

\newcommand{\dist}{\operatorname{dist}}

\newcommand{\polylog}{\operatorname{polylog}}

\newcommand{\ind}{\mathbf{1}}

\newcommand{\SVP}{\mathsf{SVP}}

\newcommand{\CVP}{\mathsf{CVP}}

\newcommand{\BDD}{\mathsf{BDD}}
\newcommand{\DGS}{\mathsf{DGS}}

\DeclareMathOperator{\GL}{GL}
\DeclareMathOperator{\SL}{SL}

\newcommand{\bit}{\{0,1\}}

\newcommand{\C}{\mathbb{C}}

\newcommand{\norm}[1]{\left\lVert #1 \right\rVert}

%% file: 1.introduction.tex
\section{Introduction}
\label{sec:introduction}
An $n$-dimensional full-rank lattice $\cL=\cL(B)$ is the set of all integer combinations of the columns of a matrix $B\in\R^{n\times n}$. The shortest vector problem, $\SVP$, asks to find a shortest nonzero vector of $\cL$, and the closest vector problem, $\CVP$, asks for a lattice vector closest to a given target. 
These are central problems in the geometry of numbers, computational complexity, and the concrete security of lattice-based cryptography.

The study of computational lattice problems is shaped by a discrepancy between efficient coarse approximation algorithms and hardness of exact or tight approximations. Polynomial-time lattice-reduction algorithms give useful exponential approximations to $\SVP$ and $\CVP$ and have led to applications in polynomial factorization, Diophantine approximation, integer programming, and cryptanalysis
\cite{LLL82,Len83,STOC:Kannan83,FOCS:DadPeiVem11,FOCS:Shamir82,LO85}.
In contrast, exact and sufficiently accurate versions of $\SVP$ and $\CVP$ exhibit strong NP-hardness \cite{STOC:Ajtai98,SICOMP:Mic01,JACM:Khot05,STOC:HavReg07},
and serve as security foundation of lattice-based cryptography.
The complexity to solve these lattice problems is therefore a central benchmark for lattice algorithms, resulting in the ceaseless improvements in worst-case lattice algorithms
\cite{STOC:Kannan83,C:HanSte07,STOC:AjtKumSiv01,EPRINT:PujSte09,STOC:MicVou10,STOC:ADRS15,STOC:ADS15,SICOMP:ACKS25,CCL18,AS17}. Very recently, $2^{0.5596n+o(n)}$, $2^{0.6039n+o(n)}$ and $2^{0.7314n+o(n)}$ time worst-case $\SVP$ algorithms were suggested \cite{Kim26,cryptoeprint:2026/1587,Hhan26,cryptoeprint:2026/1844}, breaking $2^n$ time barrier.

Our main results are improved randomized algorithms for these lattice problems.
\begin{theorem}[Informal]
\label{thm:intro-results}
There are randomized classical algorithms with success probability at least $2/3$ and worst-case time and space $2^{n/2+o(n)}$ for the following problems:
\begin{enumerate}[nosep]
\item $\SVP$ on arbitrary lattices;
\item $\CVP$ on targets $y$ satisfying $\dist(y,\cL)\le\alpha\lambda_1(\cL)$ for $\alpha<1.039\cdots=\frac1{2\sqrt{t_0}}$ where $t_0=0.23147\cdots.$
\end{enumerate}
\end{theorem}

\input{table}

The first result improves the best known provable complexity for $\SVP$ problem of the recent best approaches \cite{Kim26,cryptoeprint:2026/1587,Hhan26,cryptoeprint:2026/1844} building on the various $\SVP$ algorithms mentioned above. 
It also solves constant-factor approximate $\SVP$ in time $2^{n/2+o(n)}$, which in turn gives $\widetilde O(n^c)$-$\SVP$ algorithm in time $2^{n/(2c+2)+o(n)}$ using \cite[Theorem~5.3]{EC:AggLiSte21}.
We emphasize that heuristic algorithms discussed below have better complexity than ours, and our algorithm does not have any practical implication in lattice-based cryptography.

The second result improves the distance guarantee achievable for $\CVP$ in time $2^{n/2+o(n)}$, where the previous best was $\dist(y,\cL)<0.422\lambda_1(\cL)$ \cite{CCC:DRS14,STOC:ADRS15}.
The dimension-preserving reduction of
\cite[Theorem~6.1]{CCC:DRS14} with our $\CVP$ algorithm gives a randomized algorithm for $3.67312\text{-}\CVP$
in worst-case time and space $2^{n/2+o(n)}$.
This gives a randomized $2^{n/2+o(n)}$ time and space algorithm for the approximate shortest independent vector problem and the successive minima problem with the same approximate factor $3.67312$ following \cite{SODA:Micciancio08}.

We also find an application to random $\CVP$ (\cref{cor:random-cvp}); for a Haar-random lattice $\cL$ due to Siegel \cite{Sie45,EC:PouShe26} and random target vector $y$, we show that $\dist(y,\cL)/\lambda_1(\cL) \le 1+ O(n^{-1/3})$ holds with overwhelming probability, meaning that we can solve $\CVP$ for random target in random lattice in time and space $2^{n/2+o(n)}$.

\paragraph{More related work.}

Heuristic lattice sieving gives the best known asymptotic estimates for exact $\SVP$ and underlies concrete security estimates of lattice-based cryptography
\cite{NguVid08,SODA:MicVou10,C:Laarhoven15,EPRINT:BecGamJou15,BLS16,PKC:HerKirLaa18}.
The best heuristic complexities for $\SVP$ and $\CVP$ are both about $2^{0.2925n+o(n)}$ classically \cite{SODA:BDGL16,SAC:Laarhoven16}. 
Heuristic quantum algorithms for 
$\SVP$ have been studied extensively, and the best complexity is $2^{0.2563n+o(n)}$ \cite{AC:AonNguShe18,AC:KMPM19,AC:ChaLoy21,EC:BCSS23}.

Recent work has also strengthened deterministic hardness. An OpenAI technical report gave deterministic $n^{1/400}$-factor hardness for $\CVP$ \cite{OpenAI2026}. 
Wan proved deterministic hardness of Euclidean $\SVP$ within every constant factor \cite{Wan26}.
Fine-grained hardness of both problems are also well-studied.

Fine-grained hardness results are well-studied for $\SVP$ and $\CVP$ and approximate versions
\cite{FOCS:BenGolSte17,STOC:AggSte18,SODA:ABGS21,ICALP:AggGupMorZha26},
and several barriers to strengthening them have been established
\cite{SODA:ABGS21,FOCS:AggKum23,ICALP:HuaKoWan26}.
No explicit constant lower bound in the exponent is known for both $\SVP$ and $\CVP$ in the Euclidean norm.

\paragraph{Acknowledgment.} The author is grateful to Thijs Laarhoven for discussion about random lattices.

\subsection{Technical overview}
We present a high-level idea of our algorithm, which is inspired by the recent midpoint Hessian algorithm \cite{Hhan26}, where the main idea is that, for a shortest vector $v$, the Hessian of the periodic Gaussian at $v/2$ reveals the direction of $v$. 
We will use the gradient of the periodic Gaussian instead, which has a clearer interpretation in our case.
In both cases, knowing the direction of $v$ up to an inverse polynomial error suffices to find $v$ itself in time $2^{o(n)}$ using the preprocessing bounded distance decoding (BDD) algorithm from \cite{SICOMP:ACKS25} at the preprocessing cost of $2^{0.5n+o(n)}$.

The starting point of our algorithm is similar, that for a shortest vector $v$, the gradient of the periodic Gaussian at $v/q$ for some odd prime $q$ is almost proportional to $v$.\footnote{For more details, we believe our algorithm should also work with the Hessian instead of gradient, but with a more complicated analysis. On the other hand, for the midpoint $v/2$ as in \cite{Hhan26}, the gradient vanishes so the choice of Hessian is unavoidable.} This choice increases the search space from $2^n$ midpoints to $q^n$ $q$-ary points, but also introduces a room for combinatorial optimizations using several cosets as described below. 
We discuss the barriers to go below $2^{0.5n}$ based on our algorithm and why this approach does not improve the midpoint Hessian algorithm \cite{Hhan26} in \cref{rem:belowhalf,rem:midpointmaynot}.

We focus on the algorithmic structure and the complexity in the overview. 
The rigorous analysis of the several bounds on estimations will be postponed to the main body of this paper.

\paragraph{Discrete and periodic Gaussians.}
Let $\cL$ be a full-rank lattice of dimension $n$, and let $B$ be its basis.
The dual lattice of $\cL$ is defined by $\cL^*=\{y: \ip{y}{x} \in \Z \text{ for every }x \in \cL\}$, whose basis is $B^{-T}.$

For the width parameter $s>0$, the discrete Gaussian is defined by
\begin{align}
\label{eqn: intro_DGS}
D_{\cL,s}(x):=\frac{\rho_s(x)}{\rho_s(\cL)}
\quad (x\in\cL)
\qquad\text{where }
\rho_s(x):=\exp\left(-\pi\norm{x}^2/s^2\right).
\end{align}
The periodic Gaussian of $\cL$ is a $\cL$-periodic function defined by
\[
F_{1/s}(z):=
\frac{\rho_{1/s}(\cL+z)}{\rho_{1/s}(\cL)}
\]
which satisfies $F_{1/s}(z) = F_{1/s}(w)$ if $z-w \in \cL$.
Let $q$ be an odd prime and identify
$q^{-1}\cL/\cL$ with $\F_q^n$ using a basis $B$ of $\cL$, or more formally by the map $w\bmod q \mapsto \frac{Bw}{q}+\cL.$
The dual basis $B^{-T}$ similarly identifies $\cL^*/q\cL^*$ with
$\F_q^n$ via $z\bmod q\mapsto B^{-T}z+q\cL^*$.
This gives
the following formula, which we will use extensively:
\[
\exp\left(2\pi i\left\langle\frac{Bw}{q},B^{-T}z\right\rangle\right)
=
\omega_q^{\langle w,z\rangle}.
\]

For
$w\in\F_q^n$, we define the $q$-ary periodic Gaussian and its gradient, which will be our main tools, by
\[
f(w):=
F_{1/s}\left(\frac{Bw}{q}\right),
\qquad
G(w):= \nabla_z F_{1/s}(z)|_{z=Bw/q} = 
\nabla F_{1/s}\left(\frac{Bw}{q}\right),
\]
where any integral lift of $w$ may be used.

We use two different representations of the $q$-ary periodic Gaussian function and its gradient.
Expanding the definition of $F_{1/s}$ and differentiating gives the
primal representations:
\begin{align*}
f(w)
&=
\frac{1}{\rho_{1/s}(\cL)}
\sum_{x\in\cL}
\rho_{1/s}\left(x+\frac{Bw}{q}\right),
\notag\\
G(w)
&=
-\frac{2\pi s^2}{\rho_{1/s}(\cL)}
\sum_{x\in\cL}
\left(x+\frac{Bw}{q}\right)
\rho_{1/s}\left(x+\frac{Bw}{q}\right).
\end{align*}
This will be used to observe the geometric behavior of these values. On the other hand, 
Poisson summation formula gives the dual representations, which will be used in the estimations using discrete Gaussian samples obtained from the sampler in \cite{STOC:ADRS15}:
\[
f(w)
=
\E_{X\sim D_{\cL^*,s}}
\left[
\omega_q^{\ip{[X]_q}{w}}
\right],\qquad
G(w)
=
2\pi i\,
\E_{X\sim D_{\cL^*,s}}
\left[
X\omega_q^{\ip{[X]_q}{w}}
\right],
\]
where $[X]_q:=B^TX\bmod q$ denotes the representations of $X\in \cL^*/q\cL^*$ in $\F_q^n$ and
$\omega_q:=\exp(2\pi i/q)$.

\paragraph{A gradient pointing toward a shortest vector.}
Let $v$ be a shortest nonzero vector of $\cL$ and let
\[
w_v:=B^{-1}v\bmod q
\]
be the representation of $v$ in $\F_q^n$, so that $Bw_v/q = v/q$ modulo $\cL$. For an appropriately chosen width, 
the term $v/q$ dominates the primal gradient, giving
\[
G(w_v)
=
-\frac{2\pi s^2}{\rho_{1/s}(\cL)}
\sum_{x\in\cL}
\left(x+\frac{v}{q}\right)
\rho_{1/s}\left(x+\frac{v}{q}\right)
=
-\frac{2\pi s^2}{q}
\frac{\rho_{1/s}(v/q)}{\rho_{1/s}(\cL)}
(v+e),
\qquad
\norm e\le 2^{-\Omega(n)}\norm v,
\]
where $v$ comes from the term $x=0$, while $e$ collects the remaining
terms. For every $x\ne0$,
\[
\frac{\rho_{1/s}(x+v/q)}{\rho_{1/s}(v/q)}
=
\exp\left(
-\pi s^2\left(
\norm{x+v/q}^2-\norm{v/q}^2
\right)
\right)
\le \exp(-\Omega(n)),
\]
and a standard Gaussian-mass bound shows that their summation
satisfies $\norm e\le 2^{-\Omega(n)}\norm v$.
Thus $G(w_v)$ points approximately in the direction of $v$ up to the minus sign.
Thus, if $w_v$ were known, the gradient at $w_v$ would reveal the direction of $v$.

The difficulty is that $w_v$ is one of $q^n$ possible points of
$q^{-1}\cL/\cL$. Evaluating the gradient separately at all of them would be much too expensive.
The same problem appeared in \cite{Hhan26} when the complexity going below $2^n$, where they suggested to use the cosets to optimize.

\paragraph{Working with one coset.}
Let $H\le\F_q^n$ be a subspace and let
$\cL_H:=\cL+\frac{1}{q}BH$ be the corresponding intermediate lattice
between $\cL$ and $q^{-1}\cL$.
Instead of treating the points of $\F_q^n$ separately, for each coset
$\kappa+H\in\F_q^n/H$ consider the aggregate gradient
\[
G(\kappa+H)
:=
\sum_{w\in\kappa+H}G(w)
=
-\frac{2\pi s^2}{\rho_{1/s}(\cL)}
\sum_{z\in\cL_H+B\kappa/q}
z\rho_{1/s}(z).
\]
Thus summing the $q$-ary gradients over $\kappa+H$, which we call the
coset gradient, 
can be understood as a (scalar multiple of) the periodic Gaussian gradient over the corresponding
coset of the intermediate lattice $\cL_H$.
In particular, if $w_v\in\kappa+H$, then $G(w_v)$ appears as one of the
terms in $G(\kappa+H)$ and points toward $v$ up to minus sign.
The coset size will be chosen so that this term remains detectable in
the sum, as analyzed later.

On the dual side, by the character orthogonality identity
$\sum_{u\in H}\omega_q^{\ip{u}{y}}
=|H|\mathbf 1_{\{y\in H^\perp\}}$, we have
\begin{align*}
G(\kappa+H)
&=
2\pi i\,
\E_{X\sim D_{\cL^*,s}}
\left[
X\,
\omega_q^{\ip{[X]_q}{\kappa}}
\sum_{u\in H}\omega_q^{\ip{[X]_q}{u}}
\right]
\notag\\
&=
2\pi i\,|H|\,
\E_{X\sim D_{\cL^*,s}}
\left[
X\,
\omega_q^{\ip{[X]_q}{\kappa}}
\mathbf 1_{\{[X]_q\in H^\perp\}}
\right].
\end{align*}
Hence summing over a primal coset restricts the dual residue to
$H^\perp$.
The approach of \cite{Hhan26} uses the near-uniform distribution of
discrete Gaussian samples among the residue classes to obtain samples
conditioned on this event.
Since its density is approximately
$|H^\perp|/q^n=1/|H|$, this incurs a multiplicative overhead $|H|$.

\paragraph{Meet-in-the-middle.}
We instead multiply the gradient by $f$ before taking the coset sum.
By independence and $f(-w)=f(w)$, we can write $G(w)f(w)$ using two discrete Gaussian samples as follows
\begin{align*}
G(w)f(w)=G(w)f(-w)
&=
\E_{X_0\sim D_{\cL^*,s}}
\left[
2\pi iX_0\omega_q^{\ip{w}{[X_0]_q}}
\right]
\E_{X_1\sim D_{\cL^*,s}}
\left[
\omega_q^{-\ip{w}{[X_1]_q}}
\right]
\notag\\
&=
\E_{X_0,X_1\sim D_{\cL^*,s}}
\left[
2\pi iX_0
\omega_q^{\ip{w}{[X_0]_q-[X_1]_q}}
\right].
\end{align*}
Consequently, the sum over the coset $\kappa+H$ becomes
\begin{align}
\sum_{w\in\kappa+H}G(w)f(w)
&=
\E_{X_0,X_1\sim D_{\cL^*,s}}
\left[
2\pi iX_0\,
\omega_q^{\ip{\kappa}{[X_0]_q-[X_1]_q}}
\sum_{u\in H}
\omega_q^{\ip{u}{[X_0]_q-[X_1]_q}}
\right]
\notag\\
&=
|H|\,
\E_{X_0,X_1\sim D_{\cL^*,s}}
\left[
2\pi iX_0\,
\omega_q^{\ip{\kappa}{[X_0]_q-[X_1]_q}}
\mathbf 1_{\{[X_0]_q-[X_1]_q\in H^\perp\}}
\right].
\label{eqn:overview-two-sample-coset}
\end{align}
Thus the factor $f$ replaces the condition on one discrete Gaussian
sample by a condition on the difference of two ordinary samples.

More generally, a phase-weighted sum over a coset allows this difference to
lie in an arbitrary prescribed coset.
We explain this phenomenon by considering the coordinate subspaces
\[
\F_q^n=H_1\oplus H_2\oplus H_3
\]
and randomizing this decomposition using a uniform
$U\in\GL_n(\F_q)$;
here $H_2,H_3$ are introduced for comparison with the later example.
For now, we consider $H_2 \oplus H_3$ together, which parametrizes the cosets of $H_1$
Write $Y_r:=U[X_r]_q$ for each $r$.
For $\theta\in H_2\oplus H_3$ and $\delta\in H_1$, we define the phase-weighted version of \cref{eqn:overview-two-sample-coset} by
\begin{align}
A_\delta(\theta)
&:=
\sum_{\alpha\in H_1}
\omega_q^{-\ip{\alpha}{\delta}}
G\left(U^T(\theta+\alpha)\right)
f\left(U^T(\theta+\alpha)\right)
\notag\\
&=
|H_1|\,
\E_{X_0,X_1\sim D_{\cL^*,s}}
\left[
2\pi iX_0\,
\omega_q^{\ip{\theta}{Y_0-Y_1}}
\mathbf 1_{\{(Y_0-Y_1)_{H_1}=\delta\}}
\right].
\label{eqn:overview-two-sample-estimator}
\end{align}
The equality follows from 
$\sum_{\alpha\in H_1}
\omega_q^{\ip{\alpha}{z-\delta}}
=
|H_1|\mathbf 1_{\{(z)_{H_1}=\delta\}}.$
Thus a phase-weighted sum of $Gf$ over the frequency coset
$U^T(\theta+H_1)$ is equivalent to prescribing the $H_1$-component
of the sample difference.
The unweighted coset sum corresponds to $\delta=0$.

Algorithmically, we estimate $A_\delta(\theta)$ from two independent
lists of discrete Gaussian samples by matching pairs satisfying
$(Y_0-Y_1)_{H_1}=\delta$ and averaging their contributions
$2\pi iX_0\omega_q^{\ip{\theta}{Y_0-Y_1}}$. More details will follow.

\paragraph{A four-sample tree.}
We extend this single difference condition to a tiny tree with three difference conditions. We consider four discrete Gaussian samples $X_0,\cdots,X_3$.
Let $Y_r:=U[X_r]_q$ and define
\[
D_{1,0}:=Y_0-Y_1,\qquad
D_{1,1}:=Y_2-Y_3,\qquad
D_{2,0}:=D_{1,0}-D_{1,1}.
\]
Choose independently and uniformly
\[
\delta_{1,0},\delta_{1,1}\in H_1,
\qquad
\delta_{2,0}\in H_2,\qquad\boldsymbol{\delta} = (\delta_{1,0},\delta_{1,1},\delta_{2,0})
\]
and impose the following conditions
\[
(D_{1,0})_{H_1}=\delta_{1,0},
\qquad
(D_{1,1})_{H_1}=\delta_{1,1},
\qquad
(D_{2,0})_{H_2}=\delta_{2,0}.
\]

\input{fig_tree}

This gives the quantity analogous to \cref{eqn:overview-two-sample-estimator} for $\theta\in H_3$ with the three difference conditions
\begin{align*}
A_{\boldsymbol\delta}(\theta)
:={}&
|H_1|^2|H_2|\,
\E_{X_0,\ldots,X_3\sim D_{\cL^*,s}}
\left[
2\pi iX_0\,
\omega_q^{\ip{\theta}{D_{2,0}}}
\right.
\notag\\[-1mm]
&\hspace{21mm}\left.
{}\cdot
\mathbf 1_{\{(D_{1,0})_{H_1}=\delta_{1,0}\}}
\mathbf 1_{\{(D_{1,1})_{H_1}=\delta_{1,1}\}}
\mathbf 1_{\{(D_{2,0})_{H_2}=\delta_{2,0}\}}
\right].
\end{align*}
Expanding the three conditions by character orthogonality gives the following alternative representation
\begin{align}
A_{\boldsymbol\delta}(\theta)
=
\sum_{\substack{\alpha_0,\alpha_1\in H_1\\ \beta\in H_2}}
&
\omega_q^{
-\ip{\alpha_0}{\delta_{1,0}}
+\ip{\alpha_1}{\delta_{1,1}}
-\ip{\beta}{\delta_{2,0}}
}
\notag\\[-1mm]
&\quad\cdot
G\left(U^T(\theta+\beta+\alpha_0)\right)
f\left(U^T(\theta+\beta+\alpha_0)\right)
f\left(U^T(\theta+\beta+\alpha_1)\right)^2.
\label{eqn:intro_Adelta}
\end{align}

The expansion contains, for a suitable $\theta$, a contribution
associated with a shortest vector $v$.
To see this explicitly, recall that
$w_v:=B^{-1}v\bmod q$ is the $\F_q^n$ representation of $v$, and let
\[
k:=U^{-T}w_v
=
k_1+k_2+k_3,
\qquad
k_j:=(k)_{H_j}.
\]
At $\theta_*:=k_3$, the choice
$\beta=k_2$ and $\alpha_0=\alpha_1=k_1$ contributes
\[
c_{\boldsymbol\delta}(k)G(w_v)f(w_v)^3,
\qquad
c_{\boldsymbol\delta}(k)
:=
\omega_q^{
-\ip{k_1}{\delta_{1,0}}
+\ip{k_1}{\delta_{1,1}}
-\ip{k_2}{\delta_{2,0}}
},
\]
where $|c_{\boldsymbol\delta}(k)|=1$.
This particular term is approximately proportional to $v$ up to a
unit complex factor.
The components $k_1,k_2$ are handled implicitly by the sums over
$\alpha_0,\alpha_1,\beta$, while $k_3$ is the only component explicitly
indexed by $\theta$, with the desired value $\theta_*=k_3$.
The later analysis shows that the remaining terms are sufficiently
small.

More concretely, we estimate $A_{\boldsymbol\delta}(\theta)$ from four
independent lists of discrete Gaussian samples by first matching pairs
satisfying the prescribed $H_1$-conditions, then matching their
differences according to the prescribed $H_2$-condition, and summing,
with the appropriate normalization, their contributions
$2\pi iX_0\omega_q^{\ip{\theta}{D_{2,0}}}$.

\paragraph{The general tree.}
The four-sample construction extends naturally to any
$p=2^{m-1}$ samples.
For now, let $q$ be an odd prime and assume for simplicity that
$m\mid n$. Decompose
\[
\F_q^n
=
H_1\oplus\cdots\oplus H_m,
\qquad
|H_1|=\cdots=|H_m|=:Q,
\]
where the $H_j$ are coordinate subspaces, and choose a uniform
$U\in\GL_n(\F_q)$.

For independent $X_0,\ldots,X_{p-1}\sim D_{\cL^*,s}$, define
\[
D_{0,r}:=U[X_r]_q,
\qquad
D_{j,a}:=D_{j-1,2a}-D_{j-1,2a+1}
\]
for $j=1,\ldots,m-1$ and $a=0,\ldots,p/2^j-1$.
At every node $(j,a)$, choose an independent uniform
$\delta_{j,a}\in H_j$ and impose the difference constraints
\[
(D_{j,a})_{H_j}=\delta_{j,a}.
\]
Thus, at level $j$, pairs of elements are matched according to their
$H_j$-components and replaced by their difference before being passed
to the next level.

For $\theta\in H_m$, define the corresponding ideal quantity satisfying all the constraints $\boldsymbol{\delta}=(\delta_{j,a})_{j,a}$
\begin{align}
A_{\boldsymbol\delta}(\theta)
:=
Q^{p-1}
\E_{X_0,\ldots,X_{p-1}\sim D_{\cL^*,s}}
\left[
2\pi iX_0\,
\omega_q^{\ip{\theta}{D_{m-1,0}}}
\prod_{j=1}^{m-1}
\prod_{a=0}^{p/2^j-1}
\mathbf 1_{\{(D_{j,a})_{H_j}=\delta_{j,a}\}}
\right].
\label{eqn:overview-general-tree-estimator}
\end{align}
The normalization is $Q^{p-1}$ because the binary tree has $p-1=\sum_{j=1}^{m-1}\frac{p}{2^j}$ internal nodes.

Let
$\theta_*:=(U^{-T}w_v)_{H_m}.$
As in the four-sample example, the expansion of
\cref{eqn:overview-general-tree-estimator} contains the desired term approximately pointing a direction $v$
\[
c_{\boldsymbol\delta}(U^{-T}w_v)
G(w_v)f(w_v)^{p-1},
\qquad
\left|c_{\boldsymbol\delta}(U^{-T}w_v)\right|=1.
\]
Again, the components of $U^{-T}w_v$ in $H_1,\ldots,H_{m-1}$ are handled
implicitly by the character sums, while its $H_m$-component is indexed
by $\theta_*$. 
The remaining terms in the expansion are small with high probability:
\begin{align}
\left\|
A_{\boldsymbol\delta}(\theta_*)
-
c_{\boldsymbol\delta}(U^{-T}w_v)
G(w_v)f(w_v)^{p-1}
\right\|
\le
\frac{s^2\norm v}{q}Q^{-1/2}2^{o(n)}.
\label{eqn:overview-general-tree-approximation}
\end{align}
We will discuss the size of the main term later in the complexity analysis.

\paragraph{Estimation and parameter choice.}
We now explain how to estimate the ideal quantity in
\cref{eqn:overview-general-tree-estimator}.
Recall that $p=2^{m-1}$, and let $N$ be a list-size parameter to be
chosen below.
Draw independent samples
\[
X_{r,i}\sim D_{\cL^*,s},
\qquad
r\in[0,p),\quad i\in[N],
\]
viewed as $p$ lists of size $N$.
For $\boldsymbol i=(i_0,\ldots,i_{p-1})\in[N]^p$, define
$D_{j,a}(\boldsymbol i)$ recursively from
\[
D_{0,r}(\boldsymbol i):=U[X_{r,i_r}]_q,
\qquad
D_{j,a}(\boldsymbol i)
:=
D_{j-1,2a}(\boldsymbol i)-D_{j-1,2a+1}(\boldsymbol i).
\]
We then estimate $A_{\boldsymbol\delta}(\theta)$ by
\begin{align}
\widehat A_{\boldsymbol\delta}(\theta)
:=
\frac{Q^{p-1}}{N^p}
\sum_{\boldsymbol i\in[N]^p}
2\pi iX_{0,i_0}\,
\omega_q^{\ip{\theta}{D_{m-1,0}(\boldsymbol i)}}
\cdot
\prod_{j=1}^{m-1}
\prod_{a=0}^{p/2^j-1}
\mathbf 1_{\{
(D_{j,a}(\boldsymbol i))_{H_j}
=
\delta_{j,a}
\}}.
\label{eqn:overview-general-tree-empirical}
\end{align}
Note that this estimation requires the summation over $N^p$ different $\boldsymbol{i}$.
The gradient difference tree algorithm evaluates this sum without enumerating all
$N^p$ tuples.

The computation of
\cref{eqn:overview-general-tree-empirical} for a fixed
$\theta\in H_m$ proceeds as follows:
\begin{enumerate}[nosep]
\item Initialize the $p$ lists with the discrete Gaussian samples. This will be the level 0.
\item For each $j=1,\ldots,m-1$, combine pairs from level $j-1$ to form $D_{j,a}
=
D_{j-1,2a}-D_{j-1,2a+1},$
collecting only those satisfying
\[
(D_{j,a})_{H_j}=\delta_{j,a}.
\]
\item Sum
\[
2\pi iX_{0,i_0}\,
\omega_q^{\ip{\theta}{D_{m-1,0}(\boldsymbol i)}}
\]
over all tuples $\boldsymbol i$ satisfying every difference condition,
and multiply the result by $Q^{p-1}/N^p$.
\end{enumerate}
By construction, this evaluates
$\widehat A_{\boldsymbol\delta}(\theta)$ exactly.
We analyze its time and space complexity below.
To evaluate all $\theta\in H_m$ simultaneously, define and compute the following values for all $ b\in H_m$:
\begin{align*}
V_{\boldsymbol\delta}(b)
:=
\sum_{\substack{\boldsymbol i\in[N]^p\\
\text{all difference conditions}\\
(D_{m-1,0}(\boldsymbol i))_{H_m}=b}}
2\pi iX_{0,i_0}.
\end{align*}
Then, we have the representation
\[
\widehat A_{\boldsymbol\delta}(\theta)
=
\frac{Q^{p-1}}{N^p}
\sum_{b\in H_m}
V_{\boldsymbol\delta}(b)\omega_q^{\ip{\theta}{b}},\]
so one fast Fourier transform over $H_m$ computes the estimator for every
$\theta\in H_m$.

\paragraph{Complexity and parameter choice.}
Now we discuss the complexity of the above algorithm and the parameter choice for the optimized complexity.

Let $L_j$ be the maximum size of a list at level $j$, with $L_0=N$.
Since each $\delta_{j,a}$ is uniform in $H_j$, every fixed pair is
collected with probability $1/Q$. Thus, for input lists
of size at most $L_{j-1}$, the expected output size for the next level is at most
$L_{j-1}^2/Q$.
Grouping by the relevant level $j$ components in complexity\footnote{Space can be slightly reduced, which is unimportant in the asymptotic analysis.} therefore gives
\[
S,T=
\widetilde O\left(
pN+
\sum_{j=1}^{m-1}
\frac{p}{2^j}(L_{j-1}+L_j)
+
Q\log Q
\right),\]
where $pN$ is for the discrete Gaussian sampling (here we assume $N \ge 2^{0.5n}$ for the simple use of the discrete Gaussian sampling algorithm in \cite{STOC:ADRS15}), and the middle term is about collecting the difference satisfying the conditions at each level, and the last term is for the final fast Fourier transform.

To balance the first and the last terms, we choose $N=Q p^{O(1)}$.
We also assume $L_j = L_{j-1}^2/Q$ for a simple exposition, so that
\[
L_j = Q p^{O(2^j)} = Qp^{O(p)}.
\]
By choosing $p\ll n/\log n$, we have $p^{O(p)}=2^{o(n)}$ and $L_j=Q\cdot 2^{o(n)}$ for all $j$. Then the complexity becomes
\[
T=(N+Q)2^{o(n)},\qquad S=Q2^{o(n)}.
\]
Choosing $N=Q p^{O(1)}$ and recalling $Q=(q^{n})^{1/m}$ for $m=\log_2 p +1$, both bounds are $Q2^{o(n)}=2^{n/2+o(n)}$. In other words, finding the optimal complexity depends on the parameter $q$. 

The constraint on $q$ is related to the size of the desired term and the estimation error
\begin{align}
\norm{G(w_v)f(w_v)^{p-1}}
=
\frac{s^2\norm v}{q}
2^{-(4tp/q^2-o(1))n}, \qquad
(\text{estimation error})
=
\frac{s^2\norm v}{q}
Q^{-1/2}2^{o(n)}
\label{eqn:intro-estimate}
\end{align}
hold with high probability.
Here $t>0$ is a fixed constant determined by the Gaussian width and therefore does not affect the asymptotic exponent in this optimization.
The first bound (\cref{lem:odd-local}) can be shown by the standard Gaussian mass inequalities, and the second bound (\cref{lem:svp-finite-list-second-moment}), together with the above bound \cref{eqn:overview-general-tree-approximation} (i.e., \cref{lem:svp-nonprincipal-mass}), is rather complicated to prove. 

Since $Q=q^{n/m}$ and $m=\log_2p+1$, the above bounds and \cref{eqn:overview-general-tree-approximation} show that the desired term is detectable
when
\[
\frac{4tp}{q^2}
\lesssim
\frac{\log_2q}{2\log_2 p}.
\]
This asymptotically enforces $q\gtrsim \sqrt p$ or $\log_2 q \gtrsim 0.5 \log_2 p$, which makes
\[
Q = q^{n/m} \approx 2^{(\log_2 q /\log_2 p)n}\gtrsim 2^{0.5n}.
\]
We choose the parameters accordingly to have $Q=2^{0.5 n + o(n)}$, which is our optimized complexity.

\begin{remark}[Going below $2^{0.5n}$ seems difficult]
    The above optimization reaches the complexity $2^{0.5n+o(n)}$
    independently of the requirement $N\ge 2^{0.5n}$ imposed by the
    known discrete Gaussian sampling algorithm \cite{STOC:ADRS15},
    which is also used in the preprocessing for BDD \cite{SICOMP:ACKS25}.
    Thus two separate bottlenecks lead to the same exponent: balancing
    the coset size against the estimation errors, and generating the
    required discrete Gaussian samples.
    \label{rem:belowhalf}
\end{remark}

\begin{remark}[The tree algorithm may not improve the midpoint Hessian]
    See the first bound in \cref{eqn:intro-estimate}. An analogous bound
    holds for the midpoint Hessian of \cite{Hhan26}, where $q=2$ is
    fixed. Hence increasing $p$ rapidly decreases the norm of the
    desired term, whereas the error terms do not decrease at a
    comparable rate, eventually making the desired term too small to
    detect. The choice $p=2$ is already used in the essential
    single-coset optimization.
    \label{rem:midpointmaynot}
\end{remark}

\paragraph{High-level idea of the analysis.}
We briefly describe the main idea behind the most technical part of
the analysis.
The primal expansion of the ideal quantity
$A_{\boldsymbol\delta}(\theta_*)$ in
\cref{eqn:overview-general-tree-estimator} contains the desired vector
$G(w_v)f(w_v)^{p-1}$.
To bound the total squared contribution of the remaining terms, we
expand the primal representation.  
This introduces vectors $
u_{r,0},u_{r,1}\in Bz_r+q\cL$
and Gaussian exponents of the form
\[
\sum_{r=0}^{p-1}
\left(
\norm{u_{r,0}}^2+\norm{u_{r,1}}^2
\right).
\]
Here $z_r$ are coupled by the constraints of the tree, making a direct application of discrete Gaussian mass bounds difficult.
We instead repeatedly apply, following the same binary tree, the
average--difference transformation
\[
(x_L,x_R)
\longmapsto
(a,d)
:=
\left(
\frac{x_L+x_R}{2},
x_L-x_R
\right).
\]
Writing $d_{h,r}$ for the difference created at node $r$ of level $h$,
the resulting full transformation is
\[
\Psi:\quad
\left(u_{r,b}\right)_
{\substack{0\le r<p\\ b\in\{0,1\}}}
\longmapsto
\left(
a_{m-1,0},
\left(d_{h,r}\right)_
{\substack{0\le h\le m-1\\0\le r<p/2^h}}
\right),
\]
where $a_{m-1,0}$ is the final average at the root.
At each node, the average is passed to its parent and the difference
is retained as a new summation variable. The identity
$\norm{x_L}^2+\norm{x_R}^2
=
2\norm{a}^2+\frac12\norm{d}^2$ allows us to rewrite the original exponent as a (weighted) sum of the squared norms of $a$ and $d$'s.

In this representation, the gradient-tree constraints determine the affine lattice coset containing each variable.
Together with the injectivity of $\Psi$, this coset structure allows us to bound the sum by a product of discrete-Gaussian masses, yielding the desired upper
bound.
In the full proof, we apply this bound recursively through the levels of the tree to control the contribution of all terms other than $G(w_v)f(w_v)^{p-1}$; see \cref{fig:svp-mass-analysis-flow}.

\paragraph{The $\CVP$ extension.}
Now we move to the CVP problem with the distance guarantee. 
Let $y\in\R^n$ be a target, let $v\in\cL$ be a closest lattice vector
to $y$, and write $e:=v-y$ and $w_v:=B^{-1}v\bmod q$.
Assume that $\norm e \le \alpha \lambda_1(\cL)$ for $\alpha>0$ to be optimized.
We define the shifted $q$-ary periodic Gaussian and its gradient by
\begin{align*}
f_y(w)
&:=
F_{1/s}\left(\frac{Bw-y}{q}\right)
=
\E_{X\sim D_{\cL^*,s}}
\left[
\exp\left(-\frac{2\pi i}{q}\ip{X}{y}\right)
\omega_q^{\ip{[X]_q}{w}}
\right],
\notag\\
G_y(w)
&:=
\left.\nabla_zF_{1/s}(z)\right|_{z=(Bw-y)/q}
=
2\pi i\,
\E_{X\sim D_{\cL^*,s}}
\left[
X\exp\left(-\frac{2\pi i}{q}\ip{X}{y}\right)
\omega_q^{\ip{[X]_q}{w}}
\right].
\end{align*}

Since $Bw_v=v\bmod q\cL$, the periodicity of $F_{1/s}$ gives
\[
G_y(w_v)
=
\nabla F_{1/s}\left(\frac{e}{q}\right)
\approx
-\frac{2\pi s^2}{q}
\frac{\rho_{1/s}(e/q)}{\rho_{1/s}(\cL)}
e.
\]
Thus $-G_y(w_v)$ points toward $e=v-y$.

We apply the same coset gradient tree construction to $f_y$ and $G_y$,
obtaining an estimate analogous to
\cref{eqn:overview-general-tree-empirical}.
Given $\norm e\le\alpha\lambda_1(\cL)$,\footnote{Technically, we also require $\norm e$ not to be too small.
Otherwise, the closest vector is found directly by the BDD algorithm.}
the desired term satisfies
\[
\norm{G_y(w_v)f_y(w_v)^{p-1}}
\ge
\frac{s^2\norm e}{q}
2^{-(t\alpha^2+o(1))n}.
\]
The difference between the coset gradient tree estimate and
$cG_y(w_v)f_y(w_v)^{p-1}$ for some $c\in\C$ with $|c|=1$ is bounded at most about $Q^{-1/2}=2^{-n/4}$ up to small factors as above.
Therefore, the desired term is detectable
when
\[
t\alpha^2<\frac{1}{4}.
\]
The lattice-packing bound used in the analysis requires
$t>t_0=0.23147\cdots$. 
For $\alpha<1/2\sqrt{t_0}$, the whole procedure have time and
space complexity $2^{n/2+o(n)}$ as in the SVP case, with the distance guarantee
$\dist(y,\cL)\le1.039\cdots\lambda_1(\cL)$.

%% file: table.tex
\begin{center}
\begingroup
\scriptsize
\resizebox{\linewidth}{!}{%
\begin{tabular}{c|ccccccccccc}
\hline
Reference
&
\shortstack{Kan87/HS07}
&
AKS01
&
PS09
&
MV10
&
ADRS15
&
CCL18
&
ACKS25
&
\shortstack{GFH26a/Kim26}
&
Hhan26
&
GFH26b
&
\textbf{This work}
\\
\hline
\shortstack{Time}
&
$n^{n/2e}$
&
$2^{O(n)}$
&
$2^{2.465n}$
&
$2^{2n}$
&
$2^n$
&
$2^{2.05n}$
&
$2^{1.669n}$
&
$2^{0.7314n}$
&
$2^{0.6039n}$
&
$2^{0.5596n}$
&
$\mathbf{2^{n/2}}$
\\
\shortstack{Space}
&
$\mathrm{poly}(n)$
&
$2^{O(n)}$
&
$2^{1.233n}$
&
$2^n$
&
$2^n$
&
$2^{n/2}$
&
$2^{n/2}$
&
$2^{n/2}$
&
$2^{n/2}$
&
$2^{n/2}$
&
$\mathbf{2^{n/2}}$
\\
\hline
\end{tabular}%
}

{\scriptsize
Classical algorithms for the shortest vector problem. $o(n)$ factors in exponents are omitted.
}
\endgroup
\end{center}

%% file: fig_tree.tex
\begin{figure}[h]
\centering
\begin{tikzpicture}[
    x=1cm,
    y=1cm,
    value/.style={
        draw,
        rounded corners,
        inner sep=3pt,
        font=\small,
        align=center
    },
    edge/.style={->,>=latex}
]

\node[value] (y0) at (0,0) {$Y_0$};
\node[value] (y1) at (2,0) {$Y_1$};
\node[value] (y2) at (4,0) {$Y_2$};
\node[value] (y3) at (6,0) {$Y_3$};

\node[value] (d10) at (1,1.2)
{$D_{1,0}=Y_0-Y_1$\\[-1pt]
{\color{blue!70!black}\scriptsize
$(D_{1,0})_{H_1}=\delta_{1,0}$}};

\node[value] (d11) at (5,1.2)
{$D_{1,1}=Y_2-Y_3$\\[-1pt]
{\color{blue!70!black}\scriptsize
$(D_{1,1})_{H_1}=\delta_{1,1}$}};

\node[value] (d20) at (3,2.4)
{$D_{2,0}=D_{1,0}-D_{1,1}$\\[-1pt]
{\color{blue!70!black}\scriptsize
$(D_{2,0})_{H_2}=\delta_{2,0}$}};

\draw[edge] (y0) -- (d10);
\draw[edge] (y1) -- (d10);
\draw[edge] (y2) -- (d11);
\draw[edge] (y3) -- (d11);

\draw[edge] (d10) -- (d20);
\draw[edge] (d11) -- (d20);

\end{tikzpicture}
\caption{The difference tree for $G(w)f(w)^3$.
}
\label{fig:svp-difference-tree}
\end{figure}
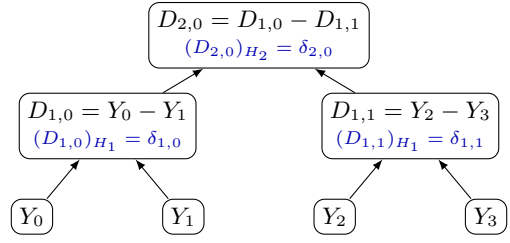

%% file: 2.prel.tex
\section{Preliminaries}
\label{sec:preliminaries}

We write $\Z$, $\N$, $\R$, and $\C$ to denote the sets of integers, positive integers, real numbers, and complex numbers.
$q$ will always be a prime number.
Let $\omega_q=\exp(2\pi i/q)$ be the $q$-th root of unity.
$\F_q$ denotes the field of $q$ elements $\{0,1,...,q-1\}$.
We write $[m]=\{1,\ldots,m\}$ for $m\in \N$ and
$[0,m):=\{0,\ldots,m-1\}$.
For integers $a\le b$, we write $[a,b):=\{a,a+1,\ldots,b-1\}$.
The Euclidean norm is denoted by $\norm{\cdot}$.
We write $\dist(t,S):=\inf_{x\in S}\norm{t-x}$ for a set $S$ and a vector $t$.
The indicator function for an event $E$ is denoted by $\ind_E$.
All logarithms are natural unless a base is specified.

All polynomial factors in this paper are polynomial in the dimension parameter
$n$ and the bit length of the inputs, e.g., the input basis; the latter is assumed to be polynomial in $n$ and is omitted in the paper throughout.
We suppress routine finite-precision issues and describe arithmetic over
$\R$ and $\C$ as exact; polynomially many bits of precision suffice throughout,
with only polynomial overhead and errors absorbed into the stated bounds.
The dimension $n$ tends to infinity in the asymptotic equations.
The constants in the $O$ and $\Omega$ notations, as well as $o(1)$ and $o(n)$ terms, are uniform over parameters unless specified otherwise.

\subsection{Lattice and discrete Gaussian}
Let $B\in\R^{n\times n}$ be a nonsingular matrix and let
$\cL=B\Z^n$ be a full-rank lattice generated by $B$.  
The first minimum is defined by $\lambda_1(\cL)=\min_{0\ne y\in\cL}\norm y$, which is occasionally written as $\lambda$.
The dual lattice of $\cL$ is defined by $ \cL^*=B^{-T}\mathbb Z^n$.
Let $\beta=2^{0.4014\ldots}$ denote the lattice-point constant obtained from
the spherical-code bound of \cite{KL78} and
\cite[Lemma~3]{EPRINT:PujSte09}.
These bounds are slightly improved in the recent work \cite{OpenAI2026}, but we need not use that improved bound.

For $\sigma>0$ and a countable set $A$, write
$\rho_\sigma(x)=e^{-\pi\norm{x}^2/\sigma^2}$ and
$\rho_\sigma(A)=\sum_{x\in A}\rho_\sigma(x)$.  
The centered discrete Gaussian and the
smoothing parameter $\eta_\varepsilon(\cL)$ for $0<\varepsilon<1$ are defined by
\[
D_{\cL,\sigma}(x)=\rho_\sigma(x)/\rho_\sigma(\cL),\qquad
 \rho_{1/\eta_\varepsilon(\cL)}(\cL^*\setminus\{0\})=\varepsilon.
\]

For a full-rank lattice $M$, $\sigma>0$, and $z\in\R^n$,
Poisson summation formula gives
\[
 \sum_{x\in M}e^{-\pi\norm{x}^2/\sigma^2}e^{2\pi i\langle x,z\rangle}
 =\frac{\sigma^n}{\det M}
   \sum_{y\in M^*}e^{-\pi \sigma^2\norm{y-z}^2} \qquad \Longrightarrow \qquad
   \rho_\sigma(M)=\frac{\sigma^n}{\det M}\rho_{1/\sigma}(M^*)
   \]
This formula and the triangle inequality show that a centered
coset has maximal Gaussian mass.
\begin{lemma}
\label{lem:shifted-mass-maximum}
For every full-rank lattice $M$, shift $t$, and $\sigma>0$,
$ \rho_\sigma(M+t)\le\rho_\sigma(M).$
\end{lemma}

The following theorem shows that the discrete Gaussian sampling can be efficiently done slightly above the smoothing parameter.

\begin{theorem}[{\cite[Theorem~5.11]{STOC:ADRS15}}]\label{thm:dgs-sampling}
Let $\cL\subset\R^n$ be a full-rank lattice, $\sigma>0$, and
$\kappa=\Omega(n)$. For every requested $1\le M\le2^{n/2}$, there is a
classical algorithm $\DGS$ that outputs $M$ lattice vectors in time
$2^{n/2+\polylog(\kappa)+o(n)}$ and space $2^{n/2+o(n)}$.
If $\sigma>\sqrt2\eta_{1/2}(\cL)$,
its output distribution is $\exp(-\Omega(\kappa))$-close to the distribution of
$M$ independent samples
$X_1,\ldots,X_M$ from $D_{\cL,\sigma}$.
\end{theorem}

We consider the following lattice problems. We use the search versions by default.
\begin{definition}
The $\SVP$ problem asks to find a nonzero vector $v\in\cL$ with
$\norm{v}=\lambda_1(\cL)$ given a basis of a full-rank lattice $\cL$.
\end{definition}

\begin{definition}
The $\CVP$ problem on input a basis of $\cL$ and a target vector $y$ asks to find a vector $v \in \cL$ such that $\norm{y-v}=\dist(y,\cL)$.
\end{definition}

\begin{definition}
For $0<\alpha<1/2$, the $\alpha$-$\BDD$ problem on input a basis of $\cL$ and a target vector $y$ asks to find the unique closest lattice vector to $y$, given the promise $\dist(y,\cL)<\alpha\lambda_1(\cL)$.
\end{definition}

We use a time-truncated version \cite{Hhan26} of 
the preprocessing BDD algorithm from \cite{SICOMP:ACKS25}.
\begin{theorem}
\label{thm:bdd-preprocessing}
There is a randomized classical preprocessing algorithm for the
$n^{-1/3}$-$\BDD$ problem such that:
\begin{itemize}[nosep]
    \item its preprocessing runs in worst-case time and space
    $2^{n/2+o(n)}$ and succeeds with constant probability. The
    resulting advice has size $2^{o(n)}$;
    \item given a successful preprocessing output, any instance of
    the $n^{-1/3}$-$\BDD$ problem can be deterministically solved in
    time $2^{o(n)}$ and polynomial additional space.
\end{itemize}
\end{theorem}

\subsection{Useful lemmas}
\begin{lemma}
\label{lem:normalization-stability}
Let $x,y\in\R^n$ with $y\ne0$, and let $0\le\varepsilon<1$. If
$\norm{x-y}\le\varepsilon\norm y,$
then $x\ne0$ and
\[
\norm{
\frac{x}{\norm x}-\frac{y}{\norm y}
}
\le
\frac{2\varepsilon}{1-\varepsilon}.\]
\end{lemma}

\begin{proof}
The triangle inequality gives
$\norm x\ge\norm y-\norm{x-y}\ge(1-\varepsilon)\norm y>0.$
Then we have
\[
\norm{
\frac{x}{\norm x}-\frac{y}{\norm y}
}
=
\norm{
\frac{x-y}{\norm x} + y\left(\frac{1}{\norm x}-\frac{1}{\norm y}\right)
}
\le
\frac{\norm{x-y}}{\norm x}
+\norm y\left|\frac1{\norm x}-\frac1{\norm y}\right|
\le
\frac{2\norm{x-y}}{\norm x}
\le
\frac{2\varepsilon}{1-\varepsilon}.\qedhere
\]
\end{proof}

%% file: 3.derivatives.tex
\section{Preparation}
We use the following parameters in the remainder of this paper.
Let $\cL$ be a full-rank $n$-dimensional lattice and let $\lambda=\lambda_1(\cL)$.
For $t_0:=\beta^2/(4e\ln 2)=0.23147\ldots$, fix $t_0<t<1/4$. Let $d>0$. We define
\[    g=\frac12\log_2(t/t_0)>0, \qquad s^2 = \frac{4nt\ln 2}{\pi d^2}.
\]
In most of our analysis, we assume that $\lambda \le d \le (1+1/n) \lambda.$ 
In the algorithm, $d$ is guessed from a polynomial-size list and the
main subroutine is run for every value in the list. Only a value
satisfying $\lambda\le d\le(1+1/n)\lambda$ is used in the correctness
analysis; outputs from the other guesses do not affect the
time complexity and correctness.
\begin{lemma}[{\cite[Corollary 3.2]{Hhan26}}]
    \label{cor:working-scale}
Let $\lambda=\lambda_1(\cL)$, let
$\lambda\le d\le(1+1/n)\lambda$, and let $t_0<t<1/4$.
Then,
\[
 \rho_{1/s}(\cL)-1\le2^{-(1/2+g-o(1))n},\qquad
 \rho_{\sqrt2/s}(\cL)-1\le2^{-(g-o(1))n},\qquad
 s>\sqrt2 \eta_{1/2}(\cL^*).
\]
\end{lemma}

\begin{lemma}[{\cite[Lemma~2.8]{Hhan26}}]
\label{lem:dgs-tail}
For every $c>0$, there is $C_c>0$ such that, for every lattice
$\cL$, $\sigma>0$,
\[
  \Prob_{X\sim D_{\cL,\sigma}}[\norm{X}>C_c \sigma\sqrt n]\le 2^{-2cn}\qquad
\text{and}\qquad
  \E_{X\sim D_{\cL,\sigma}}\left[
     \norm{X}^2\ind_{\{\norm{X}>C_c \sigma\sqrt n\}}
  \right]
  \le \sigma^2n\,2^{-cn}.
\]
\end{lemma}

The following lemma is a slight extension of \cite[Lemma~3.1]{Hhan26}. The proof is unchanged.
\begin{lemma}
\label{lem:shell}
For every fixed integer $k\ge0$, uniformly
over $\cL$ and $\sigma>0$ satisfying
$n\sigma^2/\lambda^2\le n^{O(1)}$,
\[
 \sum_{x\in\cL\setminus\{0\}}
 \left(1+\frac{\norm{x}}{\lambda_1(\cL)}\right)^k
 \rho_\sigma(x)
 \le
 2^{o(n)}
 \left(\frac{\beta^2n\sigma^2}{2\pi e\lambda_1^2(\cL)}\right)^{n/2}.\]
\end{lemma}

We define the periodic Gaussian \cite{AR05,CCC:DRS14} for $\cL$ with the width $\sigma>0$ by
\begin{align}
F_\sigma(z):=\frac{\rho_\sigma(\cL+z)}{\rho_\sigma(\cL)}
=\E_{X\sim D_{\cL^*,1/\sigma}}
\left[e^{2\pi i\ip{X}{z}}\right]
\label{eqn:periodic-Gaussian}
\end{align}
where the last equality holds because of the Poisson summation formula.
Differentiating two representations, we have the following representations of its gradient:
\begin{align}
\nabla F_\sigma(z)
=-\frac{2\pi}{\sigma^2\rho_\sigma(\cL)}
\sum_{x\in\cL}(x+z)\rho_\sigma(x+z)
=2\pi i\E_{X\sim D_{\cL^*,1/\sigma}}\left[Xe^{2\pi i\ip{X}{z}}\right].
\label{eqn:periodic-Gaussian-gradient}
\end{align}
Recall $q$ is a prime number.
Let $B$ be the basis of $\cL$.
For $X\in\cL^*$, we use the following notation excessively:
\[
[X]_q := B^TX\bmod q\in\F_q^n.
\]
This is a residue representation of the coset $X+q\cL^*\in\cL^*/q\cL^*\cong(\Z/q\Z)^n=\F_q^n$.
Recall $s^2=\frac{4nt\ln 2}{\pi d^2}$ for $t_0<t<1/4$.
We define the $q$-ary periodic Gaussian for $w\in\F_q^n$ and its gradient by
\begin{align}
f(w) := F_{1/s}\left(\frac{Bw}{q}\right) = \E_{X\sim D_{\cL^*,s}}
\left[
\omega_q^{\ip{[X]_q}{w}}
\right],\quad
G(w):=
\nabla F_{1/s}\left(\frac{Bw}{q}\right)=2\pi i\E_{X\sim D_{\cL^*,s}}
\left[
X\omega_q^{\ip{[X]_q}{w}}
\right]
\label{eqn:qary-affine-transforms}
\end{align}
using the second representations of \cref{eqn:periodic-Gaussian,eqn:periodic-Gaussian-gradient}.
For $v\in\cL$ and $w=B^{-1}v\bmod q$, periodicity and the first representations in
\cref{eqn:periodic-Gaussian,eqn:periodic-Gaussian-gradient} give
\begin{align}
f(w)=\frac{1}{\rho_{1/s}(\cL)}\sum_{x\in\cL}\rho_{1/s}\left(x+\frac{v}{q}\right),\qquad
G(w)=-\frac{2\pi s^2}{\rho_{1/s}(\cL)}\sum_{x\in\cL}\left(x+\frac{v}{q}\right)\rho_{1/s}\left(x+\frac{v}{q}\right).
\label{eqn:odd-local-primal}
\end{align}
We call these representation of $f$ and $G$ as \emph{primal}, and the second representation in \cref{eqn:qary-affine-transforms} as \emph{dual}.
It is not hard to see that $f$ is even and $G$ is odd from the first representations of \cref{eqn:periodic-Gaussian,eqn:periodic-Gaussian-gradient}.

The following lemma shows that the gradient of the periodic Gaussian at $v/q$ has direction close to $v$.
\begin{lemma}
\label{lem:odd-local}
Let $t_0<t<1/4$ and $\lambda=\lambda_1(\cL)\le d\le(1+1/n)\lambda.$
Let $q$ be an odd prime.
Let $v\in\cL$ be a shortest vector and let
$w=B^{-1}v\bmod q$.
It holds that

\[
f(w)=
2^{-4tn\lambda^2/(d^2q^2)}
\left(1+O(2^{-\Omega(n)})\right),\qquad
\norm{
-\frac{G(w)}{f(w)}-\frac {2\pi s^2 v}q
}\le
2^{-\Omega(n)}s^2 \lambda.\]
\end{lemma}
\begin{proof}
Recall $s^2=\frac{4nt\ln 2}{\pi d^2}$.
For $x\in\cL\setminus\{0,-v\}$, it holds that $\norm{x+v}\ge\lambda$, which gives $2\ip{x}{v}\ge-\norm{x}^2$ by squaring both sides. Therefore,
\begin{align}
\norm{x+v/q}^2-\norm{v/q}^2
=\norm{x}^2+\frac{2}{q}\ip{x}{v}
\ge\left(1-\frac1q\right)\norm{x}^2.
\label{eqn:odd-local-gap}
\end{align}
By the definition of $\rho_{1/s}$ and \cref{eqn:odd-local-gap}, the following holds for $x\in\cL\setminus\{0,-v\}$
\begin{align}
\frac{\rho_{1/s}(x+v/q)}{\rho_{1/s}(v/q)}
=e^{-\pi s^2(\norm{x+v/q}^2-\norm{v/q}^2)}
\le e^{-\pi s^2(1-1/q)\norm{x}^2}
=\rho_{1/(s\sqrt{1-1/q})}(x)
\le
\rho_{\sqrt2/s}(x),
\label{eqn:odd-local-rho-comparison}
\end{align}
where the last inequality follows from $q\ge3$ and the fact that $\rho_\sigma(x)$ is increasing in $\sigma$.
\cref{lem:shell} with $\sigma=\sqrt2/s$ and $k=1$ gives, using $s^2=\frac{4nt\ln 2}{\pi d^2}$ and $g=\frac12\log_2(t/t_0)$,
\begin{align}
\sum_{x\in\cL\setminus\{0\}}
\left(1+\frac{\norm{x}}{\lambda}\right)
\rho_{\sqrt2/s}(x)
&\le
2^{o(n)}
\left(
\frac{\beta^2n}{\pi e\lambda^2s^2}
\right)^{n/2}
=
2^{o(n)}
\left(
\frac{t_0d^2}{t\lambda^2}
\right)^{n/2}
\le
2^{-(g-o(1))n}.
\label{eqn:odd-local-shell}
\end{align}
For $x=-v$, we have
\[
\frac{\rho_{1/s}(-v+v/q)}{\rho_{1/s}(v/q)}
=e^{-\pi s^2(1-2/q)\lambda^2}
\le e^{-\pi s^2\lambda^2/3}
=2^{-\Omega(n)}\]
using $q\ge 3.$
Combining this bound with \cref{eqn:odd-local-rho-comparison,eqn:odd-local-shell} gives
\begin{align}
\sum_{x\in\cL\setminus\{0\}}
\left(1+\frac{\norm{x}}{\lambda}\right)
{\rho_{1/s}(x+v/q)}
\le 2^{-\Omega(n)} \rho_{1/s}(v/q).
\label{eqn:odd-local-tail}
\end{align}
By \cref{cor:working-scale}, $\rho_{1/s}(\cL)=1+2^{-\Omega(n)}$. It follows from \cref{eqn:odd-local-primal,eqn:odd-local-tail} that
\[
f(w)
=\rho_{1/s}(v/q)\left(1+2^{-\Omega(n)}\right)
=2^{-4tn\lambda^2/(d^2q^2)}\left(1+2^{-\Omega(n)}\right)
.
\]
Moreover, $\norm{x+v/q}\le\norm{x}+\lambda=\lambda(1+\norm{x}/\lambda)$ for every nonzero $x\in\cL$.
Hence, the gradient representation in \cref{eqn:odd-local-primal} and the triangle inequality, together with \cref{eqn:odd-local-tail}, give
\[
\norm{G(w)+\frac{2\pi s^2}{\rho_{1/s}(\cL)}\frac{v}{q}\rho_{1/s}(v/q)}
=\frac{2\pi s^2}{\rho_{1/s}(\cL)} \norm{\sum_{x\in\cL\setminus\{0\}}\left(x+\frac{v}{q}\right)\rho_{1/s}\left(x+\frac{v}{q}\right)}
\le
\frac{2\pi s^2\rho_{1/s}(v/q)}{\rho_{1/s}(\cL)}
2^{-\Omega(n)}\lambda.\]
Finally, dividing the last inequality by $f(w)=\frac{\rho_{1/s}(v/q)}{\rho_{1/s}(\cL)}\left(1+2^{-\Omega(n)}\right)$ gives the second bound.
\end{proof}

%% file: 4.svp.tex
\section{The SVP algorithm}
\label{sec:svp}

We prove the following result.
\begin{theorem}
\label{thm:main-svp}
There is a randomized classical algorithm that solves $\SVP$ with probability at least $2/3$ in worst-case time and space $2^{n/2+o(n)}$.
\end{theorem}

We will use the following parameters in the algorithm.
Let $p$ be the largest power of two satisfying $p\le n/\log^3n$, let $q$ be the smallest prime larger than $2\sqrt p$, and let $m=\log_2p+1$. Then,
\begin{align}
p=\Theta(n/\log^3n),\qquad q=(2+o(1))\sqrt p,\qquad \log_2q=\frac{m+1}{2}+o(1)
    \label{eqn:svp-parameter}
\end{align}
where the asymptotic size of $q$ is due to the prime number theorem.\footnote{For every fixed $\eta>0$, the prime number theorem gives
$\pi((1+\eta)2\sqrt{p})-\pi(2\sqrt{p})
=(\eta+o(1))\frac{2\sqrt{p}}{\log (2\sqrt{p})}\ge 1$ for all sufficiently large $p$.}

\subsection{The first estimation target}
We first explain that it is enough to approximate $G(w)f(w)^{p-1}$ where $f(w)^{p-1}$ is a scalar, but will be essential in the algorithm. 
The estimation we will see loses the unit factor $c$, which should be taken into account.
\begin{lemma}
\label{lem:svp-recovery}
Let $\lambda=\lambda_1(\cL)\le d\le(1+1/n)\lambda$ and  $t_0<t<1/4$, let $v\in\cL$ be a shortest vector, and let $w=B^{-1}v\bmod q$. Suppose that $A\in\C^n$ satisfies for some $|c|=1$ that
\begin{align}
\norm{A-cG(w)f(w)^{p-1}}\le\frac{s^2\lambda}{q}2^{-n/4+o(n)},
\label{eqn:svp-required-accuracy}
\end{align}
then one of the well-defined\footnote{That is, one with $\norm{\Re A} \neq 0$ or $\norm{\Im A}\neq 0$.} nonzero vectors $\Re A/\norm{\Re A}$ and $\Im A/\norm{\Im A}$ is within the Euclidean distance $2^{-\Omega(n)}$ of either $v/\lambda$ or $-v/\lambda$. For the corresponding unit vector $u$ (i.e., one of $\Re A/\norm{\Re A}$ or $\Im A/\norm{\Im A}$), the following holds for all sufficiently large $n$:\[
\norm{\pm du-v}<2\lambda/n.\]
\end{lemma}

\begin{proof}
By \cref{lem:odd-local} and the choice of $q$ that $q^2=(4+o(1))p$, we have
\begin{align}
f(w)^p=2^{-(t-o(1))n},\qquad
-G(w)f(w)^{p-1}=\frac{2\pi s^2}{q}f(w)^p(v+e)
\label{eqn:svp-direction-signal}
\end{align}
for some $e\in\R^n$ satisfying $\norm e\le2^{-\Omega(n)}\lambda$.
By \cref{eqn:svp-required-accuracy,eqn:svp-direction-signal} and the triangle inequality,
\[\norm{A+\frac{2\pi s^2}{q}cf(w)^pv}
\le
\norm{A-cG(w)f(w)^{p-1}}
+\frac{2\pi s^2}{q}f(w)^p\norm e \le
\frac{s^2\lambda}{q}f(w)^p2^{-\Omega(n)}.
\]
where we use $t<1/4$ in \cref{eqn:svp-required-accuracy} to bound the first term.
Taking the real and imaginary parts gives
\[
\norm{\Re A+\frac{2\pi s^2}{q}(\Re c)f(w)^pv},
\norm{\Im A+\frac{2\pi s^2}{q}(\Im c)f(w)^pv}
\le
\frac{s^2\lambda}{q}f(w)^p2^{-\Omega(n)}.\]
At least one of $|\Re c|$ and $|\Im c|$ is at least $1/\sqrt2$. If $|\Re c|\ge1/\sqrt2$,
\[
\norm{\frac{2\pi s^2}{q}(\Re c)f(w)^pv}
\ge
\frac{\sqrt2\pi s^2\lambda}{q}f(w)^p \gg \frac{s^2\lambda}{q}f(w)^p2^{-\Omega(n)}.\]
Thanks to \cref{lem:normalization-stability},
$\Re A$ is nonzero and
\[\norm{
\frac{\Re A}{\norm{\Re A}}
+\operatorname{sgn}(\Re c)\frac{v}{\lambda}
}
\le 2^{-\Omega(n)}.
\]
An analogous conclusion holds for $\Im A$ if $|\Im c|\ge1/\sqrt2$.
The final conclusion holds because for large $n$
\[\norm{\pm du-v}
\le |d-\lambda|+d\norm{\pm u-v/\lambda}
\le\frac\lambda n+2^{-\Omega(n)}\lambda
<2\lambda/n.\qedhere
\]
\end{proof}

\subsection{The coordinate decomposition and the estimator}
\paragraph{Coordinate decomposition.}
We first present a formula for $G(w)f(w)^{p-1}$ from independent discrete Gaussian samples $X_0,...,X_{p-1} \sim D_{\cL^*,s}$. We begin with the observation that for any $\varepsilon_r \in \{-1,1\}$ with $\varepsilon_0=1$,
\begin{align}
\label{eqn: Gfexpectation}
\E\left[
2\pi i X_0 \omega_q^{\ip{w}{\sum_{r=0}^{p-1} \varepsilon_r [X_r]_q}}
\right]
= \E\left[2\pi iX_0 \omega_q^{\ip{w}{[X_0]_q}}\right] \prod_{r=1}^{p-1} \E\left[\omega_q^{\ip{\varepsilon_rw}{[X_r]_q}}\right] 
= G(w) f(w)^{p-1}
\end{align}
where we use $\E\left[\omega_q^{\ip{\varepsilon_rw}{[X_r]_q}}\right]  = f(\varepsilon_r w) =f(w)$ because $f$ is an even function. 
While this formula works for any $\epsilon$'s, we henceforth let
\[
\varepsilon_r:=(-1)^{\operatorname{wt}_2(r)}
\qquad\text{for every }r\in[0,p),
\]
where $\operatorname{wt}_2(r)$ is the Hamming weight of the binary
representation of $r$.

For simplicity, we assume that $m=\log_2 p+1$ divides $n$ throughout in this paper.
This assumption does not affect every asymptotic in this paper much, that is, all the changes are at most $2^{o(n)}$.\footnote{For the parameters below, we may choose $Q:=|H_1|=\cdots = |H_{m-1}|=q^\ell$ for $\ell:=\lfloor n/m\rfloor$ and $Q_m:=|H_m|=q^{n-(m-1)\ell}.$ In particular, $Q,Q_m=2^{n/2+o(n)}$ for our parameter setting.}
For $j\in[m]$, let
\[
H_j:=\{x=(x_1,\ldots,x_n)\in\F_q^n:x_k=0
\text{ for every }k\notin\{(j-1)n/m+1,\ldots,jn/m\}\}.
\]
This gives a decomposition
\[
    \F_q^n=H_1\oplus\cdots\oplus H_m, \qquad Y=\sum_{j=1}^m (Y)_{H_j}
    \quad \text{for}\quad Y\in\F_q^n
\]
where $(Y)_{H_j}\in H_j$ denotes the unique component satisfying
$Y=\sum_{j=1}^m (Y)_{H_j}$.

We first define the value that we will approximate. 
Let $U$ be uniformly random in $\GL_n(\F_q)$.
For the independent samples $X_0,\ldots,X_{p-1}$ in \cref{eqn: Gfexpectation}, let $Y_{0,r}:=U[X_r]_q=U(B^TX_r\bmod q)$ for $r\in[0,p)$ and define $Y_{j,a}:=Y_{j-1,2a}-Y_{j-1,2a+1}$ for $j\in[m-1]$ and $a\in[0,p/2^j)$. 
For $r\in[0,p)$, let $\varepsilon_r:=(-1)^{\operatorname{wt}_2(r)}$, where $\operatorname{wt}_2(r)$ is the Hamming weight of the binary representation of $r$.
It is not hard to see by induction that
\begin{align}
Y_{m-1,0}
=
\sum_{r=0}^{p-1}\varepsilon_rY_{0,r},\qquad 
Y_{j,a}
=
\sum_{b=0}^{2^j-1}
\varepsilon_bY_{0,a2^j+b}
=
\varepsilon_{a2^j}
\sum_{r=a2^j}^{(a+1)2^j-1}
\varepsilon_rY_{0,r}.
\label{eqn:svp-final-difference}
\end{align}
Note that the indices for $Y_{j,a}$ satisfy $a \in [0,p/2^{j})$.

We now give an alternative representation of \cref{eqn: Gfexpectation}
for $w=B^{-1}v \bmod q$ for a shortest vector $v$. 
Let
\begin{align}
k_*=U^{-T}w,\qquad 
\theta_*:=(k_*)_{H_m},\qquad
\text{and}\qquad
\alpha_{j,a}^*
:=
\varepsilon_{a2^j}(k_*)_{H_j} \quad \text{for}\quad
j\in[m-1],a\in[0,p/2^j)\label{eqn:svp-target-alpha}
\end{align}
and let
$\boldsymbol\alpha^*:=(\alpha_{j,a}^*)_{j,a}$.
By $Y_{0,r}=U[X_r]_q$ and $k_*=U^{-T}w$, we have
\[
\ip{w}{[X_r]_q}
=\ip{U^{-T}w}{U [X_r]_q}
=
\ip{k_*}{Y_{0,r}}
\pmod q.
\]
Thus the exponent $\ip{w}{\sum_{r=0}^{p-1} \varepsilon_r [X_r]_q}$ in \cref{eqn: Gfexpectation} is
$\sum_{r=0}^{p-1}
\varepsilon_r\ip{k_*}{Y_{0,r}}.$

We decompose this sum according to
$\F_q^n=H_1\oplus\cdots\oplus H_m$. 
From the second identity in \cref{eqn:svp-final-difference}, we have
\[
\sum_{a=0}^{p/2^j-1}
\ip{\alpha_{j,a}^*}{Y_{j,a}}
=
\sum_{r=0}^{p-1}
\varepsilon_r
\ip{(k_*)_{H_j}}{Y_{0,r}}
\qquad\text{for every }j\in[m-1].
\]
The remaining $H_m$-component satisfies $\ip{\theta_*}{(Y_{m-1,0})_{H_m}}
=
\sum_{r=0}^{p-1}
\varepsilon_r
\ip{(k_*)_{H_m}}{Y_{0,r}}$ since
$\theta_*=(k_*)_{H_m}$ and
$Y_{m-1,0}=\sum_{r=0}^{p-1}\varepsilon_rY_{0,r}$.
Combining these identities gives the alternative expression of the exponents in \cref{eqn: Gfexpectation}
\begin{align}
\omega_q^{
\sum_{r=0}^{p-1}
\varepsilon_r\ip{w}{[X_r]_q}
}
=\omega_q^{\sum_{r=0}^{p-1}
\varepsilon_r\ip{k_*}{Y_{0,r}}}
=
\omega_q^{\ip{\theta_*}{(Y_{m-1,0})_{H_m}}}
\prod_{j=1}^{m-1}
\prod_{a=0}^{p/2^j-1}
\omega_q^{\ip{\alpha_{j,a}^*}{Y_{j,a}}}.
\label{eqn:svp-desired-character-factorization}
\end{align}

\paragraph{What to estimate.}
We define an ideal quantity that contains
\cref{eqn: Gfexpectation} as a dominating term. We later
approximate all its values efficiently from finite sample lists.
Let $
Q:=|H_j|=q^{n/m}=2^{n/2+o(n)}.$
For every $j\in[m-1]$ and $a\in[0,p/2^j)$, let
$\delta_{j,a}$ be independently and uniformly distributed over $H_j$,
and let $\boldsymbol\delta=(\delta_{j,a})_{j,a}$. 
For every $j$ and $a$, character orthogonality gives
\begin{align}
Q\ind_{\{(Y_{j,a})_{H_j}=\delta_{j,a}\}}
=
\sum_{\alpha_{j,a}\in H_j}
\omega_q^{-\ip{\alpha_{j,a}}{\delta_{j,a}}}
\omega_q^{\ip{\alpha_{j,a}}{Y_{j,a}}}.
\label{eqn:svp-character-orthogonality}
\end{align}
Here we used
$\ip{\alpha_{j,a}}{(Y_{j,a})_{H_j}}
=\ip{\alpha_{j,a}}{Y_{j,a}}$ because
$\alpha_{j,a}\in H_j$. 
We define the conjunction of all constraints
\begin{align}
C_{\boldsymbol\delta}(X_0,\ldots,X_{p-1})
:=
\prod_{j=1}^{m-1}
\prod_{a=0}^{p/2^j-1}
\ind_{\{(Y_{j,a})_{H_j}=\delta_{j,a}\}}.
\label{eqn:svp-ideal-indicator}
\end{align}

Multiplying \cref{eqn:svp-character-orthogonality} over all pairs
$(j,a)$ gives the following, which includes \cref{eqn:svp-desired-character-factorization} as a term,
\begin{align}
Q^{p-1}
C_{\boldsymbol\delta}(X_0,\ldots,X_{p-1})
=
\sum_{\boldsymbol\alpha}
\left(
\prod_{j=1}^{m-1}
\prod_{a=0}^{p/2^j-1}
\omega_q^{-\ip{\alpha_{j,a}}{\delta_{j,a}}}
\right)
\prod_{j=1}^{m-1}
\prod_{a=0}^{p/2^j-1}
\omega_q^{\ip{\alpha_{j,a}}{Y_{j,a}}},
\label{eqn:svp-all-characters}
\end{align}
where the sum is over $\boldsymbol{\alpha} = (\alpha_{j,a})_{j,a}$ for all
$\alpha_{j,a}\in H_j$.
This motivates the following definition. For $\theta\in H_m$, let
\begin{align}
A_{\boldsymbol\delta}(\theta)
:=
Q^{p-1}
\E\left[
C_{\boldsymbol\delta}(X_0,\ldots,X_{p-1})
2\pi iX_0
\omega_q^{\ip{\theta}{(Y_{m-1,0})_{H_m}}}
\mid U,\boldsymbol\delta
\right].
\label{eqn:svp-ideal-value}
\end{align}
Here the expectation is over the independent samples
$X_0,\ldots,X_{p-1}\sim D_{\cL^*,s}$, while
$U$ and $\boldsymbol\delta$ are fixed.
Note that for $A_{\boldsymbol\delta}$ with the fixed $U,\boldsymbol{\delta}$, the search space is of size $|H_m|=Q=2^{n/2+o(n)}$. 

We later prove the following lemma, stating that one of $A_{\boldsymbol\delta}(\theta)$ gives the direction of a shortest vector $v$.

\begin{lemma}
\label{lem:svp-ideal-identity}
Let $v\in\cL$ be a shortest vector, let $\lambda=\norm v$, let
$w=B^{-1}v\bmod q$, let $k_*=U^{-T}w$, and let
$\theta_*=(k_*)_{H_m}$. Except with probability $o(1)$ over
$U$ and $\boldsymbol\delta$, there are
$c_{\boldsymbol\delta}\in\C$ and
$e_{\boldsymbol\delta}\in\C^n$ satisfying
\[
A_{\boldsymbol\delta}(\theta_*)
=
c_{\boldsymbol\delta}G(w)f(w)^{p-1}
+
e_{\boldsymbol\delta}\quad
\text{where}\quad
|c_{\boldsymbol\delta}|=1\text{ and }
\norm{e_{\boldsymbol\delta}}
\le
\frac{s^2\lambda}{q}2^{-n/4-\Omega(n)}.
\]
\end{lemma}
\begin{proof}
Substituting \cref{eqn:svp-all-characters} into
\cref{eqn:svp-ideal-value} gives
\[
A_{\boldsymbol\delta}(\theta)
=
\sum_{\boldsymbol\alpha}
\left(
\prod_{j=1}^{m-1}
\prod_{a=0}^{p/2^j-1}
\omega_q^{-\ip{\alpha_{j,a}}{\delta_{j,a}}}
\right)
\E\left[
2\pi iX_0
\omega_q^{
\ip{\theta}{Y_{m-1,0}}
+
\sum_{j=1}^{m-1}\sum_{a=0}^{p/2^j-1}
\ip{\alpha_{j,a}}{Y_{j,a}}
}
\mid U
\right],
\]
where we used $\theta\in H_m$ and hence
$\ip{\theta}{(Y_{m-1,0})_{H_m}}=\ip{\theta}{Y_{m-1,0}}$.

For each $\boldsymbol\alpha$ and $\theta$, let
$k_{\boldsymbol\alpha,\theta,0},\ldots,k_{\boldsymbol\alpha,\theta,p-1}\in\F_q^n$ be the coefficients obtained by
expanding the recursive differences $Y_{j,a}$ in terms of
$Y_{0,0},\ldots,Y_{0,p-1}$, so that
\begin{align}
\ip{\theta}{Y_{m-1,0}}
+
\sum_{j=1}^{m-1}\sum_{a=0}^{p/2^j-1}
\ip{\alpha_{j,a}}{Y_{j,a}}
=
\sum_{r=0}^{p-1}\ip{k_{\boldsymbol\alpha,\theta,r}}{Y_{0,r}}.
\label{eqn:svp-leaf-frequency-identity}
\end{align}
Since $Y_{0,r}=U[X_r]_q$, independence of
$X_0,\ldots,X_{p-1}$, and the last terms in \cref{eqn:qary-affine-transforms} for $G$ and $f$ give
\begin{align}
A_{\boldsymbol\delta}(\theta)
=
\sum_{\boldsymbol\alpha}
\left(
\prod_{j=1}^{m-1}
\prod_{a=0}^{p/2^j-1}
\omega_q^{-\ip{\alpha_{j,a}}{\delta_{j,a}}}
\right)
G(U^Tk_{\boldsymbol\alpha,\theta,0})
\prod_{r=1}^{p-1}f(U^Tk_{\boldsymbol\alpha,\theta,r}).
\label{eqn:svp-ideal-expansion}
\end{align}

Let $k_*=U^{-T}w$, $\theta_*=(k_*)_{H_m}$, and $\alpha_{j,a}^*
:=
\varepsilon_{a2^j}(k_*)_{H_j}.$
For $\theta=\theta_*$ and
$\boldsymbol\alpha=\boldsymbol\alpha^*$,
\cref{eqn:svp-leaf-frequency-identity} gives
$k_{\boldsymbol\alpha^*,\theta_*,r}=\varepsilon_rk_*$ so that $U^Tk_{\boldsymbol\alpha^*,\theta_*,r} = \pm w$ for every $r\in[0,p)$. Therefore, the
corresponding term in \cref{eqn:svp-ideal-expansion} is
\[
c_{\boldsymbol\delta}G(w)f(w)^{p-1},
\qquad
c_{\boldsymbol\delta}
:=
\prod_{j=1}^{m-1}
\prod_{a=0}^{p/2^j-1}
\omega_q^{-\ip{\alpha_{j,a}^*}{\delta_{j,a}}},
\]
where $|c_{\boldsymbol\delta}|=1$. Defining
$e_{\boldsymbol\delta}$ as the sum of the remaining terms gives
\[
A_{\boldsymbol\delta}(\theta_*)
=
c_{\boldsymbol\delta}G(w)f(w)^{p-1}
+
e_{\boldsymbol\delta}.
\]
By \cref{lem:svp-offset-character-orthogonality} below with
$B_{\boldsymbol\alpha}
=
G(U^Tk_{\boldsymbol\alpha,\theta_*,0})
\prod_{r=1}^{p-1}f(U^Tk_{\boldsymbol\alpha,\theta_*,r}),$
we have
\[
\E_{\boldsymbol\delta}\left[
\norm{e_{\boldsymbol\delta}}^2
\mid U
\right]
=
\sum_{\boldsymbol\alpha\ne\boldsymbol\alpha^*}
\norm{
G(U^Tk_{\boldsymbol\alpha,\theta_*,0})
\prod_{r=1}^{p-1}f(U^Tk_{\boldsymbol\alpha,\theta_*,r})
}^2.
\]
Bounding this term is the hardest part.
We defer this part to \cref{lem:svp-nonprincipal-mass} below, which gives 
\[
\E_{\boldsymbol\delta}\left[
\norm{e_{\boldsymbol\delta}}^2
\mid U
\right]
\le
\left(\frac{s^2\lambda}{q}\right)^2
2^{-n/2-\Omega(n)}.
\]
The claimed bound on $\norm{e_{\boldsymbol\delta}}$ now follows from
Markov's inequality.
\end{proof}

\begin{lemma}
\label{lem:svp-offset-character-orthogonality}
For
$\boldsymbol\alpha=(\alpha_{j,a})_{j,a}$ with
$\alpha_{j,a}\in H_j$, let $B_{\boldsymbol\alpha}\in\C^n$ and let
$\chi_{\boldsymbol\alpha}(\boldsymbol\delta)
:=
\prod_{j=1}^{m-1}
\prod_{a=0}^{p/2^j-1}
\omega_q^{-\ip{\alpha_{j,a}}{\delta_{j,a}}}.$
Then,
\begin{align}
\E_{\boldsymbol\delta}\left[
\norm{
\sum_{\boldsymbol\alpha}
\chi_{\boldsymbol\alpha}(\boldsymbol\delta)
B_{\boldsymbol\alpha}
}^2
\right]
=
\sum_{\boldsymbol\alpha}\norm{B_{\boldsymbol\alpha}}^2.
\label{eqn:svp-offset-character-orthogonality}
\end{align}
\end{lemma}

\begin{proof}
For
$\boldsymbol\alpha\ne\boldsymbol\alpha'$, there is some $(j,a)$ such
that $\alpha_{j,a}\ne\alpha'_{j,a}$. Independence and uniformity of
the vectors $\delta_{j,a}$ and character orthogonality give
\[
\E_{\boldsymbol\delta}\left[
\chi_{\boldsymbol\alpha}(\boldsymbol\delta)
\overline{\chi_{\boldsymbol\alpha'}(\boldsymbol\delta)}
\right]
=
\prod_{j,a}
\E_{\delta_{j,a}}\left[
\omega_q^{-\ip{\alpha_{j,a}-\alpha'_{j,a}}{\delta_{j,a}}}
\right]
=0.
\]
The same expectation is 1 for $\boldsymbol\alpha=\boldsymbol\alpha'$.
Expanding the squared norm and taking expectation
gives \cref{eqn:svp-offset-character-orthogonality}.
\end{proof}

We defer the proof of the following lemma to \cref{subsec:prooflem}.
\begin{lemma}
\label{lem:svp-nonprincipal-mass}
Let $\lambda=\lambda_1(\cL)\le d\le(1+1/n)\lambda$, let
$v\in\cL$ be a shortest vector, let
$w=B^{-1}v\bmod q$, let $k_*=U^{-T}w$, and let
$\theta_*=(k_*)_{H_m}$. Let $\boldsymbol\alpha^*$ be defined as in
\cref{eqn:svp-target-alpha}. For each $\boldsymbol\alpha$, let
$k_{\boldsymbol\alpha,\theta_*,0},\ldots,k_{\boldsymbol\alpha,\theta_*,p-1}$ be defined
by \cref{eqn:svp-leaf-frequency-identity} with $\theta=\theta_*$.
Except with probability $o(1)$
over $U$,
\begin{align}
\sum_{\boldsymbol\alpha\ne\boldsymbol\alpha^*}
\norm{
G(U^Tk_{\boldsymbol\alpha,\theta_*,0})
\prod_{r=1}^{p-1}f(U^Tk_{\boldsymbol\alpha,\theta_*,r})
}^2
\le
\left(\frac{s^2\lambda}{q}\right)^2
2^{-n/2-\Omega(n)}.
\label{eqn:svp-nonprincipal-mass}
\end{align}
\end{lemma}

\paragraph{Estimator.}
We now define an estimator for $A_{\boldsymbol\delta}$ and show that all its values
can be computed efficiently. Let $
N:=Qp^4=2^{n/2+o(n)}.$
Recall that $U$ is sampled uniformly from $\GL_n(\F_q)$.
We consider $N$ tuples of samples: 
For every $r\in[0,p)$ and $i\in[N]$, choose independent samples
\[
X_{r,i}\sim D_{\cL^*,s},
\qquad
Y_{r,i}:=U[X_{r,i}]_q.
\]

For 
$\boldsymbol i=(i_0,\ldots,i_{p-1})\in[N]^p$, 
we consider the following indexed versions of $Y$'s
\begin{align}
    \label{eqn: recursive Yi}
Y_{0,r}(\boldsymbol i):=Y_{r,i_r}
\text{~~for }r\in[0,p),
\text{~~and~~}
Y_{j,a}(\boldsymbol i)
:=
Y_{j-1,2a}(\boldsymbol i)
-
Y_{j-1,2a+1}(\boldsymbol i)\text{~~for }j\in[m-1],a\in[0,p/2^j).
\end{align}
Note that all vectors $Y_{j,a}(\boldsymbol i)$ are determined by
the samples $X_{r,i_r}$ for $r\in[0,p)$.
Let
\[
C_{\boldsymbol\delta}(\boldsymbol i)
:=
C_{\boldsymbol\delta}
(X_{0,i_0},\ldots,X_{p-1,i_{p-1}}).
\]
For $\theta\in H_m$, define the following estimator for \cref{eqn:svp-ideal-value}:
\begin{align}
\widehat A_{\boldsymbol\delta}(\theta)
:=
\frac{Q^{p-1}}{N^p}
\sum_{\boldsymbol i\in[N]^p}
C_{\boldsymbol\delta}(\boldsymbol i)
2\pi iX_{0,i_0}
\omega_q^{
\ip{\theta}
{(Y_{m-1,0}(\boldsymbol i))_{H_m}}
}.
\label{eqn:svp-estimator}
\end{align}

We prove the accuracy of the estimator for using \cref{lem:svp-recovery}. The first step is the following lemma, whose proof is deferred to \cref{subsec:prooflem} due to its complexity.
\begin{lemma}
\label{lem:svp-finite-list-second-moment}
Let $\lambda=\lambda_1(\cL)\le d\le(1+1/n)\lambda$.
Except with probability $2^{-\Omega(n)}$ over $U$,
simultaneously for every $\theta\in H_m$,
\[
\E_{\boldsymbol\delta,X}\left[
\norm{
\widehat A_{\boldsymbol\delta}(\theta)
-
A_{\boldsymbol\delta}(\theta)
}^2
\mid U
\right]
\le
\left(\frac{s^2\lambda}{q}\right)^2
Q^{-1}2^{o(n)}.
\]
\end{lemma}

\begin{lemma}
\label{lem:svp-estimator-accuracy}
Let $v\in\cL$ be a shortest vector, let $\lambda=\norm v$, let
$w=B^{-1}v\bmod q$, and let
$\theta_*=(U^{-T}w)_{H_m}$.
Except with probability $o(1)$ over $U$, $\boldsymbol\delta$, and the samples $X_{r,i}$,
there is $c_{\boldsymbol\delta}\in\C$ with $|c_{\boldsymbol\delta}|=1$ such that
\begin{align}
\norm{\widehat A_{\boldsymbol\delta}(\theta_*)-c_{\boldsymbol\delta}G(w)f(w)^{p-1}}
\le
\frac{s^2\lambda}{q}2^{-n/4+o(n)}.
\label{eqn:svp-estimator-accuracy}
\end{align}
\end{lemma}

\begin{proof}
By \cref{lem:svp-finite-list-second-moment} and
$Q=2^{n/2+o(n)}$, except with probability
$2^{-\Omega(n)}$ over $U$,
\[
\E_{\boldsymbol\delta,X}\left[
\norm{
\widehat A_{\boldsymbol\delta}(\theta_*)
-
A_{\boldsymbol\delta}(\theta_*)
}^2
\mid U
\right]
\le
\left(\frac{s^2\lambda}{q}\right)^2
2^{-n/2+o(n)}.
\]
Markov's inequality, after choosing the $o(n)$ term sufficiently
large, shows that except with probability $o(1)$ over
$\boldsymbol\delta$ and the samples,
\[
\norm{
\widehat A_{\boldsymbol\delta}(\theta_*)
-
A_{\boldsymbol\delta}(\theta_*)
}
\le
\frac{s^2\lambda}{q}2^{-n/4+o(n)}.
\]
By \cref{lem:svp-ideal-identity}, except with probability $o(1)$,
there is $c_{\boldsymbol\delta}\in\C$ with
$|c_{\boldsymbol\delta}|=1$ such that
\[
\norm{
A_{\boldsymbol\delta}(\theta_*)
-
c_{\boldsymbol\delta}G(w)f(w)^{p-1}
}
\le
\frac{s^2\lambda}{q}2^{-n/4-\Omega(n)}.
\]
The triangle inequality proves
\cref{eqn:svp-estimator-accuracy}.
\end{proof}

We conclude this section by proving an efficient batch computation of estimator at all $\theta.$

\begin{lemma}
\label{lem:svp-estimator-computation}
Except with probability $o(1)$ over
$\boldsymbol\delta$, all vectors
$\widehat A_{\boldsymbol\delta}(\theta)$ for
$\theta\in H_m$ can be computed in time and space
$2^{n/2+o(n)}$.
\end{lemma}

\begin{algorithmblock}{Computation of $\widehat A_{\boldsymbol\delta}$}
\label{alg:svp-estimator-computation}
\item For every $r\in[0,p)$ and $i\in[N]$, let $Z_{0,i}=2\pi iX_{0,i}$ and $Z_{r,i}=0$ for $r\neq 0$. Construct lists
\[
\mathsf L_{0,r}:=
\{(i,Y_{r,i},Z_{r,i}):i\in[N]\}
\qquad\text{for }r\in[0,p).
\]
\item For $j=1,\ldots,m-1$ and $a\in[0,p/2^j)$, we do:
\begin{enumerate}
    \item For each $z\in H_j$, construct $\mathsf L_{j-1,2a+1,z} 
    =\{(\boldsymbol i_R,Y_R,Z_R)\in \mathsf L_{j-1,2a+1} : (Y_R)_{H_j}=z\}.
    $
    \item For every
    $(\boldsymbol i_L,Y_L,Z_L)\in\mathsf L_{j-1,2a}$,
    enumerate all elements
    $(\boldsymbol i_R,Y_R,Z_R)$ in
    $\mathsf L_{j-1,2a+1,(Y_L)_{H_j}-\delta_{j,a}}$
    and include
    $((\boldsymbol i_L,\boldsymbol i_R),Y_L-Y_R,Z_L+Z_R)$
    in $\mathsf L_{j,a}$.
    Abort immediately if
    $|\mathsf L_{j,a}|>Qp^{7\cdot2^j}$.
\end{enumerate}
\item For every $b\in H_m$, compute $
T_{\boldsymbol\delta}(b):=
\sum_{\substack{(\boldsymbol i,Y,Z)\in\mathsf L_{m-1,0}\\
(Y)_{H_m}=b}}Z.$
\item Using a vector-valued fast Fourier transform over $H_m$, compute for
every $\theta\in H_m$
\begin{align}
\widehat A_{\boldsymbol\delta}(\theta)
=
\frac{Q^{p-1}}{N^p}
\sum_{b\in H_m}
T_{\boldsymbol\delta}(b)
\omega_q^{\ip{\theta}{b}}.
\label{eqn:svp-final-transform}
\end{align}
\end{algorithmblock}

\begin{proof}
We call the first index of the list $\mathsf L$ by level.
We first describe the elements in each list. The list
$\mathsf L_{j,a}$ is obtained from the initial lists indexed by
$a2^j,\ldots,(a+1)2^j-1.$
We claim by induction on $j$ that, for every choice of indices
$i_{a2^j},\ldots,i_{(a+1)2^j-1}\in[N]$, the algorithm creates at most one element of the form
\[
\left((i_{a2^j},\ldots,i_{(a+1)2^j-1}),Y=
\sum_{\ell=0}^{2^j-1}
\varepsilon_\ell
Y_{a2^j+\ell,i_{a2^j+\ell}},Z\right)
\]
in the list $\mathsf L_{j,a}$ where $Z=2\pi iX_{0,i_0}$ if $a=0$ and $Z=0$ otherwise.
Moreover, a full index tuple
$\boldsymbol i=(i_0,\ldots,i_{p-1})$ appears in
$\mathsf L_{m-1,0}$ if and only if
$C_{\boldsymbol\delta}(\boldsymbol i)=1$.

The statement holds for $j=0$ by the definition of
$\mathsf L_{0,a}$. Suppose that it holds at level $j-1$.
An element in $\mathsf L_{j,a}$ is obtained from the elements in
$\mathsf L_{j-1,2a}$ and $\mathsf L_{j-1,2a+1}$. Its index tuple is
the concatenation of the two index tuples, its second component is
$Y_L-Y_R$, and its third component is $Z_L+Z_R$. Since
\[
\varepsilon_{2^{j-1}+\ell}=-\varepsilon_\ell
\qquad\text{for }0\le\ell<2^{j-1},
\]
the stated formulas for $Y$ and $Z$ follow. 
The last claim on $C_{\boldsymbol{\delta}}(\boldsymbol{i})$ is not hard to see, because Step 2 combines the elements with $Y_R$ and $Y_L$ if and only if $
(Y_L-Y_R)_{H_j}=\delta_{j,a},$ and a tuple appears at level
$j$ if and only if all the corresponding conditions up to level $j$
are satisfied. For $j=m-1$, this means that all the constraints satisfied, meaning that $C_{\boldsymbol{\delta}}(\boldsymbol{i})=1$.

Because of this claim, $\mathsf L_{m-1,0}$ contains $(\boldsymbol i,Y_{m-1,0}(\boldsymbol i),2\pi iX_{0,i_0})$ such that $C_{\boldsymbol\delta}(\boldsymbol i)=1.$
Hence Step 3 gives
\begin{align*}
T_{\boldsymbol\delta}(b)
=
\sum_{\substack{\boldsymbol i\in[N]^p\\
(Y_{m-1,0}(\boldsymbol i))_{H_m}=b}} C_{\boldsymbol\delta}(\boldsymbol i) 
2\pi iX_{0,i_0}.
\end{align*}
Substituting this equality into \cref{eqn:svp-final-transform} gives exactly
\cref{eqn:svp-estimator}. This gives the correctness of each $\widehat{A}_{\boldsymbol{\delta}}(\theta)$.

It remains to compute the complexity. We first bound the list sizes.
Fix $U$, all samples, and a candidate index tuple
$\boldsymbol i$. Then all vectors $Y_{h,b}(\boldsymbol i)$ are fixed.
The conditions for the corresponding tuple to appear in
$\mathsf L_{j,a}$ are $(Y_{h,b})_{H_h}=\delta_{h,b}$ for $h \in [j]$ and $a2^{j-h}\le b<(a+1)2^{j-h}$. The number of them is $2^j-1$. 
Since the vectors $\delta_{h,b}$ are independent and
uniform over $H_h$, each condition holds with
probability $1/|H_h|=1/Q$. Therefore, all the conditions hold with probability
$Q^{-(2^j-1)}$.
By the linearity of the expectation,
\[
\E_{\boldsymbol\delta}[|\mathsf L_{j,a}|]
=
\frac{N^{2^j}}{Q^{2^j-1}}
=
Qp^{4\cdot2^j}.
\]
By Markov's inequality, $\Pr_{\boldsymbol\delta}
\left[
|\mathsf L_{j,a}|>Qp^{7\cdot2^j}
\right]
\le p^{-3\cdot2^j}.$
Because there are $p/2^j$ lists $\mathsf L$ at level $j$, 
\begin{align}
\Pr_{\boldsymbol\delta}[\text{the algorithm aborts}]
\le
\sum_{j=1}^{m-1}\frac{p}{2^j}p^{-3\cdot2^j}
\le
p^{-5}\sum_{j=1}^{m-1}2^{-j}
\le p^{-5}
=o(1)
\label{eqn:svp-computation-abort}
\end{align}
by a union bound, where we use $2^j\ge2$ for every $j\ge1$ and $p\to\infty$.

Suppose that the algorithm does not abort. The lists
$\mathsf L_{j-1,2a+1,z}$ in Step~2 are constructed by sorting
$\mathsf L_{j-1,2a+1}$ according to the $H_j$-coordinate of its second
component, which dominates the second step.
Given the size bound $Qp^{7\cdot 2^j},$
the total number of elements over all levels is at most
\begin{align}
\sum_{j=0}^{m-1}\frac{p}{2^j}Qp^{7\cdot2^j}
\le
2pQp^{7p}
\le
Qp^{7p+2}
=
2^{n/2+o(n)}
\label{eqn:svp-total-computation-records}
\end{align}
because $2^{m-1}=p$ and we use $2^j\le p$ in the exponent. The final bound is due to  $p=\Theta(n/\log^3n)$.
All other steps incur only polynomial multiplicative factors.

Up to polynomial factors, 
the third step can be done in time and space about $Q+|{\mathsf L}_{m-1,0}|=2^{n/2+o(n)}$ after classifying the elements in ${\mathsf L}_{m-1,0}$ based on $(Y)_{H_m}$, and the last step is done in time and space about $|H_m|=Q=2^{n/2+o(n)}$. 
Together with
\cref{eqn:svp-computation-abort,eqn:svp-total-computation-records},
this proves the lemma.
\end{proof}

\subsection{The algorithm}
\paragraph{Parameters.}
Fix a constant $t$ satisfying $t_0<t<1/4$, and retain the parameters
from \cref{eqn:svp-parameter}: let $p$ be the largest power of two
with $p\le n/\log^3 n$, let $m=\log_2p+1$, and let $q$ be the
smallest prime larger than $2\sqrt p$. Set $N:=Qp^4$.
Then
\[
p\log p=O(n/\log^2 n)=o(n),
\qquad Q,N=2^{n/2+o(n)}.
\]

\noindent\textbf{Input.}
A basis $B\in\R^{n\times n}$ of a full-rank lattice
$\cL=B\Z^n$.

\noindent\textbf{Output.}
A shortest nonzero vector of $\cL$ with probability at least $2/3$.
\begin{algorithmblock}{The SVP algorithm}
\label{alg:main-svp}
\item Apply LLL reduction to obtain a nonzero $x\in\cL$ with
$\norm{x}\le2^{n/2}\lambda_1(\cL)$, and let
\[
\mathcal D:=
\{(1+1/n)^{-j}\norm{x}:0\le j\le n^2\}.
\]
\item For every $d\in\mathcal D$, independently repeat the following
procedure a constant number of times:
\begin{enumerate}
    \item Run the preprocessing algorithm in
    \cref{thm:bdd-preprocessing}.
    \item Let $s^2=4nt\ln2/(\pi d^2)$ and obtain $pN$ samples $X_{r,1},\ldots,X_{r,N}$ from
    $D_{\cL^*,s}$ for $r\in [0,p)$.
    Abort this repetition if a sample has
    norm larger than $Cs\sqrt n$, where $C$ is a sufficiently large
    constant.
    \item Choose random $U$ and $\boldsymbol\delta$, and run
    \cref{alg:svp-estimator-computation}. Abort this repetition if
    that algorithm aborts.
    \item For every $\theta\in H_m$ and every nonzero vector
    $V\in\{\Re\widehat A_{\boldsymbol\delta}(\theta),
    \Im\widehat A_{\boldsymbol\delta}(\theta)\}$, query the BDD
    algorithm at both $dV/\norm V$ and $-dV/\norm V$.
\end{enumerate}
\item Among all nonzero lattice vectors returned by the BDD algorithm,
return one of minimum Euclidean norm. If none was returned, report
failure.
\end{algorithmblock}

\begin{proof}[Proof of \cref{thm:main-svp}]
Fix a shortest vector $v$, let $\lambda=\norm v$, and consider $d\in\mathcal D$ satisfying $\lambda\le d\le(1+1/n)\lambda$, whose
existence follows from the LLL bound \cite{LLL82} and the definition of
$\mathcal D$. By \cref{cor:working-scale}, the parameter $s$ satisfies
the assumption of \cref{thm:dgs-sampling} with (say) $\kappa=n^2$. Hence, the required $pN$
samples are obtained using $2^{o(n)}$ calls, and their joint
distribution is inverse-exponentially close to that of independent
discrete Gaussian samples.

Let $w=B^{-1}v\bmod q$ and
$\theta_*=(U^{-T}w)_{H_m}$. Note that $w\ne0$, since otherwise
$v=qv'$ for some nonzero $v'\in\cL$, contradicting the minimality of
$v$.
By \cref{lem:svp-estimator-accuracy,lem:svp-recovery}, except with
probability $o(1)$, one of the targets queried to the BDD algorithm
for $\theta=\theta_*$ is within
$2\lambda/n<n^{-1/3}\lambda$ of $v$. Therefore, successful BDD
preprocessing returns $v$. Since $v$ is among the returned vectors
and every nonzero lattice vector has norm at least $\lambda$, Step~3
returns a shortest vector.

It remains to bound the complexity. Since
$pN=Qp^5=2^{n/2+o(n)}$, the DGS calls and
\cref{lem:svp-estimator-computation} use time and space
$2^{n/2+o(n)}$. The BDD-query step makes $Q2^{o(n)}$ queries in each
repetition, whose total time is also $2^{n/2+o(n)}$. The
$n^2+1$ choices of $d$ and the constant number of repetitions incur
only a polynomial factor.

By
\cref{lem:dgs-tail,lem:svp-estimator-accuracy,lem:svp-estimator-computation,thm:bdd-preprocessing},
each repetition for the value $d$ fixed above succeeds with constant
probability and has worst-case time and space $2^{n/2+o(n)}$. A
sufficiently large constant number of independent repetitions raises
the success probability to at least $2/3$.
\end{proof}

\input{4z.proof}

%% file: 4z.proof.tex
\subsection{Proofs of the lemmas}
\label{subsec:prooflem}

\cref{fig:svp-mass-analysis-flow} shows a high-level flow of the analysis of the proof of \cref{lem:svp-nonprincipal-mass}.

\input{fig_analysis}

\paragraph{Gaussian mass bounds.}
For $j\in\{0,\ldots,m-1\}$, let
\begin{align}
P_j&:=H_1\oplus\cdots\oplus H_j,&
S_j&:=U^TP_j,&
R_j&:=H_{j+1}\oplus\cdots\oplus H_m,\notag\\
K_j&:=\{Bx:x\in\Z^n,\ x\bmod q\in S_j\},&
\sigma_j&:=\frac{q}{2^{j/2}s},
\label{eqn:svp-mass-notation}
\end{align}
where $P_0=S_0=\{0\}$. Thus
$\F_q^n=P_j\oplus R_j$. Thus $R_j$ contains exactly one
representative of every coset of $P_j$, and the map
$\kappa+P_j\mapsto U^T\kappa+S_j$ is a bijection from
$\F_q^n/P_j$ to $\F_q^n/S_j$. Recall
$g=\frac12\log_2(t/t_0)>0$ and
$s^2=4nt\ln2/(\pi d^2)$, where $t_0<t<1/4$.
Throughout this subsection, $t$ is a fixed constant, and hence
$g=\Omega(1)$.

We use the following Gaussian mass bounds regarding the cosets. Note that the second inequality has a factor $1/\sqrt{1-1/n}$ in the width. Looking ahead, this term is, in the norm of gradient terms, to absorb the scalar factor in the exponent at the cost of small factors; e.g., $xe^{-x} \le ne^{-(1-1/n)x}$.
\begin{lemma}
\label{lem:svp-random-lattice-masses}
Let $\lambda=\lambda_1(\cL)\le d\le(1+1/n)\lambda$. Except with
probability $2^{-\Omega(n)}$ over $U$, simultaneously for every
$j\in\{0,\ldots,m-1\}$,
\begin{align}
\rho_{\sigma_j}(K_j\setminus\{0\})
&\le2^{-(g/2-o(1))n},&
\rho_{\sigma_j/\sqrt{1-1/n}}(K_j\setminus\{0\})
&\le2^{-(g/2-o(1))n}.
\label{eqn:svp-random-lattice-gradient-mass}
\end{align}
\end{lemma}

\begin{proof}
Since $U$ is uniform in $\GL_n(\F_q)$, the subspace $S_j=U^TP_j$
is uniform among all $jn/m$-dimensional subspaces of $\F_q^n$.
Consequently, every fixed nonzero vector of $\F_q^n$ belongs to
$S_j$ with probability
\[
\frac{q^{jn/m}-1}{q^n-1}\le2q^{-(n-jn/m)}.
\]
The sublattice $q\cL$ is contained in $K_j$ independently of $U$,
and
$\rho_\sigma(q\cL\setminus\{0\})
=\rho_{\sigma/q}(\cL\setminus\{0\})$. Therefore, for every
$\sigma>0$,
\begin{align}
\E_U[\rho_\sigma(K_j\setminus\{0\})]
&=
\rho_\sigma(q\cL\setminus\{0\})
+
\sum_{x\in\cL\setminus q\cL}
\Pr_U[x\in K_j]\rho_\sigma(x)\notag\\
&\le
\rho_{\sigma/q}(\cL\setminus\{0\})
+2q^{-(n-jn/m)}\rho_\sigma(\cL\setminus\{0\}).
\label{eqn:svp-random-lattice-expectation}
\end{align}
We apply this inequality with $\sigma=\sigma_j$. By
\cref{lem:shell} with $k=0$, first at width $\sigma_j/q$ and then at
width $\sigma_j$,
\begin{align*}
\rho_{\sigma_j/q}(\cL\setminus\{0\})
&\le
2^{o(n)}
\left(\frac{\beta^2n}{2\pi e\lambda^2\,2^js^2}\right)^{n/2}
=2^{-(g+(j+1)/2-o(1))n},\notag\\
\rho_{\sigma_j}(\cL\setminus\{0\})
&\le
2^{o(n)}
\left(\frac{\beta^2nq^2}{2\pi e\lambda^2\,2^js^2}\right)^{n/2}
=2^{(\log_2q-(j+1)/2-g+o(1))n}.
\end{align*}
Here we used $d/\lambda=1+O(1/n)$ and
$t_0=\beta^2/(4e\ln2)$. Since
$\log_2q=(m+1)/2+o(1)$, substituting these bounds into
\cref{eqn:svp-random-lattice-expectation} gives
\[
\E_U[\rho_{\sigma_j}(K_j\setminus\{0\})]
\le
2^{-(g+(j+1)/2-o(1))n}
+2^{-(g+(m-j)/(2m)-o(1))n}
\le2^{-(g-o(1))n}.
\]
Replacing $\sigma_j$ by
$\sigma_j/\sqrt{1-1/n}$ multiplies each upper bound obtained from
\cref{lem:shell} by at most
$(1-1/n)^{-n/2}=2^{O(1)}$. Hence, the same argument gives
\[
\E_U\left[
\rho_{\sigma_j/\sqrt{1-1/n}}(K_j\setminus\{0\})
\right]
\le2^{-(g-o(1))n}.
\]
Markov's inequality and a union bound over the $2m$ estimates prove
the lemma.
\end{proof}

\paragraph{Recursive mass representations.}
For $\kappa\in\F_q^n$, we define
\begin{align}
V_0(\kappa)&:=|f(U^T\kappa)|^2,&
T_0(\kappa)&:=\norm{G(U^T\kappa)}^2.
\label{eqn:svp-frequency-mass-base}
\end{align}
For $j\in\{1,\ldots,m-1\}$ and $\kappa\in\F_q^n$, we iteratively define
\begin{align}
V_j(\kappa)
&:=\sum_{\alpha\in H_j}V_{j-1}(\kappa+\alpha)^2,&
T_j(\kappa)
&:=\sum_{\alpha\in H_j}
T_{j-1}(\kappa+\alpha)V_{j-1}(\kappa+\alpha).
\label{eqn:svp-frequency-mass-recursion}
\end{align}
Induction shows that $V_j$ and $T_j$ are even and $P_j$-periodic:
for every $\kappa\in\F_q^n$ and $\gamma\in P_j$,
\[
V_j(-\kappa)=V_j(\kappa),\quad
T_j(-\kappa)=T_j(\kappa),\quad
V_j(\kappa+\gamma)=V_j(\kappa),\quad
T_j(\kappa+\gamma)=T_j(\kappa).
\]
Consequently, $V_j(\kappa)$ and $T_j(\kappa)$ depend only on the
coset $U^T\kappa+S_j$. Under the bijection
$\kappa+P_j\mapsto U^T\kappa+S_j$, the recurrence in
\cref{eqn:svp-frequency-mass-recursion} can be seen as the recurrence
obtained by summing over the $S_{j-1}$-cosets contained in a fixed
$S_j$-coset.

We first explain the meaning of $V_j$ and $T_j$. 
For $j\in\{0,\ldots,m-1\}$, let
\begin{align*}
\mathcal A_{\le j}
:=
\left\{
(\alpha_{h,a})_{\substack{1\le h\le j\\0\le a<2^{j-h}}}:
\alpha_{h,a}\in H_h
\right\}.
\end{align*}
We write an element of $\mathcal A_{\le j}$ as
$\boldsymbol\alpha^{\le j}$. The set $\mathcal A_{\le0}$ contains
only the empty tuple. At $j=m-1$, its indices are
$1\le h\le m-1$ and $0\le a<p/2^h$, so we identify
$\boldsymbol\alpha^{\le m-1}$ with the tuple
$\boldsymbol\alpha$ used in \cref{eqn:svp-all-characters}.
For $\kappa\in\F_q^n$,
$\boldsymbol\alpha^{\le j}\in\mathcal A_{\le j}$, and
$r\in[0,2^j)$, let
\begin{align}
k_{\boldsymbol\alpha^{\le j},\kappa,r}
:=
\varepsilon_r\kappa
+
\sum_{h=1}^j
\varepsilon_{r\bmod2^h}
\alpha_{h,\lfloor r/2^h\rfloor}.
\label{eqn:svp-subtree-leaf-frequency}
\end{align}

We prove the following lemma. At $j=m-1$ and $\kappa=\theta\in H_m$, the term $k_{\boldsymbol\alpha^{\le m-1},\theta,r}$ in
\cref{eqn:svp-subtree-leaf-frequency} agree with
$k_{\boldsymbol\alpha,\theta,r}$ from
\cref{eqn:svp-leaf-frequency-identity}. Since $2^{m-1}=p$, this gives the following equality (cf. \cref{eqn:svp-nonprincipal-mass}):
\begin{align}
T_{m-1}(\theta)
=
\sum_{\boldsymbol\alpha}
\norm{
G(U^Tk_{\boldsymbol\alpha,\theta,0})
\prod_{r=1}^{p-1}f(U^Tk_{\boldsymbol\alpha,\theta,r})
}^2.
\label{eqn:svp-full-tree-mass-closed-form}
\end{align}

\begin{lemma}
\label{lem:svp-frequency-mass-closed-form}
For every $j\in\{0,\ldots,m-1\}$ and every $\kappa\in\F_q^n$,
\[
V_j(\kappa)
=
\sum_{\boldsymbol\alpha^{\le j}\in\mathcal A_{\le j}}
\prod_{r=0}^{2^j-1}
|f(U^Tk_{\boldsymbol\alpha^{\le j},\kappa,r})|^2,\qquad
T_j(\kappa)
=
\sum_{\boldsymbol\alpha^{\le j}\in\mathcal A_{\le j}}
\norm{
G(U^Tk_{\boldsymbol\alpha^{\le j},\kappa,0})
\prod_{r=1}^{2^j-1}
f(U^Tk_{\boldsymbol\alpha^{\le j},\kappa,r})
}^2.
\]
\end{lemma}

\begin{proof}
We will prove that these terms satisfy \cref{eqn:svp-frequency-mass-recursion} by induction.
The claim for $j=0$ follows from
\cref{eqn:svp-frequency-mass-base}. 

Suppose that the statement holds at level $j-1$. 
We can uniquely decompose an element
$\boldsymbol\alpha^{\le j}\in\mathcal A_{\le j}$ into $(\alpha,\boldsymbol\beta^L,\boldsymbol\beta^R)
\in
H_j\times\mathcal A_{\le j-1}\times\mathcal A_{\le j-1}$ where we
let
\begin{align}
\alpha:=\alpha_{j,0},\qquad
\boldsymbol\beta^L=(\beta^L_{h,a}:=\alpha_{h,a})_{h,a},\qquad
\boldsymbol\beta^R=(\beta^R_{h,a}:=\alpha_{h,a+2^{j-h-1}})_{h,a}
\label{eqn:svp-left-right-alpha}
\end{align}
where the indices are over $1\le h<j$ and $0\le a<2^{j-h-1}$.

By \cref{eqn:svp-subtree-leaf-frequency}, for
$0\le r<2^{j-1}$, it is not hard to see that
\begin{align}
k_{\boldsymbol\alpha^{\le j},\kappa,r}=
k_{\boldsymbol\beta^L,\kappa+\alpha,r},\qquad
k_{\boldsymbol\alpha^{\le j},\kappa,2^{j-1}+r}
=
k_{\boldsymbol\beta^R,-\kappa-\alpha,r}
\label{eqn:svp-left-right-k}
\end{align}
because of $(2^{j-1}+r)\bmod 2^h=r\bmod 2^h,
\left\lfloor\frac{2^{j-1}+r}{2^h}\right\rfloor
=
2^{j-h-1}+\left\lfloor\frac r{2^h}\right\rfloor$ for $h<j$ and
$\varepsilon_{2^{j-1}+r}=-\varepsilon_r$.

Using \cref{eqn:svp-left-right-k} and the induction hypothesis, we
obtain
\begin{align*}
&\sum_{\boldsymbol\alpha^{\le j}\in\mathcal A_{\le j}}
\prod_{r=0}^{2^j-1}
|f(U^Tk_{\boldsymbol\alpha^{\le j},\kappa,r})|^2\notag\\
&\quad=
\sum_{\alpha\in H_j}
\left(
\sum_{\boldsymbol\beta^L\in\mathcal A_{\le j-1}}
\prod_{r=0}^{2^{j-1}-1}
|f(U^Tk_{\boldsymbol\beta^L,\kappa+\alpha,r})|^2
\right)
\left(
\sum_{\boldsymbol\beta^R\in\mathcal A_{\le j-1}}
\prod_{r=0}^{2^{j-1}-1}
|f(U^Tk_{\boldsymbol\beta^R,-\kappa-\alpha,r})|^2
\right)
\notag\\
&
\quad=
\sum_{\alpha\in H_j}
V_{j-1}(\kappa+\alpha)V_{j-1}(-\kappa-\alpha)
=
\sum_{\alpha\in H_j}V_{j-1}(\kappa+\alpha)^2
=
V_j(\kappa).
\end{align*}
Similarly,
\begin{align*}
\sum_{\boldsymbol\alpha^{\le j}\in\mathcal A_{\le j}}
\norm{
G(U^Tk_{\boldsymbol\alpha^{\le j},\kappa,0})
\prod_{r=1}^{2^j-1}
f(U^Tk_{\boldsymbol\alpha^{\le j},\kappa,r})
}^2
&=
\sum_{\alpha\in H_j}
T_{j-1}(\kappa+\alpha)V_{j-1}(-\kappa-\alpha)\notag\\
&=
\sum_{\alpha\in H_j}
T_{j-1}(\kappa+\alpha)V_{j-1}(\kappa+\alpha)
=
T_j(\kappa).
\end{align*}
Here we use the fact that $V_{j-1}$ is even.
This proves the two identities at level $j$.
\end{proof}

We need the following upper bound of the sum of $V$'s and $T$'s. 
The proof of this lemma is rather involved.
The tree structure is used in \cref{eqn:svp-all-square-decompositions}, which gives an upper bound of the first term as a product of discrete Gaussian masses. Then we use the known Gaussian masses to obtain the desired bound. For the second term, we use a trick like $xe^{-x} \le ne^{-(1-1/n)x}$ to absorb the $G$ terms in the exponent at the cost of small factor in the exponent. 
We then use the second bound of \cref{lem:svp-random-lattice-masses} instead. The same proof strategy is used when proving \cref{lem:svp-fixed-tree-masses}.

\begin{lemma}
\label{lem:svp-total-frequency-masses}
Let $\lambda=\lambda_1(\cL)\le d\le(1+1/n)\lambda$. Except with
probability $2^{-\Omega(n)}$ over $U$, simultaneously for every
$j\in\{1,\ldots,m-1\}$,
\begin{align}
\sum_{\kappa\in R_j}V_j(\kappa)
\le1+2^{((m-j-1)/2-g/2+o(1))n},\qquad
\sum_{\kappa\in R_j}T_j(\kappa)
\le
\left(\frac{s^2\lambda}{q}\right)^2
2^{((m-j-1)/2-g/2+o(1))n}.
\label{eqn:svp-total-frequency-mass-bounds}
\end{align}
\end{lemma}

\begin{proof}
Recall $R_{j-1}=H_j\oplus R_j$ by definition (\cref{eqn:svp-mass-notation})
We define for $j\in \{1,\ldots,m-1\}$ 
\begin{align}
\overline V_j:=\sum_{\kappa\in R_j}V_j(\kappa)=
\sum_{\kappa\in R_{j-1}}V_{j-1}(\kappa)^2,\qquad
\overline T_j:=\sum_{\kappa\in R_j}T_j(\kappa)
=
\sum_{\kappa\in R_{j-1}}
T_{j-1}(\kappa)V_{j-1}(\kappa).
\label{eqn:svp-total-frequency-mass-identities}
\end{align}
where the last equalities are obtained by summing
\cref{eqn:svp-frequency-mass-recursion} over $\kappa\in R_j$.

We first give a bound of $\overline V_j$, whose terms are products of $2^{j}$ factors of the form $|f(z)|^2$ for some $z$ thanks to \cref{lem:svp-frequency-mass-closed-form}.
More precisely, for each summand indexed by $\kappa$ and
$\boldsymbol\alpha^{\le j}$, we define
\[
z_r
:=
\varepsilon_rU^T
k_{\boldsymbol\alpha^{\le j},\kappa,r}.
\]
Using the fact that $f$ is even and the above equation, we have 
\[
\left|
f\left(
U^Tk_{\boldsymbol\alpha^{\le j},\kappa,r}
\right)
\right|^2
=
|f(z_r)|^2
=
\frac{1}{\rho_{1/s}(\cL)^2}
\sum_{u_{r,0},u_{r,1}\in Bz_r+q\cL}
e^{-\pi s^2(\norm{u_{r,0}}^2+\norm{u_{r,1}}^2)/q^2}
\]
using the representation of $f$ in \cref{eqn:odd-local-primal}. Plugging this to \cref{lem:svp-frequency-mass-closed-form} and expanding introduces vectors $u_{r,0},u_{r,1}\in Bz_r+q\cL$ for $r\in[0,2^j)$ in the exponents,
and the denominator $\rho_{1/s}(\cL)^{2^{j+1}}$. 
That is, the summand in $\overline{V}_j$ can be indexed by $(\boldsymbol\alpha^{\le j},\kappa)$ (from \cref{lem:svp-frequency-mass-closed-form}) and $(u_{r,b})_{r,b}$ (from the expansion).

\begin{claim}
\label{clm:svp-expanded-sum-indexing}
In the expansion of $\overline V_j$, the tuple
$(u_{r,b})_{r,b}$ uniquely determines the summation indices
$(\boldsymbol\alpha^{\le j},\kappa)$. Consequently, the expanded sum
can be indexed by $(u_{r,b})_{r,b}$ without multiplicity.
\end{claim}

\begin{proof}
Since
\[
B^{-1}u_{r,0}\bmod q
=
z_r
=
\varepsilon_rU^T
k_{\boldsymbol\alpha^{\le j},\kappa,r},
\]
the tuple $(u_{r,b})_{r,b}$ determines every leaf frequency via
\[
k_{\boldsymbol\alpha^{\le j},\kappa,r}
=
\varepsilon_rU^{-T}
\left(B^{-1}u_{r,0}\bmod q\right).
\]

Given this, $\kappa$ is determined by $k_{\boldsymbol\alpha^{\le j},\kappa,0}
=
\kappa+\sum_{h=1}^j\alpha_{h,0}$ from \cref{eqn:svp-subtree-leaf-frequency};  and $\varepsilon_0=1$;
since $\kappa\in R_j$ and $\alpha_{h,0}\in H_h$, the $R_j$-component
of this frequency is $\kappa$ under the decomposition
\[
\F_q^n
=
R_j\oplus H_1\oplus\cdots\oplus H_j.
\]

Similarly, $\alpha_{h,a}$ is determined by the $H_h$-component of
$k_{\boldsymbol\alpha^{\le j},\kappa,r}$ for $r=2^ha$, which is
\[
\varepsilon_{r\bmod 2^h}
\alpha_{h,\lfloor r/2^h\rfloor}=
\varepsilon_{(2^ha)\bmod 2^h}
\alpha_{h,\lfloor 2^ha/2^h\rfloor}
=
\varepsilon_0\alpha_{h,a}
=
\alpha_{h,a}.
\]
Thus all the variables $\kappa$ and $\alpha_{h,a}$ are uniquely
determined by $(u_{r,b})_{r,b}$.
\end{proof}

We will simplify the exponent of each term. More precisely, we prove that 
\begin{align}
\sum_{r=0}^{2^j-1}
\left(\norm{u_{r,0}}^2+\norm{u_{r,1}}^2\right)
=
2^{j+1}\norm{a_{j,0}}^2
+\frac12\sum_{r=0}^{2^j-1}\norm{d_{0,r}}^2
+\sum_{h=1}^j2^{h-1}
\sum_{r=0}^{2^{j-h}-1}\norm{d_{h,r}}^2
\label{eqn:svp-all-square-decompositions}
\end{align}
where $a_{j,0}$ and $(d_{h,r})_{h,r}$ are defined below.
These
vectors are uniquely determined by the vectors $u_{r,b}$, and the
transformation is invertible. 
More precisely, for
$r\in[0,2^j)$, let
\begin{align}
    \label{eqn:average-difference}
a_{0,r}:=\frac{u_{r,0}+u_{r,1}}2,\qquad
d_{0,r}:=u_{r,0}-u_{r,1},
\end{align}
and, for $h\in\{1,\ldots,j\}$ and $r\in[0,2^{j-h})$, let
\[
a_{h,r}:=\frac{a_{h-1,2r}+a_{h-1,2r+1}}2,\qquad
d_{h,r}:=a_{h-1,2r}-a_{h-1,2r+1}.
\]
It is not hard to see that these transformations are invertible.
\cref{eqn:svp-all-square-decompositions} is proven by repeatedly applying
\[
\norm{x_L}^2+\norm{x_R}^2
=
2\norm{\frac{x_L+x_R}{2}}^2
+\frac12\norm{x_L-x_R}^2.
\]
We call the map
\[
(u_{r,b})_{\substack{r\in[0,2^j)\\b\in\bit}}
\longmapsto
\left(
a_{j,0},
(d_{h,r})_{\substack{h\in[0,j+1)\\r\in[0,2^{j-h})}}
\right)
\]
the \emph{average-difference transformation}. The following claim
provides the properties of this transformation that we need.

\begin{claim}
\label{clm:svp-average-difference}
The average-difference transformation extends to a linear bijection
\[
\Psi_j:
(\R^n)^{[0,2^j)\times\{0,1\}}
\longrightarrow
\R^n\times
\prod_{h=0}^j(\R^n)^{[0,2^{j-h})}.
\]
When restricted to the tuples $(u_{r,b})_{r,b}$ occurring in the
expansion of $\overline V_j$, this transformation has the following
additional property. 
\begin{itemize}[nosep]
\item If all difference but $a_{j,0}$ are fixed, the possible values of
$a_{j,0}$ are contained in a single coset of $\cL$.
\item Every $d_{0,r}$ belongs to $K_0$.
\item For every $h\in[j]$ and $r\in[0,2^{j-h})$, if all
lower-level differences $d_{\ell,r'}$ with $\ell<h$ are fixed, the
possible values of $d_{h,r}$ are contained in a single coset of
$K_{h-1}$.
\end{itemize}
Each coset may depend on the fixed variables.
\end{claim}

\begin{proof}
At each node, the map
\[
(x_L,x_R)
\longmapsto
\left(\frac{x_L+x_R}{2},x_L-x_R\right)
\]
has inverse
\[
(a,d)\longmapsto
\left(a+\frac d2,a-\frac d2\right),
\]
which proves the bijection.

We now restrict to the tuples occurring in the expansion of
$\overline V_j$. Fix all but one of the transformed variables and
compare two possible values of the unfixed variable. If $a_{j,0}$
changes by $x$, the inverse transformation shows that every
$u_{r,b}$ changes by $x$. Since every $u_{r,b}$ belongs to $\cL$,
we have $x\in\cL$. Moreover, the two variables $u_{r,0}$ and
$u_{r,1}$ come from the same coset of $q\cL$, and hence
\[
d_{0,r}=u_{r,0}-u_{r,1}\in q\cL=K_0.
\]

It remains to consider $d_{h,r}$ for $h\ge1$. 
Recall the relation, for every $t$ and $b$,
\[
B^{-1}u_{t,b}\bmod q
=
\varepsilon_tU^T
k_{\boldsymbol\alpha^{\le j},\kappa,t}.
\]
On the other hand, the definition of $d_{h,r}=a_{h-1,2r}-a_{h-1,2r+1}$ with the average transformations gives
\[
2^hd_{h,r}
=
\sum_{i=0}^{2^{h-1}-1}\sum_{b\in\bit}
\left(
u_{t_i^L,b}-u_{t_i^R,b}
\right),
\]
where $t_i^L:=2^hr+i,$ and $t_i^R:=2^hr+2^{h-1}+i.$
Combining these identities yields
\[
B^{-1}(2^hd_{h,r})\bmod q
=
U^T
\sum_{i=0}^{2^{h-1}-1}\sum_{b\in\bit}
\left(
\varepsilon_{t_i^L}
k_{\boldsymbol\alpha^{\le j},\kappa,t_i^L}
-
\varepsilon_{t_i^R}
k_{\boldsymbol\alpha^{\le j},\kappa,t_i^R}
\right).
\]
Expanding the two $k$'s using
\cref{eqn:svp-subtree-leaf-frequency}, the $\kappa$-terms and the
$\alpha_{\ell,a}$-terms with $\ell\ge h$ have the same indices and
coefficients and hence cancel. Every remaining term can be written as
\[
\varepsilon_{a_L}\alpha_{\ell,a_L}
-
\varepsilon_{a_R}\alpha_{\ell,a_R}
\]
for some $\ell<h$ and some $a_L,a_R$, where
$\alpha_{\ell,a_L},\alpha_{\ell,a_R}\in H_\ell$, so that
this term belongs to $H_\ell$.
Thus the expression in parentheses
belongs to
\[
H_1\oplus\cdots\oplus H_{h-1}=P_{h-1}.
\]
Consequently,
\[
B^{-1}(2^hd_{h,r})\bmod q\in U^TP_{h-1}=S_{h-1}.
\]
Since $2^hd_{h,r}$ is a sum of vectors in $\cL$, this shows that
$2^hd_{h,r}\in K_{h-1}$.

Fix all differences $d_{\ell,s}$ with $\ell<h$, and let $v$ and
$v'$ be any two possible values of $d_{h,r}$ consistent with these
fixed values. By the inverse transformation, the corresponding
vectors $u_{t,b}$ change by a common lattice vector within each of
the two groups defining $d_{h,r}$. Denoting these changes by
$x_L,x_R\in\cL$, we have
\[
v'-v=x_L-x_R\in\cL.
\]
Moreover, the preceding containment gives $2^hv,\,2^hv'\in K_{h-1},$
so $2^h(v'-v)\in K_{h-1}$. Since $q$ is odd, multiplication by
$2^h$ is invertible modulo $q$, and therefore
$v'-v\in K_{h-1}$. Thus all possible values of $d_{h,r}$ lie in a
single coset of $K_{h-1}$.
\end{proof}

By the bijection and the uniqueness observation above, every term in
the expansion of $\overline V_j$ is counted exactly once in the
transformed variables. We first sum over $a_{j,0}$ and then over the
variables $d_{h,r}$ in decreasing order of $h$. Their numbers,
containing cosets, and Gaussian widths are as follows:
\[
\begin{array}{c|c|c|c}
\text{variable}
&
\text{number}
&
\text{containing coset}
&
\text{width}
\\ \hline
a_{j,0}
&
1
&
\text{a coset of }\cL
&
\dfrac{q}{2^{(j+1)/2}s}
\\[2mm]
d_{0,r}
&
2^j
&
K_0
&
\dfrac{\sqrt2q}{s}
\\[2mm]
d_{1,r}
&
2^{j-1}
&
\text{a coset of }K_0&
\dfrac{q}{s}
\\[2mm]
d_{h+1,r}
&
2^{j-h-1}
&
\text{a coset of }K_h
&
\dfrac{q}{2^{h/2}s}=\sigma_h
\quad(1\le h<j).
\end{array}
\]
The widths follow from the coefficients in
\cref{eqn:svp-all-square-decompositions}. At each step, the values
that occur in the current sum form a subset of one coset in the
corresponding row. Extending the sum to the entire coset and applying
\cref{lem:shifted-mass-maximum} therefore gives
\begin{align}
\overline V_j
\le
\frac{\rho_{q/(2^{(j+1)/2}s)}(\cL)}
{\rho_{1/s}(\cL)^{2^{j+1}}}
\rho_{\sqrt2q/s}(K_0)^{2^j}
\rho_{q/s}(K_0)^{2^{j-1}}
\prod_{h=1}^{j-1}
\rho_{\sigma_h}(K_h)^{2^{j-h-1}}.
\label{eqn:svp-total-frequency-mass-product}
\end{align}
The denominator comes from the $2^j$ factors $|f(z_r)|^2$, each of
which contributes $\rho_{1/s}(\cL)^{-2}$.

We remove the denominator using $\rho_{1/s}(\cL)^{2^{j+1}}\ge1$, which only increases it.
By \cref{lem:shell},
\[
\rho_{q/(2^{(j+1)/2}s)}(\cL\setminus\{0\})
\le
2^{o(n)}
\left(
\frac{\beta^2nq^2}
{2\pi e\lambda^2\,2^{j+1}s^2}
\right)^{n/2}
=
2^{o(n)}
\left(
\frac{t_0q^2d^2}
{2^{j+2}t\lambda^2}
\right)^{n/2}
=
2^{((m-j-1)/2-g+o(1))n},
\]
where we use $t_0/t=2^{-2g}$, $d/\lambda = 1+O(1/n)$, and $\log_2 q = \frac{m+1}{2}+o(1)$.
Since $K_0=q\cL$, \cref{cor:working-scale} gives
\[
\rho_{\sqrt2q/s}(K_0)
=
\rho_{\sqrt2/s}(\cL)
\le
1+2^{-(g-o(1))n},
\qquad
\rho_{q/s}(K_0)
=
\rho_{1/s}(\cL)
\le
1+2^{-(1/2+g-o(1))n}.
\]
Moreover, \cref{lem:svp-random-lattice-masses} gives $\rho_{\sigma_h}(K_h)
\le
1+2^{-(g/2-o(1))n}$ for $1\le h\le j-1.$
Thus every remaining numerator mass in
\cref{eqn:svp-total-frequency-mass-product} is at most
$1+2^{-(g/2-o(1))n}$.
Multiplying all these bounds, we have the desired bound
\begin{align}
\label{eqn: overlineV_bound}
\overline V_j
\le1+2^{((m-j-1)/2-g/2+o(1))n}.
\end{align}

For $\overline T_j$, one of the factors $|f(z)|^2$ is replaced by
$\norm{G(z)}^2$. By the representation of $G$ in \cref{eqn:odd-local-primal},
\[
G(z)
=
-\frac{2\pi s^2}{q\rho_{1/s}(\cL)}
\sum_{u\in Bz+q\cL}
u\,e^{-\pi s^2\norm{u}^2/q^2}.
\]
Applying the Cauchy-Schwarz inequality gives
\begin{align}
\norm{G(z)}^2
&\le
\left(\frac{2\pi s^2}{q}\right)^2
\frac{1}{\rho_{1/s}(\cL)^2}
\left(
\sum_{u_0\in Bz+q\cL}
\norm{u_0}^2e^{-\pi s^2\norm{u_0}^2/q^2}
\right)
\left(
\sum_{u_1\in Bz+q\cL}
e^{-\pi s^2\norm{u_1}^2/q^2}
\right) \notag\\
&=
\left(\frac{2\pi s^2}{q}\right)^2
\frac{1}{\rho_{1/s}(\cL)^2}
\sum_{u_0,u_1\in Bz+q\cL}
\norm{u_0}^2
e^{-\pi s^2(\norm{u_0}^2+\norm{u_1}^2)/q^2}.
\label{eqn:svp-gradient-cauchy-schwarz}
\end{align}
Thus the expansion for $\overline T_j$ is bounded by the
expansion for $\overline V_j$ with an additional factor
$(2\pi s^2/q)^2\norm{u_{0,0}}^2$ in each summand.

From the definitions of $a_{h,r}$ and $d_{h,r}$, we have $u_{0,0}
=
a_{j,0}+\frac12\sum_{h=0}^j d_{h,0} $ which gives
\begin{align*}
\norm{u_{0,0}}^2
&\le
(j+2)\left(
\norm{a_{j,0}}^2
+\frac14\sum_{h=0}^j\norm{d_{h,0}}^2
\right) \notag
\\&
\le (j+1) \left( 2^{j+1}\norm{a_{j,0}}^2
+\frac12\sum_{r=0}^{2^j-1}\norm{d_{0,r}}^2
+\sum_{h=1}^j2^{h-1}
\sum_{r=0}^{2^{j-h}-1}\norm{d_{h,r}}^2 \right)
\end{align*}
Note that all the terms in the right hand side appear in \cref{eqn:svp-all-square-decompositions}.
We will use the inequality for all $R \ge 0$
\begin{align}
R e^{-\pi s^2R/q^2}
\le
\frac{nq^2}{\pi s^2}
e^{-\pi(1-1/n)s^2R/q^2}.
\label{eqn:svp-gradient-moment-absorption}
\end{align}
because of $e^{x+y} \ge (1+x) e^y \ge xe^y$ with $x=\pi Rs^2/nq^2$. 
From this, we have
\begin{align*}
    &\norm{u_{0,0}}^2 \exp \left(-\frac{\pi s^2 \left(2^{j+1}\norm{a_{j,0}}^2
+\frac12\sum_{r=0}^{2^j-1}\norm{d_{0,r}}^2
+\sum_{h=1}^j2^{h-1}
\sum_{r=0}^{2^{j-h}-1}\norm{d_{h,r}}^2 \right)}{q^2}\right)
\\
& \le 2^{o(n)} \frac{nq^2}{\pi s^2} \exp\left(-\frac{\pi s^2 (1-1/n) \left(2^{j+1}\norm{a_{j,0}}^2
+\frac12\sum_{r=0}^{2^j-1}\norm{d_{0,r}}^2
+\sum_{h=1}^j2^{h-1}
\sum_{r=0}^{2^{j-h}-1}\norm{d_{h,r}}^2 \right)}{q^2}\right)
\end{align*}
where we use $j=2^{o(n)}$.

The moment absorption multiplies every Gaussian width in the table
above by $1/\sqrt{1-1/n}$, but does not change the summation domains
or their containing cosets. Hence the same ordered summation argument
based on \cref{clm:svp-expanded-sum-indexing,clm:svp-average-difference} applies with these widened
widths.

Note that the summand in which $a_{j,0}=0$ and every $d_{h,r}=0$ vanishes,
since then $u_{0,0}=0$. Hence, it suffices to bound the contribution
in which at least one of these variables is nonzero.

We classify every transformed tuple other than the all-zero tuple as
follows. Either all differences $d_{h,r}$ vanish and $a_{j,0}\ne0$, or there
is a smallest level $h$ containing a nonzero difference; in the
latter case, choose $r$ such that $d_{h,r}\ne0$. Thus
$d_{\ell,r'}=0$ for every $\ell<h$ and every $r'$.

In the first case, the tuple $u_{t,b}=0$ for every $t,b$ also occurs
in the expansion and has all differences equal to zero. Hence the
coset containing the possible values of $a_{j,0}$ contains zero and
is therefore $\cL$. Since $a_{j,0}\ne0$, the corresponding sum is
bounded by the Gaussian mass of $\cL\setminus\{0\}$.

In the second case, if $h=0$, then
$d_{0,r}\in K_0\setminus\{0\}$. Suppose that $h\ge1$. The tuple
$u_{t,b}=0$ for every $t,b$ has the same fixed lower-level
differences and has $d_{h,r}=0$. Therefore, the coset of
$K_{h-1}$ containing the possible values of $d_{h,r}$ contains zero
and hence equals $K_{h-1}$. Since $d_{h,r}\ne0$, we obtain $d_{h,r}\in K_{h-1}\setminus\{0\}.$

We bound all other variables by their full Gaussian masses as above
and sum over at most $2^{j+1}=2^{o(n)}$ choices of the distinguished
nonzero variable. The second estimate in
\cref{lem:svp-random-lattice-masses} and \cref{lem:shell} then bound
the Gaussian contribution by
\[
2^{((m-j-1)/2-g/2+o(1))n}.
\]
The prefactor from the primal representation of $G$,
Cauchy-Schwarz, and moment absorption is at most
$(s^2\lambda/q)^22^{o(n)}$. Multiplying this by the preceding
Gaussian-mass bound and absorbing all subexponential factors proves the second bound in
\cref{eqn:svp-total-frequency-mass-bounds}.
\end{proof}

\paragraph{The nonprincipal mass.}
As in \cref{lem:svp-nonprincipal-mass}, let $v\in\cL$ be a shortest
vector, $w=B^{-1}v\bmod q$, $k_*=U^{-T}w$, and
$\theta_*=(k_*)_{H_m}$. For $j\in\{0,\ldots,m-1\}$, let
\[
C_j^*:=U^Tk_*+S_j=w+S_j, \qquad C_{j-1}^*\subseteq C_j^*\quad \text{ for every }j\ge1.
\]

For $j\in\{0,\ldots,m-1\}$, we define the following terms
\begin{align}
V_j^*
:=|f(w)|^{2^{j+1}},
\quad
B_j
:=V_j(k_*)-V_j^*,\quad
T_j^*
:=\norm{G(w)}^2|f(w)|^{2(2^j-1)},\quad
E_j
:=T_j(k_*)-T_j^*,
\label{eqn:svp-target-mass-errors}
\end{align}
where $V_j^*$ and $T_j^*$ are the principal terms
in $V_j(k_*)$ and $T_j(k_*)$, respectively.
Thus $B_j$ and
$E_j$ are the sums of the remaining nonnegative terms. In
particular, $U^Tk_*=w$ and hence $B_0=E_0=0$.

Since $k_*-\theta_*\in P_{m-1}$ and $T_{m-1}$ is
$P_{m-1}$-periodic, we have
$T_{m-1}(k_*)=T_{m-1}(\theta_*)$.
Plugging $j=m-1$ in
\cref{eqn:svp-target-mass-errors,eqn:svp-full-tree-mass-closed-form}
shows that $T_{m-1}(k_*)$ is the sum over all
$\boldsymbol\alpha$, and the term indexed by
$\boldsymbol\alpha^*$ is $T_{m-1}^*$. This gives the identity
\begin{align}
E_{m-1}
=
\sum_{\boldsymbol\alpha\ne\boldsymbol\alpha^*}
\norm{
G(U^Tk_{\boldsymbol\alpha,\theta_*,0})
\prod_{r=1}^{p-1}f(U^Tk_{\boldsymbol\alpha,\theta_*,r})
}^2.
\label{eqn:svp-nonprincipal-mass-identification}
\end{align}

Our goal is to bound $E_{m-1}$. We prove the bound using induction using the recursive formula for $B$ and $E$ with some new terms.
For $j\ge1$, we define
\[
\Delta_j
:=
\sum_{\alpha\in H_j\setminus\{0\}}
V_{j-1}(k_*+\alpha)^2,\qquad
\Gamma_j
:=
\sum_{\alpha\in H_j\setminus\{0\}}
T_{j-1}(k_*+\alpha)V_{j-1}(k_*+\alpha).
\]
Separating the term $\alpha=0$ in
\cref{eqn:svp-frequency-mass-recursion} gives
\begin{align}
V_j(k_*)
&=V_{j-1}(k_*)^2+\Delta_j,\qquad
T_j(k_*)
=T_{j-1}(k_*)V_{j-1}(k_*)+\Gamma_j.
\label{eqn:svp-target-mass-recursion}
\end{align}
Since
$V_j^*=(V_{j-1}^*)^2,
T_j^*=T_{j-1}^*V_{j-1}^*$
by definition, subtracting these terms from
\cref{eqn:svp-target-mass-recursion} gives
\begin{align}
B_j=2V_{j-1}^*B_{j-1}+B_{j-1}^2+\Delta_j,\qquad
E_j
=E_{j-1}(V_{j-1}^*+B_{j-1})
+T_{j-1}^*B_{j-1}+\Gamma_j.
\label{eqn:svp-target-mass-error-recursion}
\end{align}

\begin{lemma}
\label{lem:svp-target-increments}
Let $\lambda=\lambda_1(\cL)\le d\le(1+1/n)\lambda$. Except with
probability $o(1)$ over $U$, simultaneously for every
$j\in\{1,\ldots,m-1\}$,
\begin{align}
\Delta_j
&\le2^{-(1/2+g/4-o(1))n},&
\Gamma_j
&\le
\left(\frac{s^2\lambda}{q}\right)^2
2^{-(1/2+g/4-o(1))n}.
\label{eqn:svp-target-increments}
\end{align}
\end{lemma}
\begin{proof}
Fix $j\in\{1,\ldots,m-1\}$. Under the bijection
$\kappa+P_{j-1}\mapsto U^T\kappa+S_{j-1}$, the terms in
$\Delta_j$ and $\Gamma_j$ correspond to the $S_{j-1}$-cosets
contained in $C_j^*$ other than $C_{j-1}^*$. Conditional on
$S_0,\ldots,S_{j-1}$, each fixed such coset is contained in $C_j^*$
with probability at most $\frac{q^{jn/m}-q^{(j-1)n/m}}{q^n-q^{(j-1)n/m}}
\le2q^{-(n-jn/m)}.$

On the event $w\notin S_{m-1}$, the zero coset is not contained in
$C_j^*$. Since $V_{j-1}(0)\ge1$, linearity of expectation and the definitions in
\cref{eqn:svp-total-frequency-mass-identities} give
\begin{align*}
\E\left[\ind_{\{w\notin S_{m-1}\}}\Delta_j
\mid S_0,\ldots,S_{j-1}\right]
&\le2q^{-(n-jn/m)}(\overline V_j-1) \le 2^{-(1/2+g/2+(m-j)/(2m)-o(1))n},
\notag\\
\E\left[\ind_{\{w\notin S_{m-1}\}}\Gamma_j
\mid S_0,\ldots,S_{j-1}\right]
&\le2q^{-(n-jn/m)}\overline T_j \le \left(\frac{s^2\lambda}{q}\right)^2
2^{-(1/2+g/2+(m-j)/(2m)-o(1))n},
\end{align*}
where we use \cref{lem:svp-total-frequency-masses} and $\log_2q=(m+1)/2+o(1)$ in the final inequalities.
Markov's inequality shows that either bound in \cref{eqn:svp-target-increments} fails for a fixed $j$ with probability at most
\[
2^{-(g/4+(m-j)/(2m)-o(1))n}.
\]
Finally, $\Pr[w\in S_{m-1}]
=
\frac{q^{n-n/m}-1}{q^n-1}
=o(1).$
A union bound shows that all the above events occur except with probability $o(1).$
\end{proof}

\begin{proof}[Proof of \cref{lem:svp-nonprincipal-mass}]
We will bound the term $E_{m-1}$ in \cref{eqn:svp-nonprincipal-mass-identification}. 
By \cref{lem:svp-target-increments}, the conclusion of that lemma holds simultaneously for every $j\in\{1,\ldots,m-1\}$
except with probability $o(1)$.

By \cref{lem:shifted-mass-maximum},
$|f(w)|
=
\frac{\rho_{1/s}(\cL+v/q)}{\rho_{1/s}(\cL)}
\le1,$ and hence $V_j^*=|f(w)|^{2^{j+1}}\le1$.
\cref{lem:odd-local} and $q=2^{o(n)}$ give
\[
\norm{G(w)}
\le
|f(w)|\left(
\frac{2\pi s^2\lambda}{q}
+2^{-\Omega(n)}s^2\lambda
\right)
\le
O\left(\frac{s^2\lambda}{q}\right)|f(w)|.
\]
This gives $T_j^*=
\norm{G(w)}^2|f(w)|^{2(2^j-1)}
\le
O\left(\frac{s^2\lambda}{q}\right)^2
|f(w)|^{2^{j+1}}
\le
O\left(\frac{s^2\lambda}{q}\right)^2.$

We first bound $B_j$. Starting with $B_0=0$, the first recurrence in
\cref{eqn:svp-target-mass-error-recursion} gives inductively
\begin{align}
B_j
\le
4^j2^{-(1/2+g/4-o(1))n} = o(1),
\label{eqn:svp-target-scalar-error-bound}
\end{align}
which can be verified by $2B_{j-1}+B_{j-1}^2
+2^{-(1/2+g/4-o(1))n}
\le
4^j2^{-(1/2+g/4-o(1))n}.$
Since $j\le m=O(\log n)$, the factor $4^j$ is $2^{o(n)}$.
Using this bound in the second recurrence in
\cref{eqn:svp-target-mass-error-recursion} gives, together with 
$V_{j-1}^*+B_{j-1}\le2$ for all sufficiently large $n$ because of \cref{eqn:svp-target-scalar-error-bound} and $V_{j-1}^* = |f(w)|^{2^j} \le 1$,
\[
E_j
\le
\left(\frac{s^2\lambda}{q}\right)^2
2^{-(1/2+g/4-o(1))n}
\]
for every $j\le m-1$.
Iterating the recurrence gives a factor $2^{O(j)}=2^{o(n)}$,
which is absorbed by the $o(1)$ term.

Plugging $j=m-1$ proves \cref{lem:svp-nonprincipal-mass} given $g>0$.
\end{proof}

\paragraph{The second moment.}
For $\boldsymbol i\in[N]^p$ and $\theta\in H_m$, let
\[
\mathcal K_{\boldsymbol i}(\theta)
:=
Q^{p-1}C_{\boldsymbol\delta}(\boldsymbol i)
2\pi iX_{0,i_0}
\omega_q^{\ip{\theta}{(Y_{m-1,0}(\boldsymbol i))_{H_m}}}.
\]
Then we have
\begin{align}
\widehat A_{\boldsymbol\delta}(\theta)
&=\frac1{N^p}\sum_{\boldsymbol i\in[N]^p}\mathcal K_{\boldsymbol i}(\theta),&
\E_X[\mathcal K_{\boldsymbol i}(\theta)\mid U,\boldsymbol\delta]
&=A_{\boldsymbol\delta}(\theta).
\label{eqn:svp-estimator-kernel-mean}
\end{align}
For $\boldsymbol i,\boldsymbol i'\in[N]^p$, let
$I(\boldsymbol i,\boldsymbol i')
:=\{r\in[0,p):i_r=i_r'\}.$
Expanding the squared norm in
\cref{eqn:svp-estimator-kernel-mean} gives
\begin{align}
&\E_{\boldsymbol\delta,X}\left[
\norm{\widehat A_{\boldsymbol\delta}(\theta)-A_{\boldsymbol\delta}(\theta)}^2
\mid U\right]
=
\frac1{N^{2p}}
\sum_{\boldsymbol i,\boldsymbol i'\in[N]^p}
\E_{\boldsymbol\delta,X}\left[
\left(\mathcal K_{\boldsymbol i}(\theta)-A_{\boldsymbol\delta}(\theta)\right)^*
\left(\mathcal K_{\boldsymbol i'}(\theta)-A_{\boldsymbol\delta}(\theta)\right)
\mid U\right].
\label{eqn:svp-estimator-second-moment-expansion}
\end{align}
We use the following bound for every term in this sum, whose proof is deferred to the end of this section.

\begin{lemma}
\label{lem:svp-overlap-covariance}
Let $\lambda=\lambda_1(\cL)\le d\le(1+1/n)\lambda$. Except with probability
$2^{-\Omega(n)}$ over $U$, the following hold for every $\theta\in H_m$
and every $\boldsymbol i,\boldsymbol i'\in[N]^p$: the expectation in
\cref{eqn:svp-estimator-second-moment-expansion} is zero if
$I(\boldsymbol i,\boldsymbol i')=\varnothing$, and otherwise its
absolute value is at most
\begin{align}
\left(\frac{s^2\lambda}{q}\right)^2
Q^{|I(\boldsymbol i,\boldsymbol i')|-1}2^{o(n)}.
\label{eqn:svp-overlap-covariance}
\end{align}
\end{lemma}

\begin{proof}[Proof of \cref{lem:svp-finite-list-second-moment}]
We work on the event in \cref{lem:svp-overlap-covariance}. For every
nonempty $I\subseteq[0,p)$, there are at most $N^{2p-|I|}$ ordered
pairs $(\boldsymbol i,\boldsymbol i')$ satisfying
$I(\boldsymbol i,\boldsymbol i')=I$. Therefore,
\cref{eqn:svp-estimator-second-moment-expansion,eqn:svp-overlap-covariance}
give
\begin{align*}
\E_{\boldsymbol\delta,X}\left[
\norm{\widehat A_{\boldsymbol\delta}(\theta)-A_{\boldsymbol\delta}(\theta)}^2
\mid U\right]
&\le
\frac1{N^{2p}}
\sum_{\varnothing\ne I\subseteq[0,p)}
N^{2p-|I|}
\left(\frac{s^2\lambda}{q}\right)^2
Q^{|I|-1}2^{o(n)}\notag\\
&=
\left(\frac{s^2\lambda}{q}\right)^2Q^{-1}2^{o(n)}
\left[\left(1+\frac QN\right)^p-1\right]
\le
\left(\frac{s^2\lambda}{q}\right)^2Q^{-1}2^{o(n)},
\end{align*}
where the last inequality follows from $N=Qp^4$. The event in
\cref{lem:svp-overlap-covariance} is uniform in $\theta$, so the
conclusion holds simultaneously for every $\theta\in H_m$.
\end{proof}

The first lemma bounds the sums of $V_j$ and $T_j$ over the affine space $\kappa+H_{j+1}$.
The proof uses the same tricks as in \cref{lem:svp-total-frequency-masses}.

\begin{lemma}
\label{lem:svp-fixed-tree-masses}
Let $\lambda=\lambda_1(\cL)\le d\le(1+1/n)\lambda$. Except with
probability $2^{-\Omega(n)}$ over $U$, simultaneously for every
$j\in\{0,\ldots,m-2\}$ and every $\kappa\in\F_q^n$,
\begin{align}
\sum_{\alpha\in H_{j+1}}V_j(\kappa+\alpha)
&\le1+2^{-\Omega(n)},&
\sum_{\alpha\in H_{j+1}}T_j(\kappa+\alpha)
&\le
\left(\frac{s^2\lambda}{q}\right)^22^{o(n)}.
\label{eqn:svp-fixed-tree-masses}
\end{align}
\end{lemma}

\begin{proof}
While the same strategy will apply, we first focus on $j=0$ for the intuition. 
By \cref{lem:svp-frequency-mass-closed-form}, the terms in the statement become
\[
\sum_{\alpha\in H_{1}}V_0(\kappa+\alpha)=
\sum_{\alpha\in H_1}|f(U^T(\kappa+\alpha))|^2
\qquad\text{and}\qquad
\sum_{\alpha\in H_{1}}T_0(\kappa+\alpha)=
\sum_{\alpha\in H_1}\norm{G(U^T(\kappa+\alpha))}^2.
\]
Expanding the first representations in
\cref{eqn:periodic-Gaussian,eqn:periodic-Gaussian-gradient} and using
\cref{eqn:svp-all-square-decompositions} recursively gives
\begin{align}
\sum_{\alpha\in H_1}V_0(\kappa+\alpha)
\le
\frac{\rho_{\sigma_1}(K_1+z)}
{\rho_{1/s}(\cL)^2}
\rho_{\sqrt2q/s}(K_0)
\le
\frac{\rho_{\sigma_1}(K_1)}
{\rho_{1/s}(\cL)^2}
\rho_{\sqrt2q/s}(K_0)
\label{eqn:svp-fixed-tree-base-product}
\end{align}
for some shift $z\in\R^n$, and the last inequality is due to \cref{lem:shifted-mass-maximum}.

We do the same calculation for $j\ge 1.$
By \cref{lem:svp-frequency-mass-closed-form}, the two sums in the
statement are
\[
\sum_{\alpha\in H_{j+1}}V_j(\kappa+\alpha)=\sum_{\alpha\in H_{j+1}}
\sum_{\boldsymbol\beta^{\le j}\in\mathcal A_{\le j}}
\prod_{r=0}^{2^j-1}
\left|
f\left(
U^Tk_{\boldsymbol\beta^{\le j},\kappa+\alpha,r}
\right)
\right|^2
\]
and
\begin{align}
\sum_{\alpha\in H_{j+1}}T_j(\kappa+\alpha)=\sum_{\alpha\in H_{j+1}}
\sum_{\boldsymbol\beta^{\le j}\in\mathcal A_{\le j}}
\norm{
G\left(
U^Tk_{\boldsymbol\beta^{\le j},\kappa+\alpha,0}
\right)
\prod_{r=1}^{2^j-1}
f\left(
U^Tk_{\boldsymbol\beta^{\le j},\kappa+\alpha,r}
\right)
}^2.
\label{eqn:svp-fixed-gradient-sum}
\end{align}

We first expand the scalar sum. For each
$\alpha\in H_{j+1}$,
$\boldsymbol\beta^{\le j}\in\mathcal A_{\le j}$, and
$r\in[0,2^j)$, set
\[
z_r
:=
\varepsilon_rU^T
k_{\boldsymbol\beta^{\le j},\kappa+\alpha,r}.
\]
By the evenness of $f$ and its primal representation,
\[
\left|
f\left(
U^Tk_{\boldsymbol\beta^{\le j},\kappa+\alpha,r}
\right)
\right|^2
=
|f(z_r)|^2
=
\frac{1}{\rho_{1/s}(\cL)^2}
\sum_{u_{r,0},u_{r,1}\in Bz_r+q\cL}
\exp\left(
-\frac{\pi s^2}{q^2}
\left(
\norm{u_{r,0}}^2+\norm{u_{r,1}}^2
\right)
\right).
\]
Expanding the product therefore is decomposed by the summand indexed by $u_{r,0},u_{r,1}\in Bz_r+q\cL$ for $r\in[0,2^j)$, that is, one with the Gaussian factor
\[
\frac{1}{\rho_{1/s}(\cL)^{2^{j+1}}}
\exp\left(
-\frac{\pi s^2}{q^2}
\sum_{r=0}^{2^j-1}
\left(
\norm{u_{r,0}}^2+\norm{u_{r,1}}^2
\right)
\right).
\]
For the corresponding average-difference variables,
\cref{eqn:svp-all-square-decompositions} gives
\[
\sum_{r=0}^{2^j-1}
\left(
\norm{u_{r,0}}^2+\norm{u_{r,1}}^2
\right)
=
2^{j+1}\norm{a_{j,0}}^2
+
\frac12\sum_{r=0}^{2^j-1}\norm{d_{0,r}}^2
+
\sum_{h=1}^j
2^{h-1}
\sum_{r=0}^{2^{j-h}-1}\norm{d_{h,r}}^2.
\]
The Gaussian widths in the table below follow from these
coefficients.

The proofs of
\cref{clm:svp-expanded-sum-indexing,clm:svp-average-difference}
give the following variants for the present fixed-$\kappa$ sum.

\begin{claim}
\label{clm:svp-fixed-tree-indexing}
Fix $\kappa\in\F_q^n$. In the expansion of
$\sum_{\alpha\in H_{j+1}}V_j(\kappa+\alpha),$
the tuple $(u_{r,b})_{r,b}$ uniquely determines
$(\alpha,\boldsymbol\beta^{\le j})$. Consequently, the expanded sum
can be indexed by $(u_{r,b})_{r,b}$ without multiplicity.
\end{claim}
\begin{proof}[Sketch of proof.]
As in the proof of \cref{clm:svp-expanded-sum-indexing}, the tuple
$(u_{r,b})_{r,b}$ determines every
$k_{\boldsymbol\beta^{\le j},\kappa+\alpha,r}$. Since $\kappa$ is
fixed, the direct-sum decomposition gives $\alpha=
\left(k_{\boldsymbol\beta^{\le j},\kappa+\alpha,0}-\kappa\right)_{H_{j+1}},$
and, for $h\in[j]$ and $a\in[0,2^{j-h})$, we also have
$\beta_{h,a}=
\left(k_{\boldsymbol\beta^{\le j},\kappa+\alpha,2^ha}
-\varepsilon_{2^ha}\kappa\right)_{H_h}.$
Thus $(u_{r,b})_{r,b}$ uniquely determines
$(\alpha,\boldsymbol\beta^{\le j})$.
\end{proof}

\begin{claim}
\label{clm:svp-fixed-tree-average-difference}
Fix $\kappa\in\F_q^n$. The average-difference transformation is a
bijection between the tuples $(u_{r,b})_{r,b}$ occurring in the
expansion of
$\sum_{\alpha\in H_{j+1}}V_j(\kappa+\alpha)$
and the transformed tuples arising from them. These transformed
tuples satisfy the following properties:
\begin{itemize}[nosep]
\item After all difference variables are fixed, the possible values
of $a_{j,0}$ are contained in a single coset of $K_{j+1}$.
\item Every $d_{0,r}$ belongs to $K_0$.
\item For every $h\in[j]$ and $r\in[0,2^{j-h})$, after all
lower-level differences $d_{\ell,r'}$ with $\ell<h$ are fixed, the
possible values of $d_{h,r}$ are contained in a single coset of
$K_{h-1}$.
\end{itemize}
Each coset may depend on the fixed variables.
\end{claim}
\begin{proof}[Sketch of proof.]
Bijectivity follows from the bijectivity of $\Psi_j$. Suppose that
all difference variables are fixed and that the final average changes
by $x$. Then every $u_{r,b}$ changes by the same vector $x\in\cL$.
Comparing the residues at $r=0$ gives
\[
B^{-1}x\bmod q
\in U^T(H_1\oplus\cdots\oplus H_{j+1})=S_{j+1},
\]
and hence $x\in K_{j+1}$. The assertions for the difference variables
follow exactly as in \cref{clm:svp-average-difference}; the additional
$\alpha\in H_{j+1}$ cancels from the two child blocks.
\end{proof}

By the two claims above, the expanded sum can be indexed by the
transformed variables without multiplicity. Their numbers, containing
cosets, and Gaussian widths are summarized below. For $a_{j,0}$, all
difference variables are fixed. For $d_{h,r}$ with $h\ge1$, all
lower-level differences are fixed.
\[
\begin{array}{c|c|c|c}
\text{variable}
&
\text{number}
&
\text{containing coset}
&
\text{width}
\\ \hline
a_{j,0}
&
1
&
\text{a coset of }K_{j+1}
&
\dfrac{q}{2^{(j+1)/2}s}=\sigma_{j+1}
\\[2mm]
d_{0,r}
&
2^j
&
K_0
&
\dfrac{\sqrt2q}{s}
\\[2mm]
d_{1,r}
&
2^{j-1}
&
\text{a coset of }K_0
&
\dfrac{q}{s}
\\[2mm]
d_{h+1,r}
&
2^{j-h-1}
&
\text{a coset of }K_h
&
\dfrac{q}{2^{h/2}s}=\sigma_h
\quad(1\le h<j).
\end{array}
\]

We first sum over $a_{j,0}$ while all difference variables are fixed.
The first item of
\cref{clm:svp-fixed-tree-average-difference} and
\cref{lem:shifted-mass-maximum} bound this sum by $
\rho_{\sigma_{j+1}}(K_{j+1}).$
We then sum over the variables $d_{h,r}$ in decreasing order of $h$.
When $d_{h,r}$ is summed, all lower-level differences remain fixed,
so the third item of
\cref{clm:svp-fixed-tree-average-difference} applies. Finally, every
$d_{0,r}$ belongs to $K_0$. Applying
\cref{lem:shifted-mass-maximum} at each step gives
\begin{align}
\sum_{\alpha\in H_{j+1}}V_j(\kappa+\alpha)
&\le
\frac{\rho_{\sigma_{j+1}}(K_{j+1})}
{\rho_{1/s}(\cL)^{2^{j+1}}}
\rho_{\sqrt2q/s}(K_0)^{2^j}
\rho_{q/s}(K_0)^{2^{j-1}}
\prod_{h=1}^{j-1}
\rho_{\sigma_h}(K_h)^{2^{j-h-1}}.
\label{eqn:svp-fixed-tree-product}
\end{align}

By \cref{lem:svp-random-lattice-masses},
\[
\rho_{\sigma_h}(K_h)
=
1+\rho_{\sigma_h}(K_h\setminus\{0\})
\le
1+2^{-(g/2-o(1))n}
\]
for every $h\in\{1,\ldots,j+1\}$. 
Moreover, since $K_0=q\cL$, we have
\begin{align*}
\rho_{\sqrt2q/s}(K_0)
&=\rho_{\sqrt2/s}(\cL)
\le1+2^{-(g-o(1))n},\notag\\
\rho_{q/s}(K_0)
&=\rho_{1/s}(\cL)
\le1+2^{-(1/2+g-o(1))n}
\le1+2^{-(g-o(1))n},
\end{align*}
by \cref{lem:shell,cor:working-scale}.
Since the denominator is at
least one and the total multiplicity of the numerator factors is at
most $2^{j+1}\le2p=2^{o(n)}$, the right-hand sides of
\cref{eqn:svp-fixed-tree-base-product,eqn:svp-fixed-tree-product}
are at most
\[
\left(1+2^{-(g/2-o(1))n}\right)^{2^{o(n)}}
=
1+2^{-\Omega(n)}.
\]

For \cref{eqn:svp-fixed-gradient-sum}, let $z_r:=U^Tk_{\boldsymbol\beta^{\le j},\kappa+\alpha,r}$ to simplify the expression.
The first representation of $G$ gives
\[
\norm{G(z_0)}^2
=
\left(\frac{2\pi s^2}
{q\rho_{1/s}(\cL)}\right)^2
\sum_{u_{0,0},u_{0,1}\in Bz_0+q\cL}
\ip{u_{0,0}}{u_{0,1}}
e^{-\pi s^2
(\norm{u_{0,0}}^2+\norm{u_{0,1}}^2)/q^2}.
\]
Here $u_{0,0}$ and $u_{0,1}$ are the variables from the two copies of $G(z_0)$ in the norm.

The remainder of the argument is the same as in the proof of the
second bound in \cref{lem:svp-total-frequency-masses}. 
We first
bound the inner product by
\[
\left|\ip{u_{0,0}}{u_{0,1}}\right|
\le
\frac12\left(
\norm{u_{0,0}}^2+\norm{u_{0,1}}^2
\right)
\le
\frac12
\sum_{r=0}^{2^j-1}
\left(
\norm{u_{r,0}}^2+\norm{u_{r,1}}^2
\right).
\]
Then, we apply the following inequality followed by \cref{eqn:svp-all-square-decompositions} to the exponent:
\[
\frac{R}{2}e^{-\pi s^2R/q^2}
\le
\frac{nq^2}{2\pi s^2}
e^{-\pi(1-1/n)s^2R/q^2}
\le
\frac{nq^2}{\pi s^2}
e^{-\pi(1-1/n)s^2R/q^2}.
\]
for $R=\sum_{r=0}^{2^j-1}\left(
\norm{u_{r,0}}^2+\norm{u_{r,1}}^2
\right).$ 
For $j=0$, the same calculation gives the right-hand side of
\cref{eqn:svp-fixed-tree-base-product} with $\sigma_1$ and
$\sqrt2q/s$ replaced by
$\sigma_1/\sqrt{1-1/n}$ and
$\sqrt2q/(s\sqrt{1-1/n})$, respectively, and multiplied by
$\left(2\pi s^2/q\right)^2nq^2/(\pi s^2)$.

For $j\ge1$, summing the resulting Gaussian
terms and applying \cref{lem:shifted-mass-maximum} gives
\begin{align}
\sum_{\alpha\in H_{j+1}}T_j(\kappa+\alpha)
&\le
\left(\frac{2\pi s^2}{q}\right)^2
\frac{nq^2}{\pi s^2}
\frac{
\rho_{\sigma_{j+1}/\sqrt{1-1/n}}(K_{j+1})
}{
\rho_{1/s}(\cL)^{2^{j+1}}
}
\notag\\
&\quad\cdot
\rho_{\sqrt2q/(s\sqrt{1-1/n})}(K_0)^{2^j}
\rho_{q/(s\sqrt{1-1/n})}(K_0)^{2^{j-1}}
\prod_{h=1}^{j-1}
\rho_{\sigma_h/\sqrt{1-1/n}}(K_h)^{2^{j-h-1}}.
\label{eqn:svp-fixed-gradient-mass-product}
\end{align}
The second estimate in \cref{lem:svp-random-lattice-masses} bounds
the factors involving
$\rho_{\sigma_h/\sqrt{1-1/n}}(K_h)$. Moreover, given $K_0=q\cL$, we have $\rho_{\sqrt2q/(s\sqrt{1-1/n})}(K_0)
=
\rho_{\sqrt2/(s\sqrt{1-1/n})}(\cL),$
which is bounded by \cref{lem:shell}. Hence the product of Gaussian
masses in \cref{eqn:svp-fixed-gradient-mass-product}, including its
denominator, is at most $2^{o(n)}$. Consequently,
\[
\sum_{\alpha\in H_{j+1}}T_j(\kappa+\alpha)\le
\left(\frac{2\pi s^2}{q}\right)^2
\frac{nq^2}{\pi s^2}2^{o(n)}=
\left(\frac{s^2\lambda}{q}\right)^2
\frac{4\pi nq^2}{s^2\lambda^2}2^{o(n)}\le
\left(\frac{s^2\lambda}{q}\right)^22^{o(n)},
\]
where we use 
$s^2\lambda^2=\Theta(n)$ and $q=n^{O(1)}$ in the last inequality.
This proves the second
bound in \cref{eqn:svp-fixed-tree-masses}.
\end{proof}

The next lemma bounds the same sums after the factors indexed by a
nonempty set $I$ have been removed.
We introduce some notations. 
For $j\in\{0,\ldots,m-1\}$, $\kappa\in\F_q^n$, and
$I\subseteq[0,2^j)$, let
\[
V_{j,I}(\kappa)
:=
\sum_{\boldsymbol\alpha^{\le j}\in\mathcal A_{\le j}}
\prod_{r\in[0,2^j)\setminus I}
\left|
f\left(
U^Tk_{\boldsymbol\alpha^{\le j},\kappa,r}
\right)
\right|^2.
\]
We similarly define the notation for $T$, where we require $0\notin I$ because the term indexed by $0$ is $G(\cdot)$ and is retained in $T$. For $0\notin I$, we define
\[
T_{j,I}(\kappa)
:=
\sum_{\boldsymbol\alpha^{\le j}\in\mathcal A_{\le j}}
\norm{
G\left(
U^Tk_{\boldsymbol\alpha^{\le j},\kappa,0}
\right)
\prod_{r\in[1,2^j)\setminus I}
f\left(
U^Tk_{\boldsymbol\alpha^{\le j},\kappa,r}
\right)
}^2.
\]
In particular, $V_{j,\varnothing}=V_j$ and
$T_{j,\varnothing}=T_j$.

\begin{lemma}
\label{lem:svp-omitted-factor-masses}
Let $\lambda=\lambda_1(\cL)\le d\le(1+1/n)\lambda$. On the event in
\cref{lem:svp-fixed-tree-masses}, simultaneously for every
$j\in\{0,\ldots,m-1\}$, every $\kappa\in\F_q^n$, and every nonempty
$I\subseteq[0,2^j)$,
\begin{align}
V_{j,I}(\kappa)
&\le
Q^{|I|-1}2^{o(n)},&
T_{j,I}(\kappa)
&\le
\left(\frac{s^2\lambda}{q}\right)^2
Q^{|I|-1}2^{o(n)}
\quad\text{if }0\notin I.
\label{eqn:svp-omitted-factor-masses}
\end{align}
\end{lemma}

\begin{proof}
We use induction on $j$. For $j=0$, the only nonempty set is
$I=\{0\}$, and $V_{0,\{0\}}(\kappa)=1$. The second bound does not apply, and will not be used in the later induction.

Suppose that $j\ge1$, and let $I_L:=I\cap[0,2^{j-1}),I_R:=\{r\in[0,2^{j-1}):2^{j-1}+r\in I\}.$
By the decomposition used in \cref{lem:svp-frequency-mass-closed-form} and the inductive notions in
\cref{eqn:svp-left-right-alpha,eqn:svp-left-right-k}, we have
\[
V_{j,I}(\kappa)=
\sum_{\alpha\in H_j}
V_{j-1,I_L}(\kappa+\alpha)
V_{j-1,I_R}(-\kappa-\alpha),
\qquad
T_{j,I}(\kappa)
=
\sum_{\alpha\in H_j}
T_{j-1,I_L}(\kappa+\alpha)
V_{j-1,I_R}(-\kappa-\alpha),
\]
where the second identity is used only when $0\notin I$.

We first prove the bound for $V_{j,I}(\kappa)$. If exactly one of $I_L,I_R$ is
nonempty, the induction hypothesis bounds the corresponding factor for each $\alpha$, while the first bound in
\cref{lem:svp-fixed-tree-masses} bounds the sum of the other factor
over $\alpha\in H_j$. 
If both sets are nonempty, the induction
hypothesis bounds both factors for each $\alpha$, and summing over
$|H_j|=Q$ choices of $\alpha$ gives
\[
Q\cdot Q^{|I_L|-1}Q^{|I_R|-1}
=
Q^{|I|-1}.
\]
This proves the first inequality of \cref{eqn:svp-omitted-factor-masses}.

Now suppose that $0\notin I$, and hence $0\notin I_L$. If
$I_L\ne\varnothing$ and
$I_R=\varnothing$, the induction hypothesis bounds the first factor
in the second recurrence, and the first bound in
\cref{lem:svp-fixed-tree-masses} bounds the sum of the second factor.
If $I_L=\varnothing$ and $I_R\ne\varnothing$, the second bound in
\cref{lem:svp-fixed-tree-masses} bounds the sum of the first factor,
and the first inequality of \cref{eqn:svp-omitted-factor-masses} bounds the second factor
for each $\alpha$. If both sets are nonempty, the induction hypotheses bound
both factors for each $\alpha$, and summing over $\alpha\in H_j$ gives
\[
Q\cdot
\left(\frac{s^2\lambda}{q}\right)^2
Q^{|I_L|-1}Q^{|I_R|-1}
=
\left(\frac{s^2\lambda}{q}\right)^2Q^{|I|-1}.
\]
This proves the second inequality of
\cref{eqn:svp-omitted-factor-masses}. 

The suppressed factor $2^{o(n)}$ should be carefully proven. Observe that in the scalar induction we use the first
bound in \cref{lem:svp-fixed-tree-masses} at most $2^j\le p=2^{o(n)}$ times.
Their product is
$(1+2^{-\Omega(n)})^p=1+2^{-\Omega(n)}$ given $p=2^{o(n)}$.
In the induction for $T_{j,I}$, the second bound in
\cref{lem:svp-fixed-tree-masses} is used at most once: when
$I_L=\varnothing$, it directly bounds the sum of
$T_{j-1}(\kappa+\alpha)$ over $\alpha\in H_j$, so no further
induction involving $T$ is needed.
All other factors are the factors discussed in the bound of the first term. Thus the total suppressed factor is $2^{o(n)}$. This proves the lemma.
\end{proof}

Now we are ready to prove, using the above lemmas, \cref{lem:svp-overlap-covariance} which states that
\begin{align}
    \label{eqn:svp-overlap-target}
    \left|
\E_{\boldsymbol\delta,X}\left[
\left(\mathcal K_{\boldsymbol i}(\theta)-A_{\boldsymbol\delta}(\theta)\right)^*
\left(\mathcal K_{\boldsymbol i'}(\theta)-A_{\boldsymbol\delta}(\theta)\right)
\mid U
\right]
\right|
\le
\left(\frac{s^2\lambda}{q}\right)^2
Q^{|I(\boldsymbol i,\boldsymbol i')|-1}2^{o(n)}.
\end{align}
for $I(\boldsymbol i,\boldsymbol i')\neq \varnothing$ and it being zero otherwise, where
\[
\mathcal K_{\boldsymbol i}(\theta)
:=
Q^{p-1}C_{\boldsymbol\delta}(\boldsymbol i)
2\pi iX_{0,i_0}
\omega_q^{\ip{\theta}{(Y_{m-1,0}(\boldsymbol i))_{H_m}}}.
\]

\begin{proof}[Proof of \cref{lem:svp-overlap-covariance}]
We assume that all of the bounds in
\cref{lem:svp-total-frequency-masses,lem:svp-fixed-tree-masses} hold,
which happens except with probability $2^{-\Omega(n)}$ over $U$. Fix
$\theta\in H_m$ and $\boldsymbol i,\boldsymbol i'\in[N]^p$, and let
$I:=I(\boldsymbol i,\boldsymbol i')$.

We first consider the case $I=\varnothing$.
In this case, $i_r\ne i'_r$ for every $r\in[0,p)$.
For fixed $U$ and $\boldsymbol\delta$, the definitions of
$C_{\boldsymbol\delta}(\boldsymbol i)$ and
$Y_{m-1,0}(\boldsymbol i)$ show that
$\mathcal K_{\boldsymbol i}(\theta)$ is determined by
$(X_{r,i_r})_{r\in[0,p)}$. Similarly,
$\mathcal K_{\boldsymbol i'}(\theta)$ is determined by
$(X_{r,i'_r})_{r\in[0,p)}$. These two collections of samples are
independent.
Both random variables have conditional expectation
$A_{\boldsymbol\delta}(\theta)$. Their independence therefore gives
\[
\E_X\left[
\left(\mathcal K_{\boldsymbol i}(\theta)
-A_{\boldsymbol\delta}(\theta)\right)^*
\left(\mathcal K_{\boldsymbol i'}(\theta)
-A_{\boldsymbol\delta}(\theta)\right)
\mid U,\boldsymbol\delta
\right]=0
\]
Averaging over $\boldsymbol\delta$ proves the claim for
$I=\varnothing$.

Suppose that $I\ne\varnothing$. We apply
\cref{eqn:svp-all-characters} to the entire scalar term
\[
Q^{p-1}C_{\boldsymbol\delta}(\boldsymbol i)
\omega_q^{\ip{\theta}{(Y_{m-1,0}(\boldsymbol i))_{H_m}}}
\]
and similarly to the term indexed by $\boldsymbol i'$. For each
$\boldsymbol\alpha$, we index the terms in the resulting summand by
\[
\omega_q^{\ip{k_{\boldsymbol\alpha,\theta,r}}{Y_{r,i_r}}}
\]
by $r$ for every $r\in[0,p)$.
The term with $r=0$ is additionally
multiplied by $2\pi iX_{0,i_0}$.

Averaging the product over $\boldsymbol\delta$, character
orthogonality eliminates all terms indexed by
$\boldsymbol\alpha\ne\boldsymbol\alpha'$ and only the terms with
$\boldsymbol\alpha=\boldsymbol\alpha'$ remain.
After averaging over $\boldsymbol\delta$, fix a remaining
$\boldsymbol\alpha$ and $r\in[1,p)$. The terms involving the
$r$-th samples in
$\mathcal K_{\boldsymbol i}(\theta)^*
\mathcal K_{\boldsymbol i'}(\theta)$ are
\[
\overline{
\omega_q^{\ip{k_{\boldsymbol\alpha,\theta,r}}{Y_{r,i_r}}}
}
\omega_q^{\ip{k_{\boldsymbol\alpha,\theta,r}}{Y_{r,i'_r}}}.
\]
If $r\notin I$, then $i_r\ne i'_r$, so $X_{r,i_r}$ and
$X_{r,i'_r}$, and hence $Y_{r,i_r}$ and $Y_{r,i'_r}$, are
independent. Taking the expectation over these two samples gives
\[
\overline{
f(U^Tk_{\boldsymbol\alpha,\theta,r})
}
f(U^Tk_{\boldsymbol\alpha,\theta,r})
=
|f(U^Tk_{\boldsymbol\alpha,\theta,r})|^2.
\]
If $r\in I$, then $i_r=i'_r$, so the two character terms are
evaluated at the same $Y_{r,i_r}$ and their product is
\[
\overline{
\omega_q^{\ip{k_{\boldsymbol\alpha,\theta,r}}{Y_{r,i_r}}}
}
\omega_q^{\ip{k_{\boldsymbol\alpha,\theta,r}}{Y_{r,i_r}}}
=1.
\]
If $r=0$, the resulting term is
$\norm{G(U^Tk_{\boldsymbol\alpha,\theta,0})}^2$ if $0\notin I$ and
$\E_{X\sim D_{\cL^*,s}}[\norm{2\pi X}^2]$ if $0\in I$. Consequently,
\begin{align}
\E_{\boldsymbol\delta,X}\left[
\mathcal K_{\boldsymbol i}(\theta)^*
\mathcal K_{\boldsymbol i'}(\theta)\mid U\right]
=
\begin{cases}
T_{m-1,I}(\theta),&0\notin I,\\
\E_{X\sim D_{\cL^*,s}}[\norm{2\pi X}^2]
V_{m-1,I}(\theta),&0\in I.
\end{cases}
\label{eqn:svp-overlap-joint-moment}
\end{align}
By \cref{lem:dgs-tail}, $s^2\lambda^2=\Theta(n)$, and
$q=n^{O(1)}$, we have
\[\E_{X\sim D_{\cL^*,s}}[\norm{2\pi X}^2]
\le O(ns^2)\le(s^2\lambda/q)^22^{o(n)}.\]
Therefore, \cref{lem:svp-omitted-factor-masses} bounds the absolute
value of \cref{eqn:svp-overlap-joint-moment} by
$(s^2\lambda/q)^2Q^{|I|-1}2^{o(n)}$.

Applying the same character expansion to two independent copies
gives
\begin{align}
\E_{\boldsymbol\delta}\left[
\norm{A_{\boldsymbol\delta}(\theta)}^2\mid U\right]
=T_{m-1}(\theta)
\le
\sum_{\kappa\in R_{m-1}}T_{m-1}(\kappa)
\le
\left(\frac{s^2\lambda}{q}\right)^2
2^{-(g/2-o(1))n},
\label{eqn:svp-overlap-mean-product}
\end{align}
where we used $R_{m-1}=H_m$ and
\cref{lem:svp-total-frequency-masses}. Since the conditional
expectation of every $\mathcal K_{\boldsymbol i}(\theta)$ is
$A_{\boldsymbol\delta}(\theta)$, the expectation in the statement is
the difference between \cref{eqn:svp-overlap-joint-moment} and
\cref{eqn:svp-overlap-mean-product}. Since $I\ne\varnothing$, the
triangle inequality proves the desired bound, after absorbing a
factor of two into $2^{o(n)}$.
\end{proof}

%% file: fig_analysis.tex
\begin{figure}[!tp]
\centering
\begin{adjustbox}{
  max width=0.9\linewidth,
  max totalheight=0.9\textheight,
  center
}
\begin{tikzpicture}[
  >=Latex,
  side/.style={
    draw,
    rounded corners,
    align=center,
    inner xsep=7pt,
    inner ysep=6pt,
    text width=6.3cm,
    font=\small
  },
  center/.style={
    draw,
    rounded corners,
    align=center,
    inner xsep=8pt,
    inner ysep=7pt,
    text width=8.2cm,
    font=\small
  },
  arrow/.style={
    -{Latex[length=2.8mm,width=2mm]},
    line width=0.9pt,
    shorten <=1.5pt,
    shorten >=1.5pt
  },
  flowlabel/.style={
    draw=none,
    fill=none,
    inner sep=0pt,
    font=\footnotesize\itshape
  },
  ref/.style={
    draw=none,
    fill=none,
    inner sep=0pt,
    font=\footnotesize
  }
]

\node[center] (subtree) at (0,0) {
\textbf{Recursive subtree masses}
({\footnotesize\cref{lem:svp-frequency-mass-closed-form}})\\[2mm]
$\displaystyle
V_0(\kappa)=|f(U^T\kappa)|^2,
\qquad
T_0(\kappa)=\norm{G(U^T\kappa)}^2
$\\[2mm]
$\displaystyle
V_j(\kappa)
=
\sum_{\alpha\in H_j}
V_{j-1}(\kappa+\alpha)^2
$\\[1mm]
$\displaystyle
T_j(\kappa)
=
\sum_{\alpha\in H_j}
T_{j-1}(\kappa+\alpha)
V_{j-1}(\kappa+\alpha)
$
};

\node[side,anchor=north] (global)
at ($(subtree.south)+(4.3cm,-11mm)$) {
\textbf{Global masses}
({\footnotesize\cref{lem:svp-total-frequency-masses}})\\[2mm]
$\displaystyle
\overline V_j
:=
\sum_{\kappa\in R_j}V_j(\kappa)
=
\sum_{\kappa\in R_{j-1}}
V_{j-1}(\kappa)^2
$\\[2mm]
$\displaystyle
\overline T_j
:=
\sum_{\kappa\in R_j}T_j(\kappa)
=
\sum_{\kappa\in R_{j-1}}
T_{j-1}(\kappa)V_{j-1}(\kappa)
$
};

\node[side,below=11mm of global] (globalbound) {
\textbf{Average-difference bounds}\\[0.5mm]
({\footnotesize
\cref{lem:svp-random-lattice-masses,lem:svp-total-frequency-masses}})\\[2mm]
$\displaystyle
\overline V_j-1
\le
2^{((m-j-1)/2-g/2+o(1))n}
$\\[2mm]
$\displaystyle
\overline T_j
\le
\left(\frac{s^2\lambda}{q}\right)^2
2^{((m-j-1)/2-g/2+o(1))n}
$
};

\node[side,below=11mm of globalbound] (increment) {
\textbf{New incorrect branches}
({\footnotesize\cref{lem:svp-target-increments}})\\[2mm]
$\displaystyle
\Delta_j
:=
\sum_{\alpha\in H_j\setminus\{0\}}
V_{j-1}(k_*+\alpha)^2
$\\[2mm]
$\displaystyle
\Gamma_j
:=
\sum_{\alpha\in H_j\setminus\{0\}}
T_{j-1}(k_*+\alpha)
V_{j-1}(k_*+\alpha)
$\\[2mm]
$\displaystyle
\Delta_j
\le
2^{-(1/2+g/4-o(1))n}
$\\[-1mm]
$\displaystyle
\Gamma_j
\le
\left(\frac{s^2\lambda}{q}\right)^2
2^{-(1/2+g/4-o(1))n}
$
};

\node[side,anchor=east] (target)
at ($(increment.west)+(-18mm,+50mm)$) {
\textbf{Target mass decomposition}
({\footnotesize\cref{eqn:svp-target-mass-errors}})\\[2mm]
$\displaystyle
V_j^*
=
|f(w)|^{2^{j+1}}
$\\[1mm]
$\displaystyle
T_j^*
=
\norm{G(w)}^2|f(w)|^{2(2^j-1)}
$\\[2mm]
$\displaystyle
V_j(k_*)=V_j^*+B_j,
\qquad
T_j(k_*)=T_j^*+E_j
$\\[2mm]
$V_j^*,T_j^*$: principal terms\\
$B_j,E_j$: nonprincipal masses
};

\coordinate (merge)
at ($(target.south)!0.5!(increment.south)$);

\node[center,anchor=north] (recurrence)
at ($(merge)+(0,-40mm)$) {
\textbf{Target-error recursion}
{(\footnotesize\cref{eqn:svp-target-mass-error-recursion})}\\[2mm]
$\displaystyle B_0=E_0=0$\\[2mm]
$\displaystyle
B_j
=
2V_{j-1}^*B_{j-1}
+B_{j-1}^2+\Delta_j
$\\[2mm]
$\displaystyle
E_j
=
E_{j-1}(V_{j-1}^*+B_{j-1})
+T_{j-1}^*B_{j-1}
+\Gamma_j
$
};

\node[center,below=11mm of recurrence] (final) {
\textbf{Isolation of the desired term}
({\footnotesize\cref{lem:svp-nonprincipal-mass}})\\[2mm]
$\displaystyle
T_{m-1}^*
=
\norm{G(w)f(w)^{p-1}}^2
$\\[2mm]
$\displaystyle
E_{m-1}
=
\E_{\boldsymbol\delta}
\left[
\norm{
A_{\boldsymbol\delta}(\theta_*)
-
c_{\boldsymbol\delta}(k_*)
G(w)f(w)^{p-1}
}^2
\right]
$\\[2mm]
$\displaystyle
E_{m-1}
\le
\left(\frac{s^2\lambda}{q}\right)^2
2^{-(1/2+g/4-o(1))n}
$
};

\draw[arrow]
(subtree.south west)
-| (target.north);

\coordinate (fixlabel)
at ($(subtree.south west)!0.48!(target.north)$);

\node[flowlabel,anchor=east]
at ($(fixlabel)+(-5mm,0)$)
{fix the target $k_*$};

\draw[arrow]
(subtree.south east)
-| (global.north);

\coordinate (sumlabel)
at ($(subtree.south east)!0.52!(global.north)$);

\node[flowlabel,anchor=west]
at ($(sumlabel)+(5mm,0)$)
{sum over $\kappa\in R_j$};

\draw[arrow]
(global.south)
--
(globalbound.north);

\node[flowlabel,anchor=west]
at ($(global.south)!0.5!(globalbound.north)+(5mm,0)$)
{average-difference};

\draw[arrow]
(globalbound.south)
--
(increment.north);

\node[flowlabel,anchor=west]
at ($(globalbound.south)!0.5!(increment.north)+(5mm,0)$)
{random-coset inclusion};

\draw[arrow]
(target.south)
-| (recurrence.north west);

\draw[arrow]
(increment.south)
-| (recurrence.north east);

\draw[arrow]
(recurrence.south)
--
(final.north);

\end{tikzpicture}
\end{adjustbox}

\caption{The outline of the proof of the nonprincipal-mass analysis \cref{lem:svp-nonprincipal-mass}. We need a similar but more involved proof for \cref{lem:svp-finite-list-second-moment}.}
\label{fig:svp-mass-analysis-flow}
\end{figure}
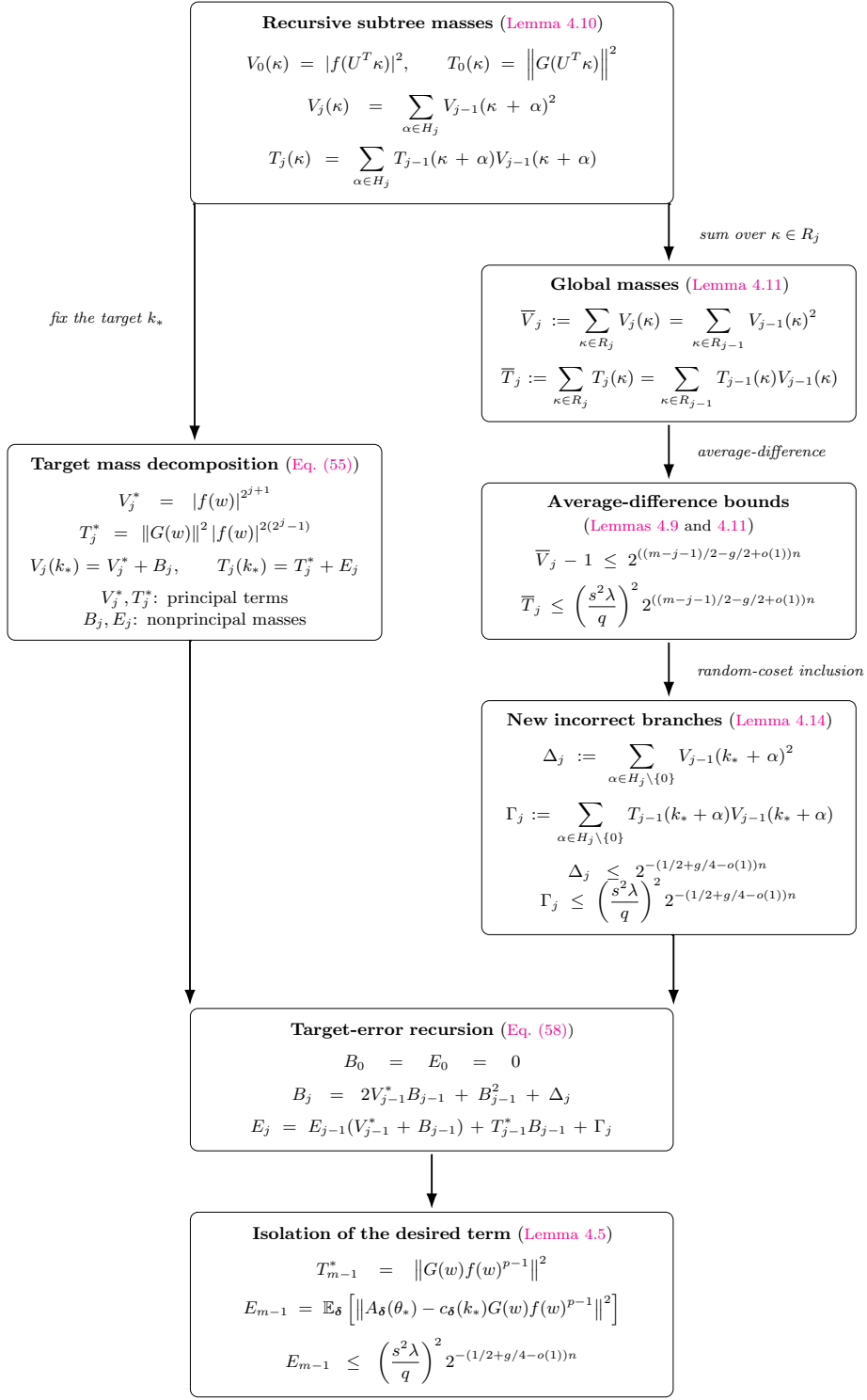

%% file: 5.cvp.tex
\section{CVP with a distance guarantee}
\label{sec:cvp}

This section proves the following result.
\begin{theorem}
\label{thm:bounded-cvp}
For every fixed constant $0< \alpha < \frac{1}{2\sqrt{t_0}}=1.03925\ldots$,
there is a randomized classical algorithm that solves $\CVP$ with probability at least $2/3$ given a lattice $\cL$ and a target $y\in\R^n$ promised that $\dist(y,\cL)\le\alpha\lambda_1(\cL)$ in worst-case time and space $2^{n/2+o(n)}$.
\end{theorem}

We use $y$ as the target vector and $\alpha < \frac1{2\sqrt{t_0}}$ in this section throughout.
We will use the same parameters as the previous section. Formally, $p$ is the largest power of two satisfying $p\le n/\log^3n$ and $q$ is the smallest prime larger than $2\sqrt p$. Let $m=\log_2 p+1$, $Q=q^{n/m}=2^{n/2 + o(n)}$ and $N=Qp^4 = 2^{n/2 + o(n)}$.

\subsection{The first estimation target}

Let $v\in\cL$ be closest to $y$ and let $e=v-y$.
We define the shifted versions of $f$ and $G$ as follows for $w\in\F_q^n$:
\begin{align*}
f_y(w)&:=F_{1/s}\left(\frac{Bw-y}{q}\right)=\E_{X\sim D_{\cL^*,s}}\left[\omega_q^{\ip{w}{[X]_q}}e^{-2\pi i\ip{X}{y}/q}\right],\notag\\
G_y(w)&:=\nabla F_{1/s}\left(\frac{Bw-y}{q}\right)=2\pi i\E_{X\sim D_{\cL^*,s}}\left[X\omega_q^{\ip{w}{[X]_q}}e^{-2\pi i\ip{X}{y}/q}\right].
\end{align*}
These functions are well defined on $\F_q^n$ because
$F_{1/s}$ and its gradient are $\cL$-periodic.
The next lemma identifies the vector that we estimate.

\begin{lemma}
\label{lem:cvp-local}
Let $\lambda=\lambda_1(\cL)\le d\le(1+1/n)\lambda$, let $v\in\cL$ be closest to $y$, let $e=v-y$, and let $w=B^{-1}v\bmod q$. 
Let $t_0<t<1/4.$
Then, it holds that
\[
f_y(w)^p=2^{-(t+o(1))n\norm e^2/d^2}(1+2^{-\Omega(n)}),\qquad
\norm{-\frac{G_y(w)}{f_y(w)}-\frac{2\pi s^2}{q}e}\le2^{-\Omega(n)}s^2\lambda.
\]
In particular, if $\norm e\le\alpha\lambda$, then $f_y(w)^p
\ge
2^{-(t\alpha^2+o(1))n}.$
\end{lemma}

\begin{proof}
Since $Bw-v\in q\cL$ and $F_{1/s}$ is $\cL$-periodic, the primal representations of $f_y(w)$ and $G_y(w)$ are
\begin{align}
f_y(w)=\frac{1}{\rho_{1/s}(\cL)}\sum_{x\in\cL}\rho_{1/s}\left(x+\frac eq\right),\qquad
G_y(w)=-\frac{2\pi s^2}{\rho_{1/s}(\cL)}\sum_{x\in\cL}\left(x+\frac eq\right)\rho_{1/s}\left(x+\frac eq\right).
\label{eqn:cvp-local-primal}
\end{align}The choice of $v$ gives $\norm{x+e}\ge\norm e$ for every $x\in\cL$. Consequently,
\begin{align}
\norm{x+e/q}^2-\norm{e/q}^2=\norm x^2+\frac2q\ip{x}{e}\ge\left(1-\frac1q\right)\norm x^2.
\label{eqn:cvp-local-gap}
\end{align}
Since $\rho_{1/s}(z)=\exp(-\pi s^2\norm z^2)$,
\cref{eqn:cvp-local-gap} implies, for $x\ne0$,
\[
\frac{\rho_{1/s}(x+e/q)}{\rho_{1/s}(e/q)}\le\rho_{1/(s\sqrt{1-1/q})}(x)\le\rho_{\sqrt2/s}(x)
\]
where we use $q\ge 2$ in the last inequality.
Applying
\cref{lem:shell} with $k=0$ and width $\sqrt2/s$ gives
\begin{align}
\sum_{x\in\cL\setminus\{0\}}\rho_{1/s}(x+e/q)
&\le
\rho_{1/s}(e/q)
\sum_{x\in\cL\setminus\{0\}}\rho_{\sqrt2/s}(x)
\le
2^{-\Omega(n)}\rho_{1/s}(e/q).
\label{eqn:cvp-local-scalar-tail}
\end{align}
It follows from
\cref{cor:working-scale,eqn:cvp-local-primal,eqn:cvp-local-scalar-tail}
that
\[
f_y(w)=\rho_{1/s}(e/q)(1+2^{-\Omega(n)}).
\]
Observe that, using $s^2=4nt\ln2/(\pi d^2)$ and
$q^2=(4+o(1))p$,
\[
\rho_{1/s}(e/q)^p
=
\exp\left(-\frac{\pi s^2p}{q^2}\norm e^2\right)
=
2^{-(t+o(1))n\norm e^2/d^2}.
\]
Since $p=2^{o(n)}$ so that $(1+2^{-\Omega(n)})^p = 1+2^{-\Omega(n)}$, we obtain the first bound
\[
f_y(w)^p
=
2^{-(t+o(1))n\norm e^2/d^2}
(1+2^{-\Omega(n)}).
\]
If $\norm e\le\alpha\lambda$, then
$\norm e^2/d^2\le\alpha^2$, and the last assertion follows.

We next bound $G_y(w)$. \cref{eqn:cvp-local-primal} gives
\begin{align}
\frac{G_y(w)}{f_y(w)}+\frac{2\pi s^2}{q}e
=
-2\pi s^2
\frac{
\sum_{x\in\cL}x\rho_{1/s}(x+e/q)
}{
\sum_{x\in\cL}\rho_{1/s}(x+e/q)
}.
\label{eqn: exact_ratio}
\end{align}
We give a bound on this expression.
Since $v$ is closest to $y$ and
$e=v-y$, for every $x\in\cL$ we have
\[
\norm e\le\norm{e+x},
\qquad\text{and hence}\qquad
2\langle x,e\rangle\ge-\norm x^2.
\]
It follows that
\[
\norm{x+e/q}^2-\norm{e/q}^2
=
\norm x^2+\frac{2}{q}\langle x,e\rangle
\ge
\left(1-\frac1q\right)\norm x^2
\ge
\frac12\norm x^2.
\]
Therefore, for every $x\in\cL$,
\[
\frac{\rho_{1/s}(x+e/q)}{\rho_{1/s}(e/q)}
=
\exp\left(
-\pi s^2\left(
\norm{x+e/q}^2-\norm{e/q}^2
\right)
\right)\le
\exp\left(-\frac{\pi s^2}{2}\norm x^2\right)
=
\rho_{\sqrt2/s}(x).
\]
Since
$\norm x\le\lambda(1+\norm x/\lambda)$, applying
\cref{lem:shell} with $k=1$ and width $\sqrt2/s$ gives
\begin{align}
\sum_{x\in\cL\setminus\{0\}}
\norm{x}\rho_{1/s}(x+e/q)
\le
\lambda\rho_{1/s}(e/q)
\sum_{x\in\cL\setminus\{0\}}
\left(1+\frac{\norm x}{\lambda}\right)
\rho_{\sqrt2/s}(x)\le
2^{-\Omega(n)}\lambda\rho_{1/s}(e/q).
\label{eqn:cvp-local-tail}
\end{align}
The denominator of \cref{eqn: exact_ratio} is at least $\rho_{1/s}(e/q)$, hence
\cref{eqn:cvp-local-tail} gives
\[
\norm{
\frac{G_y(w)}{f_y(w)}+\frac{2\pi s^2}{q}e
}
\le
2^{-\Omega(n)}s^2\lambda.
\]
This proves the second bound.
\end{proof}

The length of $e$ is not known in advance.
If $\norm e<n^{-1/3}\lambda$, then $y$ already satisfies the
BDD promise. Otherwise, we can guess its length among polynomially many candidates. 
The next lemma formalizes this idea, together with the role of \cref{lem:svp-recovery} in the CVP case.

\begin{lemma}
\label{lem:cvp-recovery}
Let $\lambda=\lambda_1(\cL)\le d\le(1+1/n)\lambda$, let $v\in\cL$ be closest to $y$, let $e=v-y$, and let $w=B^{-1}v\bmod q$. 
Let $t_0<t<\min\{1/4,1/(4\alpha^2)\}.$
Assume that $\norm e\le\alpha\lambda$.
Suppose that $A\in\C^n$ satisfies
\begin{align}
\norm{A-cG_y(w)f_y(w)^{p-1}}\le\frac{s^2\lambda}{q}2^{-n/4+o(n)}
\label{eqn:cvp-required-accuracy}
\end{align}
for some $c\in\C$ with $|c|=1$.
If $\norm e\ge n^{-1/3}\lambda$, one of the well-defined nonzero vectors $\Re A/\norm{\Re A}$ and $\Im A/\norm{\Im A}$ is within distance $2^{-\Omega(n)}$ of either $e/\norm e$ or $-e/\norm e$.
For the corresponding unit vector $u$ (i.e., one of $\Re A/\norm{\Re A}$ or $\Im A/\norm{\Im A}$), the following holds for all sufficiently large $n$:
There is $j\in[\lceil4\alpha n^{1/3}\rceil]$ such that
\[
\norm{y\pm \frac{jd}{4n^{1/3}}u-v}<n^{-1/3}\lambda.
\]
Otherwise, 
if $\norm e<n^{-1/3}\lambda$, the target $y$ itself satisfies the BDD promise.
\end{lemma}

\begin{proof}
By \cref{lem:cvp-local}, there is $e_0\in\R^n$ with $\norm{e_0}\le2^{-\Omega(n)}\lambda$ such that
\begin{align}
-G_y(w)f_y(w)^{p-1}=\frac{2\pi s^2}{q}f_y(w)^p(e+e_0).
\label{eqn:cvp-direction-signal}
\end{align}
If $\norm e\ge n^{-1/3}\lambda$, then the last statement of \cref{lem:cvp-local}, $f_y(w)^p
\ge
2^{-(t\alpha^2+o(1))n}.$, gives
\[
\norm{G_y(w)f_y(w)^{p-1}}
=\norm{\frac{2\pi s^2}{q}f_y(w)^p(e+e_0)}
\ge\frac{s^2n^{-1/3}\lambda}{q}2^{-(t\alpha^2+o(1))n}.
\]
The error in \cref{eqn:cvp-required-accuracy} is exponentially smaller because $t\alpha^2<1/4$. Since the vector in \cref{eqn:cvp-direction-signal} is real and at least one of $|\Re c|$ and $|\Im c|$ is at least $1/\sqrt2$, \cref{lem:normalization-stability} gives the claimed direction as in the proof of \cref{lem:svp-recovery}.

Let $r=\norm e$. Since $r\le\alpha\lambda\le\alpha d$, some $j\in\{0,\ldots,\lceil4\alpha n^{1/3}\rceil\}$ satisfies $|\frac{jd}{4n^{1/3}}-r|\le \frac d{8n^{1/3}}$ so that $\frac{jd}{4n^{1/3}}u$ approximates $e$ in norm.
In the current case, $\norm{e}=r\ge n^{-1/3}\lambda>\frac d{8n^{1/3}}$ for all sufficiently large $n$, so $j\ge1$. For the sign satisfying $\norm{\pm u-e/r}\le2^{-\Omega(n)}$, we have
\begin{align*}
    \norm{y\pm \frac{jd}{4n^{1/3}}u-v}
    =\norm{-e \pm ru \mp ru \pm \frac{jd}{4n^{1/3}}u}
    &
    \le r\norm{\pm u-e/r}
+
\left|\frac{jd}{4n^{1/3}}-r\right|\notag
\\&
\le \alpha \lambda 2^{-\Omega(n)} +\frac{d}{8n^{1/3}}
<
n^{-1/3}\lambda
\end{align*}
holds for sufficiently large $n$. This proves the case for $\norm{e} \ge n^{-1/3} \lambda$.
The last assertion for the other case is the definition of the $n^{-1/3}$-BDD promise.
\end{proof}

\subsection{The coordinate decomposition and the estimator}
Observe that
for any $\varepsilon\in\{-1,1\}$,
\[
\E_X\left[\omega_q^{\ip{\varepsilon w}{[X]_q}}e^{-2\pi i\varepsilon\ip{X}{y}/q}\right]=F_{1/s}\left(\varepsilon\frac{Bw-y}{q}\right)=f_y(w),
\]
where the last equality follows because $F_{1/s}$ is even.
For independent $X_0,\ldots,X_{p-1}\sim D_{\cL^*,s}$ and for
\[
\Phi_y(X_0,\ldots,X_{p-1}):=\exp\left(-\frac{2\pi i}{q}\ip{\sum_{r=0}^{p-1}\varepsilon_rX_r}{y}\right),
\]
we have the following equation because of
the independence and the symmetry of the discrete Gaussian
\[
G_y(w)f_y(w)^{p-1}=
\E\left[2\pi iX_0\Phi_y(X_0,\ldots,X_{p-1})\omega_q^{\ip{w}{\sum_{r=0}^{p-1}\varepsilon_r[X_r]_q}}\right]
\]
where we use $[X_r]_q = B^T X_r \bmod q $ and $w=B^{-1} v \bmod q$.

\paragraph{What to estimate.}
Recall the notations for the spaces $H_j$
\[
H_j:=\{x=(x_1,\ldots,x_n)\in\F_q^n:x_k=0
\text{ for every }k\notin\{(j-1)n/m+1,\ldots,jn/m\}\}, 
\]
with the assumption that $m=\log_2 p+1$ divides $n$
and the decomposition 
\[
    \F_q^n=H_1\oplus\cdots\oplus H_m, \qquad Y=\sum_{j=1}^m (Y)_{H_j}
    \quad \text{for}\quad Y\in\F_q^n
\]
where $(Y)_{H_j}\in H_j$ denotes the unique component satisfying
$Y=\sum_{j=1}^m (Y)_{H_j}$.

Consider independent discrete Gaussian samples $X_0,...,X_{p-1} \sim D_{\cL^*,s}.$
Let $U$ be uniformly random in $\GL_n(\F_q)$.
Recall the recursive definition $Y_{0,r}:=U[X_r]_q=U(B^TX_r\bmod q)$ for $r\in[0,p)$ and define $Y_{j,a}:=Y_{j-1,2a}-Y_{j-1,2a+1}$ for $j\in[m-1]$ and $a\in[0,p/2^j)$. 
For $\varepsilon_r=(-1)^{\operatorname{wt}_2(r)}$, we have
\[Y_{m-1,0}
=
\sum_{r=0}^{p-1}\varepsilon_rY_{0,r},\qquad 
Y_{j,a}
=
\sum_{b=0}^{2^j-1}
\varepsilon_bY_{0,a2^j+b}
=
\varepsilon_{a2^j}
\sum_{r=a2^j}^{(a+1)2^j-1}
\varepsilon_rY_{0,r}.\]
Recall the constraint function from 
\cref{eqn:svp-ideal-indicator}:
\[
Q^{p-1}C_{\boldsymbol\delta}(X_0,\ldots,X_{p-1})=Q^{p-1}
\prod_{j=1}^{m-1}
\prod_{a=0}^{p/2^j-1}
\ind_{\{(Y_{j,a})_{H_j}=\delta_{j,a}\}}= 
\sum_{\boldsymbol\alpha}
\left(
\prod_{j=1}^{m-1}
\prod_{a=0}^{p/2^j-1}
\omega_q^{-\ip{\alpha_{j,a}}{\delta_{j,a}}}
\right)
\prod_{j=1}^{m-1}
\prod_{a=0}^{p/2^j-1}
\omega_q^{\ip{\alpha_{j,a}}{Y_{j,a}}}.
\]

For $\theta\in H_m$, analogously to \cref{eqn:svp-ideal-value}, we define
\begin{align}
A_{y,\boldsymbol\delta}(\theta):=Q^{p-1}\E\left[C_{\boldsymbol\delta}(X_0,\ldots,X_{p-1})2\pi iX_0\Phi_y(X_0,\ldots,X_{p-1})\omega_q^{\ip{\theta}{(Y_{m-1,0})_{H_m}}}\mid U,\boldsymbol\delta\right].
\label{eqn:cvp-ideal-value}
\end{align}
The expectation is over the independent samples, while $U$ and $\boldsymbol\delta$ are fixed.

The following lemma is a CVP variant of \cref{lem:svp-ideal-identity}. The proof is analogous.

\begin{lemma}
\label{lem:cvp-ideal-identity}
Let $\lambda=\lambda_1(\cL)\le d\le(1+1/n)\lambda$, let $v\in\cL$ be closest to $y$, suppose that $\norm{v-y}\le\alpha\lambda$, let $w=B^{-1}v\bmod q$ be nonzero, and let $\theta_*=(U^{-T}w)_{H_m}$. Except with probability $o(1)$ over $U$ and $\boldsymbol\delta$, there is $c_{\boldsymbol\delta}\in\C$ with $|c_{\boldsymbol\delta}|=1$ such that
\[
\norm{A_{y,\boldsymbol\delta}(\theta_*)-c_{\boldsymbol\delta}G_y(w)f_y(w)^{p-1}}\le\frac{s^2\lambda}{q}2^{-n/4+o(n)}.
\]
\end{lemma}

\begin{proof}
For $\boldsymbol\alpha=(\alpha_{j,a})_{j,a}$, let
$\chi_{\boldsymbol\alpha}(\boldsymbol\delta):=\prod_{j=1}^{m-1}\prod_{a=0}^{p/2^j-1}\omega_q^{-\ip{\alpha_{j,a}}{\delta_{j,a}}}.
$
Substituting $Q^{p-1}C_{\boldsymbol\delta}$ in \cref{eqn:cvp-ideal-value} to the summand for $\boldsymbol\alpha$ (\cref{eqn:svp-all-characters}), using \cref{eqn:svp-leaf-frequency-identity} gives
\begin{align}
A_{y,\boldsymbol\delta}(\theta)=\sum_{\boldsymbol\alpha}\chi_{\boldsymbol\alpha}(\boldsymbol\delta)G_y(U^Tk_{\boldsymbol\alpha,\theta,0})\prod_{r=1}^{p-1}f_y(\varepsilon_rU^Tk_{\boldsymbol\alpha,\theta,r}).
\label{eqn:cvp-ideal-expansion}
\end{align}
Here, we use the fact that for every $k\in\F_q^n$ and
$\varepsilon\in\{-1,1\}$,
\[
\E_X\left[\omega_q^{\ip{U^Tk}{[X]_q}}e^{-2\pi i\varepsilon\ip{X}{y}/q}\right]=F_{1/s}\left(\frac{BU^Tk-\varepsilon y}{q}\right)=f_y(\varepsilon U^Tk).
\]
Let $k_*=U^{-T}w$ and let $\boldsymbol\alpha^*$ be defined by \cref{eqn:svp-target-alpha}. Then $k_{\boldsymbol\alpha^*,\theta_*,r}=\varepsilon_rk_*$ for every $r$. Hence, the summand indexed by $\boldsymbol\alpha^*$ in \cref{eqn:cvp-ideal-expansion} is
\[
c_{\boldsymbol\delta}G_y(w)f_y(w)^{p-1},\qquad c_{\boldsymbol\delta}:=\chi_{\boldsymbol\alpha^*}(\boldsymbol\delta),
\]
and $|c_{\boldsymbol\delta}|=1$. Applying \cref{lem:svp-offset-character-orthogonality} to the remaining summands and then \cref{lem:cvp-nonprincipal-mass} below gives
\[
\E_{\boldsymbol\delta}\left[\norm{A_{y,\boldsymbol\delta}(\theta_*)-c_{\boldsymbol\delta}G_y(w)f_y(w)^{p-1}}^2\mid U\right]\le\left(\frac{s^2\lambda}{q}\right)^22^{-n/2+o(n)}.
\]
Markov's inequality proves the lemma.
\end{proof}

The proof of the following estimate is deferred to \cref{subsec:cvp-proofs}.
\begin{lemma}
\label{lem:cvp-nonprincipal-mass}
Under the assumptions of \cref{lem:cvp-ideal-identity}, let $k_*=U^{-T}w$ and let $\boldsymbol\alpha^*$ be defined by \cref{eqn:svp-target-alpha}. Except with probability $o(1)$ over $U$,
\begin{align}
\sum_{\boldsymbol\alpha\ne\boldsymbol\alpha^*}\norm{G_y(U^Tk_{\boldsymbol\alpha,\theta_*,0})\prod_{r=1}^{p-1}f_y(\varepsilon_rU^Tk_{\boldsymbol\alpha,\theta_*,r})}^2\le\left(\frac{s^2\lambda}{q}\right)^22^{-n/2+o(n)}.
\label{eqn:cvp-nonprincipal-mass}
\end{align}
\end{lemma}

\paragraph{Estimator.}
We now explain how to estimate the value $A_{y,\boldsymbol{\delta}}(\theta)$.
For the independent samples $X_{r,i}$ and the indexed recursive differences $Y_{j,a}(\boldsymbol i)$ defined in \cref{eqn: recursive Yi}, we define
\begin{align}
\label{def: phiy}
\Phi_y(\boldsymbol i):=\exp\left(-\frac{2\pi i}{q}\ip{\sum_{r=0}^{p-1}\varepsilon_rX_{r,i_r}}{y}\right),\qquad \boldsymbol i=(i_0,\ldots,i_{p-1})\in[N]^p.
\end{align}
For $\theta\in H_m$, our estimator is defined as follows:
\begin{align}
\widehat A_{y,\boldsymbol\delta}(\theta):=\frac{Q^{p-1}}{N^p}\sum_{\boldsymbol i\in[N]^p}C_{\boldsymbol\delta}(\boldsymbol i)2\pi iX_{0,i_0}\Phi_y(\boldsymbol i)\omega_q^{\ip{\theta}{(Y_{m-1,0}(\boldsymbol i))_{H_m}}}.
\label{eqn:cvp-estimator}
\end{align}

\begin{lemma}
\label{lem:cvp-finite-list-second-moment}
Let $\lambda=\lambda_1(\cL)\le d\le(1+1/n)\lambda$. Except with probability $2^{-\Omega(n)}$ over $U$, simultaneously for every $y\in\R^n$ and for every $\theta\in H_m$, it holds that
\begin{align}
\E_{\boldsymbol\delta,X}\left[\norm{\widehat A_{y,\boldsymbol\delta}(\theta)-A_{y,\boldsymbol\delta}(\theta)}^2\mid U\right]\le\left(\frac{s^2\lambda}{q}\right)^2Q^{-1}2^{o(n)} = \left(\frac{s^2\lambda}{q}\right)^22^{-n/2+o(n)}
\label{eqn:cvp-finite-list-second-moment}
\end{align}
The constants and the exceptional probability are uniform in $y$.
\end{lemma}

The proof of \cref{lem:cvp-finite-list-second-moment} is also deferred to \cref{subsec:cvp-proofs}.

\begin{lemma}
\label{lem:cvp-estimator-accuracy}
Let $\lambda=\lambda_1(\cL)\le d\le(1+1/n)\lambda$, let $v\in\cL$ be closest to $y$, suppose that $\norm{v-y}\le\alpha\lambda$, let $w=B^{-1}v\bmod q$ be nonzero, and let $\theta_*=(U^{-T}w)_{H_m}$.
Except with probability $o(1)$ over $U$, $\boldsymbol\delta$, and the samples, there is $c_{\boldsymbol\delta}\in\C$ with $|c_{\boldsymbol\delta}|=1$ such that
\begin{align}
\norm{\widehat A_{y,\boldsymbol\delta}(\theta_*)-c_{\boldsymbol\delta}G_y(w)f_y(w)^{p-1}}\le\frac{s^2\lambda}{q}2^{-n/4+o(n)}.
\label{eqn:cvp-estimator-accuracy}
\end{align}
\end{lemma}
\begin{proof}
With probability $1-o(1)$, the conclusions of both \cref{lem:cvp-ideal-identity,lem:cvp-finite-list-second-moment} hold. 
Applying Markov's inequality to \cref{eqn:cvp-finite-list-second-moment} gives that
\[
\norm{\widehat A_{y,\boldsymbol\delta}(\theta_*)-A_{y,\boldsymbol\delta}(\theta_*)}\le\frac{s^2\lambda}{q}2^{-n/4+o(n)}
\]
holds
except with probability $o(1)$. Combining this estimate with \cref{lem:cvp-ideal-identity} and applying the triangle inequality proves the lemma.
\end{proof}

We next give the batch computation of the estimator for all $\theta$ in \cref{alg:cvp-estimator-computation} similar to \cref{lem:svp-estimator-computation}.
\begin{lemma}
\label{lem:cvp-estimator-computation}
Except with probability $o(1)$ over $\boldsymbol\delta$, approximations to all vectors $\widehat A_{y,\boldsymbol\delta}(\theta)$ for $\theta\in H_m$, each with additive error at most $(s^2d/q)2^{-n/4-\Omega(n)}$, can be computed in time and space $2^{n/2+o(n)}$.
\end{lemma}
\begin{proof}
Induction on $j$ shows that the third component of an element in
$\mathsf L_{j,a}$ is
\[
\phi
=
\exp\left(
-\frac{2\pi i}{q}
\ip{
\sum_{r=a2^j}^{(a+1)2^j-1}
\varepsilon_{r-a2^j}X_{r,i_r}
}{y}
\right).
\]
If $a=0$, its fourth component is
$Z=2\pi iX_{0,i_0}\phi$, whereas if $a\ne0$, its fourth component is
$Z=0$. In particular, every element chosen from
$\mathsf L_{j-1,2a+1}$ has fourth component $Z_R=0$, which explains
the update $Z_L\overline{\phi_R}$ in Step~2(b).
At the root, where $j=m-1$ and necessarily $a=0$, the value $\phi$
coincides with $\Phi_y(\boldsymbol i)$ in \cref{def: phiy}.
Thus, Steps~3 and~4 compute
$\widehat A_{y,\boldsymbol\delta}(\theta)$ as defined in
\cref{eqn:cvp-estimator}.

The difference conditions and the abort bounds are exactly those in
\cref{alg:svp-estimator-computation}. The additional phase arithmetic
incurs only polynomial overhead under the finite-precision convention
in \cref{sec:preliminaries}. Hence,
\cref{eqn:svp-computation-abort,eqn:svp-total-computation-records}
give the claimed abort probability, time, and space.
\end{proof}

\begin{algorithmblock}{Computation of $\widehat A_{y,\boldsymbol\delta}$}
\label{alg:cvp-estimator-computation}
\item For every $r\in[0,p)$ and $i\in[N]$, let
$\phi_{r,i}:=e^{-2\pi i\ip{X_{r,i}}{y}/q}$, let
$Z_{0,i}:=2\pi iX_{0,i}\phi_{0,i}$, and let $Z_{r,i}:=0$ for
$r\ne0$. Construct lists
\[
\mathsf L_{0,r}:=
\{(i,Y_{r,i},\phi_{r,i},Z_{r,i}):i\in[N]\}
\qquad\text{for }r\in[0,p).
\]
\item For $j=1,\ldots,m-1$ and $a\in[0,p/2^j)$, we do:
\begin{enumerate}
    \item For each $z\in H_j$, construct
    $\mathsf L_{j-1,2a+1,z}
    :=
    \{(\boldsymbol i_R,Y_R,\phi_R,Z_R)
    \in\mathsf L_{j-1,2a+1}:(Y_R)_{H_j}=z\}.$
    \item For every
    $(\boldsymbol i_L,Y_L,\phi_L,Z_L)\in\mathsf L_{j-1,2a}$,
    enumerate all elements
    $(\boldsymbol i_R,Y_R,\phi_R,Z_R)$ in
    $\mathsf L_{j-1,2a+1,(Y_L)_{H_j}-\delta_{j,a}}$
    and include
    $((\boldsymbol i_L,\boldsymbol i_R),
    Y_L-Y_R,\phi_L\overline{\phi_R},
    Z_L\overline{\phi_R})$
    in $\mathsf L_{j,a}$. Abort immediately if
    $|\mathsf L_{j,a}|>Qp^{7\cdot2^j}$.
\end{enumerate}
\item For every $b\in H_m$, compute
$T_{y,\boldsymbol\delta}(b):=
\sum_{\substack{
(\boldsymbol i,Y,\phi,Z)\in\mathsf L_{m-1,0}\\
(Y)_{H_m}=b}}Z.$
\item Using a vector-valued fast Fourier transform over $H_m$, compute for
every $\theta\in H_m$
\[
\widehat A_{y,\boldsymbol\delta}(\theta)
=
\frac{Q^{p-1}}{N^p}
\sum_{b\in H_m}
T_{y,\boldsymbol\delta}(b)
\omega_q^{\ip{\theta}{b}}.
\]
\end{algorithmblock}

\subsection{The algorithm}
\paragraph{Parameters.}
Fix constants $\alpha>0$ and $t$ satisfying $t_0<t<\min\{1/4,1/(4\alpha^2)\}.$
Use the parameters fixed at the beginning of this section, and set
$N:=Qp^4$. Then $p\log p=o(n)$ and $Q,N=2^{n/2+o(n)}$.

\noindent\textbf{Input.}
A basis $B\in\R^{n\times n}$ of a full-rank lattice
$\cL=B\Z^n$ and a target $y\in\R^n$ satisfying $\dist(y,\cL)\le\alpha\lambda_1(\cL).$
 
\noindent\textbf{Output.}
A vector $v\in\cL$ minimizing $\norm{v-y}$, with probability at
least $2/3$.
\begin{algorithmblock}{CVP with a distance guarantee}
\label{alg:bounded-cvp}
\item Use LLL to construct the polynomial-size set $\mathcal D$ of
values $d$ used in \cref{alg:main-svp}.
\item For every $d\in\mathcal D$, independently repeat the following
procedure a constant number of times:
\begin{enumerate}
    \item Run the preprocessing algorithm in
    \cref{thm:bdd-preprocessing}.
    \item Choose $a$ uniformly from $\F_q^n$, using its representative
    in $\{0,\ldots,q-1\}^n$, and let $\widetilde y:=y+Ba$. Query the
    BDD algorithm at $\widetilde y$.
    \item Let $s^2=4nt\ln2/(\pi d^2)$ and obtain the $pN$ samples.
    Abort this repetition if a sample has norm larger than
    $Cs\sqrt n$, where $C$ is a sufficiently large constant. Choose
    $U$ and $\boldsymbol\delta$ randomly, and run
    \cref{alg:cvp-estimator-computation} for $\widetilde y$. Abort
    this repetition if that subroutine aborts.
    \item Let $h=d/(4n^{1/3})$ and
    $J=\lceil4\alpha n^{1/3}\rceil$. For every $\theta\in H_m$, every
    nonzero $V\in\{\Re\widehat A_{\widetilde y,\boldsymbol\delta}(\theta),
    \Im\widehat A_{\widetilde y,\boldsymbol\delta}(\theta)\},$
    and every $j\in[J]$, query the BDD algorithm at
    $\widetilde y+jhV/\norm V$ and
    $\widetilde y-jhV/\norm V$.
    \item Discard every failure symbol and every returned vector
    outside $\cL$.
    Subtract $Ba$ from every remaining vector and
    collect the resulting vectors.
\end{enumerate}
\item If no vector is collected, declare failure. Otherwise, return
the vector minimizing its distance to $y$.
\end{algorithmblock}
\begin{proof}[Proof of \cref{thm:bounded-cvp}]
Fix a closest vector $v$ to $y$, let $\lambda=\lambda_1(\cL)$, and
let $e=v-y$. Consider a value $d\in\mathcal D$ satisfying
$\lambda\le d\le(1+1/n)\lambda$, and fix one repetition corresponding to this value.

After the random translation $\widetilde y = y+Ba$, $\widetilde v:=v+Ba$ is closest to
$\widetilde y$, $\widetilde v-\widetilde y=e$.
If $\norm e<n^{-1/3}\lambda$, the query at $\widetilde y$ returns
$\widetilde v$ whenever the BDD preprocessing succeeds. 
We consider the other case in the analysis below.
Note that given random $a$,
\[
w:=B^{-1}\widetilde v\bmod q
=
B^{-1}v+a\bmod q
\]
is uniform in $\F_q^n$. In particular, for the independently chosen
$U$,
\[
\Pr[w\in U^T(H_1\oplus\cdots\oplus H_{m-1})]
=
Q^{-1}
=
o(1).
\]
We condition on the complementary event. Then $w\ne0$, so that we can apply \cref{lem:cvp-ideal-identity,lem:cvp-estimator-accuracy}. The excluded event contributes only $o(1)$ to the failure
probability.

By \cref{cor:working-scale,thm:dgs-sampling}, the generated discrete Gaussian samples
have inverse-exponentially small statistical distance from the
independent samples, which only changes the probabilities negligibly small. 
Let $\theta_*=(U^{-T}w)_{H_m}$. By
\cref{lem:cvp-estimator-accuracy}, except with probability $o(1)$,
the vector
$\widehat A_{\widetilde y,\boldsymbol\delta}(\theta_*)$ satisfies
\cref{eqn:cvp-required-accuracy} for some $c\in\C$ with $|c|=1$.
Applying \cref{lem:cvp-recovery}, there are a nonzero vector $V\in\{\Re\widehat A_{\widetilde y,\boldsymbol\delta}(\theta_*),
\Im\widehat A_{\widetilde y,\boldsymbol\delta}(\theta_*)\},$
an index $j\in[J]$, and a sign $\sigma\in\{-1,1\}$ such that
\[
\norm{\widetilde y+\sigma jh\frac{V}{\norm V}-\widetilde v}
<
n^{-1/3}\lambda.
\]
The algorithm queries the BDD algorithm at this point, since it
considers every such $V$, $j$, and sign. Hence, successful BDD
preprocessing returns $\widetilde v$.
After
subtracting $Ba$, the algorithm therefore obtains $v$, which will be returned (or one among the same minimal distance).

The sample generation and \cref{lem:cvp-estimator-computation} use
time and space $2^{n/2+o(n)}$ in each repetition. The shifted-query
step makes $Q2^{o(n)}$ BDD queries because
$J=O(n^{1/3})$, so its total time has the same bound. The
$n^2+1$ choices of $d$ and the constant number of repetitions incur
only a polynomial factor.

Each repetition for the target value $d$ succeeds with constant
probability by
\cref{lem:dgs-tail,lem:cvp-estimator-accuracy,lem:cvp-estimator-computation,thm:bdd-preprocessing}. 
A constant number of
repetitions raises the success probability to at least $2/3$.
\end{proof}

\subsection{Random target CVP in a random lattice}

Let $\mu_n$ denote the Haar probability measure on
$X_n:=\SL_n(\R)/\SL_n(\Z),$
the space of covolume-one lattices in $\R^n$. In this subsection,
$\cL\sim\mu_n$ and, conditional on $\cL$, the target
$y\bmod\cL$ is uniform in $\R^n/\cL$. We have the following result.

\begin{corollary}
\label{cor:random-cvp}
Let $\cL\sim\mu_n$. With probability
$1-\exp(-\Omega(n^{2/3}))$ over $\cL$,
\cref{alg:bounded-cvp} solves a uniformly random
CVP target $(y,\cL)$ with probability at least $1/2$ for all sufficiently large $n$ over $y\bmod\cL$ and the internal
randomness of the algorithm.
\end{corollary}

The corollary follows from the following lemma because, for $\alpha=1.01<1/2\sqrt{t_0}=1.039\cdots$ and all sufficiently large $n$, the bound in the lemma below implies
$\dist(y,\cL)\le\alpha\lambda_1(\cL)$, under which one execution of
\cref{alg:bounded-cvp} succeeds with probability at least $2/3$.

\begin{lemma}
\label{lem:random-cvp-distance}
With probability $1-\exp(-\Omega(n^{2/3}))$ over
$\cL\sim\mu_n$, we have
\begin{align}
\Pr_{y\bmod\cL}\left[
\dist(y,\cL)\le
\frac{1+n^{-1/3}}{1-n^{-1/3}}\lambda_1(\cL)
\right]
\ge 1-\exp(-\Omega(n^{2/3})).
\label{eqn:random-cvp-distance}
\end{align}
\end{lemma}

\begin{proof}
Let $V_n$ be the volume of the Euclidean unit ball, let
$r_n=V_n^{-1/n}$. For $R>0$, let
$N_R(\cL)$ be the number of nonzero vectors of $\cL$ in the
Euclidean ball of radius $R$. Siegel's mean-value formula
\cite{Sie45} gives
$\E[N_R(\cL)]=V_nR^n$, and hence
\begin{align}
\Pr_{\cL}\left[
\lambda_1(\cL)<(1-n^{-1/3})r_n
\right]
\le
(1-n^{-1/3})^n.
\label{eqn:random-lattice-short-vector}
\end{align}

Let $M_R(\cL,y)$ be the number of $x\in\cL$ satisfying
$\norm{x-y}\le R$, and write $V=V_nR^n$. Let $\mathcal F$ be a
fundamental domain of $\cL$ and let $\ind_{B_R}$ denote the indicator
of the Euclidean ball $B_R$ of radius $R$. Since $\cL$ has covolume
one, $\operatorname{vol}(\mathcal F)=1$, and $M_R(\cL,y)=\sum_{x\in\cL}\ind_{B_R}(x-y).
$
Unfolding the integral over $\mathcal F$ gives, for every fixed
$\cL$,
\[
\E_{y\bmod\cL}[M_R(\cL,y)]
=
\sum_{x\in\cL}\int_{\mathcal F}\ind_{B_R}(x-y)\,dy
=
\int_{\R^n}\ind_{B_R}(u)\,du
=
V.
\]
Similarly, expanding the square and writing $z=x-x'$ gives
\begin{align*}
\E_{y\bmod\cL}[M_R(\cL,y)^2]
&=
\sum_{x,x'\in\cL}
\int_{\mathcal F}
\ind_{B_R}(x-y)\ind_{B_R}(x'-y)\,dy
\notag\\
&=
\sum_{z\in\cL}\sum_{x\in\cL}
\int_{\mathcal F}
\ind_{B_R}(x-y)\ind_{B_R}(x-z-y)\,dy
\notag\\
&=
\sum_{z\in\cL}
\int_{\R^n}
\ind_{B_R}(u)\ind_{B_R}(u-z)\,du
\notag\\
&=
\sum_{z\in\cL}
\operatorname{vol}\bigl(B_R\cap(B_R+z)\bigr).
\end{align*}
Let $h(z):=\operatorname{vol}\bigl(B_R\cap(B_R+z)\bigr).
$
The term corresponding to $z=0$ is $h(0)=\operatorname{vol}\bigl(B_R\bigr)=V$. Siegel's mean-value
formula applied to $h$ gives
\[
\E_{\cL\sim\mu_n}
\left[
\sum_{z\in\cL\setminus\{0\}}h(z)
\right]
=
\int_{\R^n}h(z)\,dz.
\]
Moreover, by Fubini's theorem and translation invariance of Lebesgue
measure,
\begin{align*}
\int_{\R^n}h(z)\,dz
&=
\int_{\R^n}\int_{\R^n}
\ind_{B_R}(u)\ind_{B_R}(u-z)\,du\,dz
\notag\\
&=
\int_{\R^n}\ind_{B_R}(u)
\left(\int_{\R^n}\ind_{B_R}(u-z)\,dz\right)du
=
V^2.
\end{align*}
Consequently,
\[
\E_{\cL,y}[M_R(\cL,y)^2]
=V^2+V,\qquad
\operatorname{Var}_{\cL,y}(M_R(\cL,y))
=V.
\]
Chebyshev's inequality gives
\begin{align}
\Pr_{\cL,y}\left[\dist(y,\cL)>R\right]
=
\Pr_{\cL,y}\left[M_R(\cL,y)=0\right]
\le V^{-1}.
\label{eqn:random-cvp-covering}
\end{align}

Taking $R=(1+n^{-1/3})r_n$ and combining
\cref{eqn:random-lattice-short-vector,eqn:random-cvp-covering} gives
\begin{align}
\Pr_{\cL,y}\left[
\frac{\dist(y,\cL)}{\lambda_1(\cL)}
>
\frac{1+n^{-1/3}}{1-n^{-1/3}}
\right]
&\le
(1-n^{-1/3})^n+(1+n^{-1/3})^{-n}
=
\exp(-\Omega(n^{2/3})).
\label{eqn:random-cvp-joint-distance}
\end{align}
Since $\frac{1+n^{-1/3}}{1-n^{-1/3}}
=
1+O(n^{-1/3}),
$
this proves the claimed bound jointly over $\cL$ and $y\bmod\cL$.

Finally, let $b(\cL)$ be the conditional probability over
$y\bmod\cL$ that \cref{eqn:random-cvp-distance} fails, or more precisely,
\[
b(\cL)
:=
\Pr_{y\bmod\cL}\left[
\frac{\dist(y,\cL)}{\lambda_1(\cL)}>\frac{1+n^{-1/3}}{1-n^{-1/3}}
\right].
\]
By
\cref{eqn:random-cvp-joint-distance}, there is a constant $c>0$ such
that, for all sufficiently large $n$,
\[
\E_{\cL}[b(\cL)]\le \exp(-c n^{2/3}).
\]
Markov's inequality therefore gives
\[
\Pr_{\cL}\left[
b(\cL)>\exp(-c n^{2/3}/2)
\right]
\le \exp(-c n^{2/3}/2).
\]
This proves the lemma.
\end{proof}

\input{5z.proof}

%% file: 5z.proof.tex
\subsection{Proofs of the CVP lemmas}
\label{subsec:cvp-proofs}

We use the notation from \cref{subsec:prooflem} as follows:
\begin{align*}
P_j&=H_1\oplus\cdots\oplus H_j,&
S_j&=U^TP_j,&
R_j&=H_{j+1}\oplus\cdots\oplus H_m,\notag\\
K_j&=\{Bx:x\in\Z^n,\ x\bmod q\in S_j\},&
\sigma_j&=\frac{q}{2^{j/2}s},\\
\mathcal A_{\le j}
&=
\left\{
(\alpha_{h,a})_{\substack{1\le h\le j\\0\le a<2^{j-h}}}:
\alpha_{h,a}\in H_h
\right\},&
k_{\boldsymbol\alpha^{\le j},\kappa,r}
&=
\varepsilon_r\kappa
+
\sum_{h=1}^j
\varepsilon_{r\bmod2^h}
\alpha_{h,\lfloor r/2^h\rfloor}.
\end{align*}
The proofs in this section follow the corresponding proofs
for SVP in that subsection with small changes, which will be clearly noted below.

\paragraph{Recursive mass representations.}
For $\kappa\in\F_q^n$, we define the shifted versions of $V$ and $T$ in \cref{eqn:svp-frequency-mass-base}
\[
V_{0,y}(\kappa):=|f_y(U^T\kappa)|^2,\qquad
T_{0,y}(\kappa):=\norm{G_y(U^T\kappa)}^2,
\]
and, for $j\in\{1,\ldots,m-1\}$, define recursively following \cref{eqn:svp-frequency-mass-recursion}
\begin{align}
V_{j,y}(\kappa)&:=\sum_{\alpha\in H_j}
V_{j-1,y}(\kappa+\alpha)^2,&
T_{j,y}(\kappa)&:=\sum_{\alpha\in H_j}
T_{j-1,y}(\kappa+\alpha)V_{j-1,y}(\kappa+\alpha).
\label{eqn:cvp-frequency-mass-recursion}
\end{align}
These functions are $P_j$-periodic.
The following lemma is analogous to \cref{lem:svp-frequency-mass-closed-form}.

\begin{lemma}
\label{lem:cvp-frequency-mass-closed-form}
For every $j\in\{0,\ldots,m-1\}$ and $\kappa\in\F_q^n$,
\begin{align*}
V_{j,y}(\kappa)
&=
\sum_{\boldsymbol\alpha^{\le j}\in\mathcal A_{\le j}}
\prod_{r=0}^{2^j-1}
\left|f_y\left(
\varepsilon_rU^Tk_{\boldsymbol\alpha^{\le j},\kappa,r}
\right)\right|^2,
\notag\\
T_{j,y}(\kappa)
&=
\sum_{\boldsymbol\alpha^{\le j}\in\mathcal A_{\le j}}
\norm{
G_y\left(U^Tk_{\boldsymbol\alpha^{\le j},\kappa,0}\right)
\prod_{r=1}^{2^j-1}
f_y\left(
\varepsilon_rU^Tk_{\boldsymbol\alpha^{\le j},\kappa,r}
\right)
}^2.
\end{align*}
\end{lemma}

\begin{proof}
We follow the induction in the proof of
\cref{lem:svp-frequency-mass-closed-form}. The case $j=0$ is obvious by definition. 
Decompose
$\boldsymbol\alpha^{\le j}$ as
$(\alpha,\boldsymbol\beta^L,\boldsymbol\beta^R)$ using
\cref{eqn:svp-left-right-alpha}. The left half is indexed by
$(\boldsymbol\beta^L,\kappa+\alpha)$. For
$0\le r<2^{j-1}$, the second identity in
\cref{eqn:svp-left-right-k} and
$\varepsilon_{2^{j-1}+r}=-\varepsilon_r$ give
\begin{align}
\varepsilon_{2^{j-1}+r}
k_{\boldsymbol\alpha^{\le j},\kappa,2^{j-1}+r}
=
\varepsilon_r
k_{-\boldsymbol\beta^R,\kappa+\alpha,r}.
\label{eqn:cvp-right-subtree-reindexing}
\end{align}
Since $\boldsymbol\beta^R\mapsto-\boldsymbol\beta^R$ is a
bijection of $\mathcal A_{\le j-1}$, the scalar sums from both
halves equal $V_{j-1,y}(\kappa+\alpha)$. This gives the first
recursion in \cref{eqn:cvp-frequency-mass-recursion}. 

For the second
recursion, the index $r=0$, at which $G_y$ appears, belongs to the
left half. Hence the left half gives
$T_{j-1,y}(\kappa+\alpha)$, while the right half gives
$V_{j-1,y}(\kappa+\alpha)$. Summing over $\alpha\in H_j$ gives the
second recursion. This proves the lemma.
\end{proof}

\paragraph{Shifted Gaussian mass bounds.}
The next lemma is the CVP counterpart of
\cref{lem:svp-total-frequency-masses,lem:svp-fixed-tree-masses}.

\begin{lemma}
\label{lem:cvp-shifted-mass-bounds}
Let $\lambda=\lambda_1(\cL)\le d\le(1+1/n)\lambda$. Except with
probability $2^{-\Omega(n)}$ over $U$, simultaneously for every
$y\in\R^n$, the following estimates hold. For every
$j\in\{1,\ldots,m-1\}$,
\begin{align}
\sum_{\kappa\in R_j}V_{j,y}(\kappa)
&\le2^{((m-j-1)/2+o(1))n},&
\sum_{\kappa\in R_j}T_{j,y}(\kappa)
&\le
\left(\frac{s^2\lambda}{q}\right)^2
2^{((m-j-1)/2+o(1))n}.
\label{eqn:cvp-total-mass-bounds}
\end{align}
Moreover, for every $j\in\{0,\ldots,m-2\}$ and every
$\kappa\in\F_q^n$,
\begin{align}
\sum_{\alpha\in H_{j+1}}V_{j,y}(\kappa+\alpha)
&\le1+2^{-\Omega(n)},&
\sum_{\alpha\in H_{j+1}}T_{j,y}(\kappa+\alpha)
&\le
\left(\frac{s^2\lambda}{q}\right)^22^{o(n)}.
\label{eqn:cvp-fixed-mass-bounds}
\end{align}
\end{lemma}

\begin{proof}
By \cref{lem:svp-random-lattice-masses}, except with probability
$2^{-\Omega(n)}$ over $U$, the Gaussian-mass estimates in
\cref{eqn:svp-random-lattice-gradient-mass} hold simultaneously for
every $j$. Fix any $U$ for which these estimates hold. We prove the
claimed inequalities for this $U$, simultaneously for every
$y\in\R^n$.

We first prove the two bounds involving $V_{j,y}$. We refer to
\[
\sum_{\kappa\in R_j}V_{j,y}(\kappa)
\qquad\text{and}\qquad
\sum_{\gamma\in H_{j+1}}V_{j,y}(\kappa+\gamma)
\]
as the $R_j$-sum and the $H_{j+1}$-sum, respectively.

For the $R_j$-sum,
\cref{lem:cvp-frequency-mass-closed-form} gives
\[
\sum_{\kappa\in R_j}V_{j,y}(\kappa)
=
\sum_{\kappa\in R_j}
\sum_{\boldsymbol\alpha^{\le j}\in\mathcal A_{\le j}}
\prod_{r=0}^{2^j-1}|f_y(z_r)|^2,
\]
where
$z_r
:=
\varepsilon_rU^T
k_{\boldsymbol\alpha^{\le j},\kappa,r}.$
Similarly, for $w_r
:=
\varepsilon_rU^T
k_{\boldsymbol\beta^{\le j},\kappa+\gamma,r},$
the $H_{j+1}$-sum for fixed $\kappa$ is
\[
\sum_{\gamma\in H_{j+1}}V_{j,y}(\kappa+\gamma)
=
\sum_{\gamma\in H_{j+1}}
\sum_{\boldsymbol\beta^{\le j}\in\mathcal A_{\le j}}
\prod_{r=0}^{2^j-1}|f_y(w_r)|^2.
\]

For every $\xi\in\F_q^n$, the primal representation of $f_y$ gives
\[
|f_y(\xi)|^2
=
\frac{1}{\rho_{1/s}(\cL)^2}
\sum_{x_0,x_1\in B\xi-y+q\cL}
\exp\left(
-\frac{\pi s^2}{q^2}
\left(
\norm{x_0}^2+\norm{x_1}^2
\right)
\right).
\]
Accordingly, expanding the $R_j$-sum and the $H_{j+1}$-sum introduces
\[
u_{r,0},u_{r,1}\in Bz_r-y+q\cL,
\qquad\text{and}\qquad
v_{r,0},v_{r,1}\in Bw_r-y+q\cL,
\]
respectively.
To deal with two cases simultaneously, 
we define
\[
x^R_{r,b}:=u_{r,b},
\qquad
x^H_{r,b}:=v_{r,b}.
\]
For each $X\in\{R,H\}$, define
$a^X_{j,0}$ and $(d^X_{h,r})_{h,r}$ by
\[
\Psi_j((x^X_{r,b})_{r,b})
=
\left(a^X_{j,0},(d^X_{h,r})_{h,r}\right).
\]

Adding $y$ to every entry removes the common shift:
\[
u_{r,b}+y\in Bz_r+q\cL,
\qquad
v_{r,b}+y\in Bw_r+q\cL.
\]
Moreover, the definition of the average-difference transformation gives for each $X\in \{R,H\}$
\[
\Psi_j((x^X_{r,b}+y)_{r,b})
=
\left(a^X_{j,0}+y,(d^X_{h,r})_{h,r}\right).
\]
Thus
\cref{clm:svp-expanded-sum-indexing,clm:svp-average-difference}
apply to the tuples $(u_{r,b}+y)_{r,b}$, while
\cref{clm:svp-fixed-tree-indexing,clm:svp-fixed-tree-average-difference}
apply to the tuples $(v_{r,b}+y)_{r,b}$. We obtain the following two
claims.

\begin{claim}
\label{clm:cvp-shifted-sum-indexing}
Fix $y\in\R^n$.
\begin{itemize}[nosep]
\item In the $R_j$-sum, the tuple $(u_{r,b})_{r,b}$ uniquely
determines $(\boldsymbol\alpha^{\le j},\kappa)$.
\item For fixed $\kappa$, in the $H_{j+1}$-sum, the tuple
$(v_{r,b})_{r,b}$ uniquely determines
$(\boldsymbol\beta^{\le j},\gamma)$.
\end{itemize}
Consequently, the $R_j$-sum can be indexed by $(u_{r,b})_{r,b}$,
and, for fixed $\kappa$, the $H_{j+1}$-sum can be indexed by
$(v_{r,b})_{r,b}$, without multiplicity.
\end{claim}

\begin{claim}
\label{clm:cvp-shifted-average-difference}
Fix $y\in\R^n$.
\begin{itemize}[nosep]
\item In the $R_j$-sum, if $a^R_{j,0}$ is allowed to vary while all
$d^R_{h,r}$ are fixed, its possible values are contained in a single
coset of $\cL$.
\item In the $H_{j+1}$-sum, if $a^H_{j,0}$ is allowed to vary while
all $d^H_{h,r}$ are fixed, its possible values are contained in a
single coset of $K_{j+1}$.
\item Every $d^R_{0,r}$ and $d^H_{0,r}$ belongs to $K_0$.
\item For $X\in\{R,H\}$, $h\in[j]$, and
$r\in[0,2^{j-h})$, if $d^X_{h,r}$ is allowed to vary while all
lower-level differences $d^X_{\ell,r'}$ with $\ell<h$ are fixed,
its possible values are contained in a single coset of $K_{h-1}$.
\end{itemize}
Each coset may depend on the fixed variables and on $y$.
\end{claim}

By \cref{clm:cvp-shifted-sum-indexing} and the bijectivity of
$\Psi_j$, every term in either expansion is counted exactly once in
the transformed variables. For each $X\in\{R,H\}$,
\cref{eqn:svp-all-square-decompositions} gives
\[
\sum_{r=0}^{2^j-1}\sum_{b\in\bit}\norm{x^X_{r,b}}^2
=
2^{j+1}\norm{a^X_{j,0}}^2
+
\frac12\sum_{r=0}^{2^j-1}\norm{d^X_{0,r}}^2
+
\sum_{h=1}^j
2^{h-1}
\sum_{r=0}^{2^{j-h}-1}\norm{d^X_{h,r}}^2.
\]
We first sum over $a^X_{j,0}$ and then over the variables
$d^X_{h,r}$ in decreasing order of $h$. For $j\ge1$, their numbers,
containing cosets, and Gaussian widths are as follows:
\[
\begin{array}{c|c|c|c|c}
\text{variable}
&
\text{number}
&
R_j\text{-sum}
&
H_{j+1}\text{-sum}
&
\text{width}
\\ \hline
a_{j,0}
&
1
&
\text{a coset of }\cL
&
\text{a coset of }K_{j+1}
&
\dfrac{q}{2^{(j+1)/2}s}=\sigma_{j+1}
\\[2mm]
d_{0,r}
&
2^j
&
K_0
&
K_0
&
\dfrac{\sqrt2q}{s}
\\[2mm]
d_{1,r}
&
2^{j-1}
&
\text{a coset of }K_0
&
\text{a coset of }K_0
&
\dfrac{q}{s}
\\[2mm]
d_{h+1,r}
&
2^{j-h-1}
&
\text{a coset of }K_h
&
\text{a coset of }K_h
&
\dfrac{q}{2^{h/2}s}=\sigma_h
\quad(1\le h<j).
\end{array}
\]
For $j=0$, only the $H_1$-sum is needed. Its variables are
$a^H_{0,0}$, contained in a coset of $K_1$ with width $\sigma_1$,
and $d^H_{0,0}\in K_0$, with width $\sqrt2q/s$.

Applying
\cref{clm:cvp-shifted-average-difference,lem:shifted-mass-maximum}
at each step bounds each expansion by the product of the
corresponding unshifted Gaussian masses listed above.

For the $R_j$-sum, this gives the right-hand side of
\cref{eqn:svp-total-frequency-mass-product}. The estimates following
that equation give
\[
\sum_{\kappa\in R_j}V_{j,y}(\kappa)
\le
1+2^{((m-j-1)/2-g/2+o(1))n}
\le
2^{((m-j-1)/2+o(1))n}.
\]
For the $H_{j+1}$-sum, this gives the final right-hand side of
\cref{eqn:svp-fixed-tree-base-product} when $j=0$ and of
\cref{eqn:svp-fixed-tree-product} when $j\ge1$. Therefore
\[
\sum_{\gamma\in H_{j+1}}V_{j,y}(\kappa+\gamma)
\le
1+2^{-\Omega(n)}.
\]
These estimates hold simultaneously for every $y\in\R^n$, since the
event fixed above depends only on $U$ and
\cref{lem:shifted-mass-maximum} holds for every shift.

It remains to prove the two inequalities involving $T_{j,y}$.
The same Cauchy--Schwarz argument as in
\cref{eqn:svp-gradient-cauchy-schwarz} bounds the corresponding
expansion by the expansion used above for $V_{j,y}$ with the
additional factor
\[
\left(\frac{2\pi s^2}{q}\right)^2
\norm{x^X_{0,0}}^2
\]
for $X\in\{R,H\}$. Let
\[
\mathcal E_X
:=
\sum_{r=0}^{2^j-1}\sum_{b\in\bit}\norm{x^X_{r,b}}^2.
\]
Since $\norm{x^X_{0,0}}^2\le\mathcal E_X$,
\cref{eqn:svp-gradient-moment-absorption} gives
\[
\norm{x^X_{0,0}}^2
e^{-\pi s^2\mathcal E_X/q^2}
\le
\frac{nq^2}{\pi s^2}
e^{-\pi(1-1/n)s^2\mathcal E_X/q^2}.
\]
the preceding calculation for $V$'s above applies with every Gaussian width
multiplied by $1/\sqrt{1-1/n}$ and with the additional factor
\[
\left(\frac{2\pi s^2}{q}\right)^2
\frac{nq^2}{\pi s^2}
=4\pi ns^2
\le
\left(\frac{s^2\lambda}{q}\right)^22^{o(n)},
\]
where we use $s^2\lambda^2=\Theta(n)$ and $q=n^{O(1)}$.

The remaining argument is the same as in the proofs of
\cref{lem:svp-total-frequency-masses,lem:svp-fixed-tree-masses},
with two changes. First, the common translation by $-y$ remains only
in the final average, and its Gaussian mass is bounded by the
corresponding unshifted mass using
\cref{lem:shifted-mass-maximum}. Second, for the $R_j$-sum, we retain
the zero term in this unshifted mass. Consequently, its leading
Gaussian mass is
\[
\rho_{\sigma_{j+1}/\sqrt{1-1/n}}(\cL)
\le
2^{((m-j-1)/2+o(1))n},
\]
whereas for the $H_{j+1}$-sum it is
\[
\rho_{\sigma_{j+1}/\sqrt{1-1/n}}(K_{j+1})
\le
1+2^{-\Omega(n)}.
\]
All difference-variable masses are bounded as in the corresponding
SVP proofs using
\cref{eqn:svp-random-lattice-gradient-mass} and the bounds for
$K_0$. Therefore,
\[
\sum_{\kappa\in R_j}T_{j,y}(\kappa)
\le
\left(\frac{s^2\lambda}{q}\right)^2
2^{((m-j-1)/2+o(1))n},
\]
and, including the case $j=0$,
\[
\sum_{\gamma\in H_{j+1}}T_{j,y}(\kappa+\gamma)
\le
\left(\frac{s^2\lambda}{q}\right)^22^{o(n)}.
\]
These comparisons hold for every shift, so the bounds hold
simultaneously for every $y\in\R^n$, completing the proof.
\end{proof}

\paragraph{The nonprincipal mass.}
Fix $w\in\F_q^n\setminus\{0\}$ independently of $U$, let
$k_*:=U^{-T}w$, and, for $j\ge1$, define
\[
\Delta_{j,y}
:=\sum_{\alpha\in H_j\setminus\{0\}}
V_{j-1,y}(k_*+\alpha)^2,\qquad
\Gamma_{j,y}
:=\sum_{\alpha\in H_j\setminus\{0\}}
T_{j-1,y}(k_*+\alpha)V_{j-1,y}(k_*+\alpha).
\]

\begin{lemma}
\label{lem:cvp-target-increments}
Let $\lambda=\lambda_1(\cL)\le d\le(1+1/n)\lambda$. For fixed
$y$ and $w\ne0$, except with probability $o(1)$ over $U$,
simultaneously for every $j\in\{1,\ldots,m-1\}$,
\begin{align}
\Delta_{j,y}
&\le2^{-n/2+o(n)},&
\Gamma_{j,y}
&\le
\left(\frac{s^2\lambda}{q}\right)^22^{-n/2+o(n)}.
\label{eqn:cvp-target-increment-bounds}
\end{align}
\end{lemma}

\begin{proof}
We follow the proof of \cref{lem:svp-target-increments}, using the total bounds in \cref{eqn:cvp-total-mass-bounds}. Under the
bijection
$\kappa+P_{j-1}\mapsto U^T\kappa+S_{j-1}$, the terms in $\Delta_{j,y}$ and $\Gamma_{j,y}$ above with
$\alpha\ne0$ correspond to the $S_{j-1}$-cosets other than
$w+S_{j-1}$ that are contained in $w+S_j$. Conditional on
$S_0,\ldots,S_{j-1}$, each fixed such coset is included with
probability at most $2q^{-(n-jn/m)}$.

Let $\mathcal G_j$ be the event that
\cref{eqn:cvp-total-mass-bounds} holds at level $j$. As in the SVP
proof, summing \cref{eqn:cvp-frequency-mass-recursion} over $R_j$
shows that the total weights assigned to the $S_{j-1}$-cosets are
$\sum_{\kappa\in R_j}V_{j,y}(\kappa)$ and
$\sum_{\kappa\in R_j}T_{j,y}(\kappa)$. These weights and
$\mathcal G_j$ are fixed after conditioning on
$S_0,\ldots,S_{j-1}$. Hence linearity of conditional expectation
gives
\begin{align*}
\E[\ind_{\mathcal G_j}\ind_{\{w\notin S_{m-1}\}}\Delta_{j,y}
\mid S_0,\ldots,S_{j-1}]
&\le
2q^{-(n-jn/m)}2^{((m-j-1)/2+o(1))n},\notag\\
\E[\ind_{\mathcal G_j}\ind_{\{w\notin S_{m-1}\}}\Gamma_{j,y}
\mid S_0,\ldots,S_{j-1}]
&\le
\left(\frac{s^2\lambda}{q}\right)^2
2q^{-(n-jn/m)}2^{((m-j-1)/2+o(1))n}.
\end{align*}
Since
$\log_2q=(m+1)/2+o(1)$, the common factor on the right is
\[
q^{-(n-jn/m)}2^{((m-j-1)/2+o(1))n}
=2^{-(1/2+(m-j)/(2m)-o(1))n}.
\]
Markov's inequality, with a slightly larger $o(n)$ term in the
threshold, and a union bound over $j$ prove
\cref{eqn:cvp-target-increment-bounds} on
$\bigcap_j\mathcal G_j$. Finally,
$\Pr[w\in S_{m-1}]=Q^{-1}(1+o(1))=o(1)$, and
$\Pr[\bigcap_j\mathcal G_j]=1-2^{-\Omega(n)}$ by
\cref{lem:cvp-shifted-mass-bounds}.
\end{proof}

\begin{proof}[Proof of \cref{lem:cvp-nonprincipal-mass}]
Let $k_*=U^{-T}w$ and define
\begin{align*}
V_{j,y}^*&:=|f_y(w)|^{2^{j+1}},&
B_{j,y}&:=V_{j,y}(k_*)-V_{j,y}^*,\notag\\
T_{j,y}^*&:=\norm{G_y(w)}^2|f_y(w)|^{2(2^j-1)},&
E_{j,y}&:=T_{j,y}(k_*)-T_{j,y}^*.
\end{align*}
We follow the proof of \cref{lem:svp-nonprincipal-mass}. The
choice $\alpha=0$ at every level gives the displayed
terms above, so $B_{0,y}=E_{0,y}=0$. Separating that choice in
\cref{eqn:cvp-frequency-mass-recursion} gives
\begin{align*}
B_{j,y}
&=2V_{j-1,y}^*B_{j-1,y}+B_{j-1,y}^2+\Delta_{j,y},\notag\\
E_{j,y}
&=E_{j-1,y}(V_{j-1,y}^*+B_{j-1,y})
+T_{j-1,y}^*B_{j-1,y}+\Gamma_{j,y}.
\end{align*}
By the primal representation and \cref{lem:shifted-mass-maximum}, $0<f_y(w)\le1$. Moreover,
\cref{lem:cvp-local} and $\norm{v-y}\le\alpha\lambda$ give
\[
\norm{G_y(w)}
\le O\left(\frac{s^2\lambda}{q}\right)f_y(w),
\qquad
T_{j,y}^*
\le O\left(\left(\frac{s^2\lambda}{q}\right)^2\right).
\]
Using \cref{lem:cvp-target-increments} in the two recurrences, the
same induction as in the SVP proof gives
\[
B_{j,y}\le2^{-n/2+o(n)},\qquad
E_{j,y}\le
\left(\frac{s^2\lambda}{q}\right)^22^{-n/2+o(n)}
\]
for every $j\le m-1$. The factor accumulated over
$m=O(\log n)$ levels is absorbed into the $o(n)$ terms. Finally,
$k_*-\theta_*\in P_{m-1}$, and the term obtained by choosing
$\alpha=0$ in every recurrence at $k_*$ is the term indexed by
$\boldsymbol\alpha^*$ at $\theta_*$. Hence
\cref{lem:cvp-frequency-mass-closed-form} identifies
$E_{m-1,y}$ with the sum in
\cref{eqn:cvp-nonprincipal-mass}, proving the lemma.
\end{proof}

\paragraph{The second moment bound.}
We now introduce the CVP version of the quantities used in
\cref{lem:svp-omitted-factor-masses}. For
$I\subseteq[0,2^j)$, let
\[
V_{j,I,y}(\kappa)
:=
\sum_{\boldsymbol\alpha^{\le j}\in\mathcal A_{\le j}}
\prod_{r\in[0,2^j)\setminus I}
\left|f_y\left(
\varepsilon_rU^Tk_{\boldsymbol\alpha^{\le j},\kappa,r}
\right)\right|^2.
\]
For $0\notin I$, define
\[
T_{j,I,y}(\kappa)
:=
\sum_{\boldsymbol\alpha^{\le j}\in\mathcal A_{\le j}}
\norm{
G_y\left(U^Tk_{\boldsymbol\alpha^{\le j},\kappa,0}\right)
\prod_{r\in[1,2^j)\setminus I}
f_y\left(
\varepsilon_rU^Tk_{\boldsymbol\alpha^{\le j},\kappa,r}
\right)
}^2.
\]
Thus $V_{j,\varnothing,y}=V_{j,y}$ and
$T_{j,\varnothing,y}=T_{j,y}$.

\begin{lemma}
\label{lem:cvp-omitted-factor-masses}
Let $\lambda=\lambda_1(\cL)\le d\le(1+1/n)\lambda$. On the event
in \cref{lem:cvp-shifted-mass-bounds}, simultaneously
for every $y$, $j\in\{0,\ldots,m-1\}$,
$\kappa\in\F_q^n$, and every nonempty
$I\subseteq[0,2^j)$,
\[
V_{j,I,y}(\kappa)
\le Q^{|I|-1}2^{o(n)},\qquad
T_{j,I,y}(\kappa)
\le
\left(\frac{s^2\lambda}{q}\right)^2
Q^{|I|-1}2^{o(n)}
\quad\text{if }0\notin I.
\]
\end{lemma}

\begin{proof}
We follow the three-case induction in the proof of
\cref{lem:svp-omitted-factor-masses}. For $j=0$, the only nonempty
choice is $I=\{0\}$ and $V_{0,\{0\},y}=1$. Suppose that $j\ge1$ and let
$I_L:=I\cap[0,2^{j-1})$ and
$I_R:=\{r\in[0,2^{j-1}):2^{j-1}+r\in I\}$. The reindexing in
\cref{eqn:cvp-right-subtree-reindexing} gives
\begin{align*}
V_{j,I,y}(\kappa)
&=\sum_{\alpha\in H_j}
V_{j-1,I_L,y}(\kappa+\alpha)
V_{j-1,I_R,y}(\kappa+\alpha),\notag\\
T_{j,I,y}(\kappa)
&=\sum_{\alpha\in H_j}
T_{j-1,I_L,y}(\kappa+\alpha)
V_{j-1,I_R,y}(\kappa+\alpha),
\end{align*}
where the second identity is used only when $0\notin I$.

If exactly one of $I_L,I_R$ is nonempty, apply the induction
hypothesis for each $\alpha$ to the term whose omitted set is
nonempty. The other term is
$V_{j-1,\varnothing,y}(\kappa+\alpha)=V_{j-1,y}(\kappa+\alpha)$,
whose sum over $\alpha\in H_j$ is bounded by the first inequality in
\cref{eqn:cvp-fixed-mass-bounds}. If both $I_L$ and $I_R$ are
nonempty, apply the induction hypothesis to both terms for each
$\alpha$ and sum over the $Q$ choices of $\alpha$. The resulting
power of $Q$ is
\[
Q^{|I_L|-1}Q^{|I_R|-1}Q=Q^{|I|-1}.
\]

For the bound on $T_{j,I,y}$, the assumption $0\notin I$ implies
$0\notin I_L$. If $I_L\ne\varnothing$ and $I_R=\varnothing$, apply
the induction hypothesis to
$T_{j-1,I_L,y}(\kappa+\alpha)$ for each $\alpha$ and use the first
inequality in \cref{eqn:cvp-fixed-mass-bounds} to sum
$V_{j-1,y}(\kappa+\alpha)$ over $\alpha$. If
$I_L=\varnothing$ and $I_R\ne\varnothing$, use the second inequality
in \cref{eqn:cvp-fixed-mass-bounds} to sum
$T_{j-1,y}(\kappa+\alpha)$ over $\alpha$, and apply the induction
hypothesis to $V_{j-1,I_R,y}(\kappa+\alpha)$ for each $\alpha$. If
both $I_L$ and $I_R$ are nonempty, apply the induction hypotheses to
$T_{j-1,I_L,y}(\kappa+\alpha)$ and
$V_{j-1,I_R,y}(\kappa+\alpha)$ for each $\alpha$, and then sum over
the $Q$ choices of $\alpha$.

For completeness, the first inequality in
\cref{eqn:cvp-fixed-mass-bounds} is invoked at most
$2^j\le p$ times, and the resulting multiplicative contribution is
\[
(1+2^{-\Omega(n)})^p=1+2^{-\Omega(n)}.
\]
At every level of the recursion for $T_{j,I,y}$, exactly one term is
of the form $T_{j',I',y}$. If $I'=\varnothing$, its sum is bounded
directly by the second inequality in
\cref{eqn:cvp-fixed-mass-bounds}; otherwise the induction continues
through this term. Hence the second inequality is invoked at most
once, and all accumulated factors are absorbed into $2^{o(n)}$.
\end{proof}

For $\boldsymbol i\in[N]^p$ and $\theta\in H_m$, denote the
corresponding summand of \cref{eqn:cvp-estimator} by
\[
\mathcal K^y_{\boldsymbol i}(\theta)
:=
Q^{p-1}C_{\boldsymbol\delta}(\boldsymbol i)
2\pi iX_{0,i_0}\Phi_y(\boldsymbol i)
\omega_q^{\ip{\theta}{(Y_{m-1,0}(\boldsymbol i))_{H_m}}}.
\]
Then
$\widehat A_{y,\boldsymbol\delta}(\theta)
=N^{-p}\sum_{\boldsymbol i}\mathcal K^y_{\boldsymbol i}(\theta)$
and
$\E_X[\mathcal K^y_{\boldsymbol i}(\theta)\mid
U,\boldsymbol\delta]=A_{y,\boldsymbol\delta}(\theta)$.

\begin{lemma}
\label{lem:cvp-overlap-bound}
Let $\lambda=\lambda_1(\cL)\le d\le(1+1/n)\lambda$. On the event
in \cref{lem:cvp-shifted-mass-bounds}, simultaneously
for every $y$, $\theta\in H_m$, and
$\boldsymbol i,\boldsymbol i'\in[N]^p$, let
$I:=\{r\in[0,p):i_r=i'_r\}$. If $I=\varnothing$, then
\[
\E_{\boldsymbol\delta,X}\left[
(\mathcal K^y_{\boldsymbol i}(\theta)-A_{y,\boldsymbol\delta}(\theta))^*
(\mathcal K^y_{\boldsymbol i'}(\theta)-A_{y,\boldsymbol\delta}(\theta))
\mid U\right]=0.
\]
If $I\ne\varnothing$, the absolute value of this expectation is at
most
\begin{align}
\left(\frac{s^2\lambda}{q}\right)^2
Q^{|I|-1}2^{o(n)}.
\label{eqn:cvp-overlap-bound}
\end{align}
\end{lemma}

\begin{proof}
We follow the proof of \cref{lem:svp-overlap-covariance} and spell
out the target phases. If $I=\varnothing$, then, for fixed
$U$ and $\boldsymbol\delta$, the two collections
$(X_{r,i_r})_{r\in[0,p)}$ and
$(X_{r,i'_r})_{r\in[0,p)}$ are independent. Each displayed
factor of the form $\mathcal K^y_{\boldsymbol i}(\theta)
-A_{y,\boldsymbol\delta}(\theta)$ has expectation zero over its
collection, so the expectation of their product is zero.

Suppose that $I\ne\varnothing$. Expand
\[
Q^{p-1}C_{\boldsymbol\delta}(\boldsymbol i)
\omega_q^{\ip{\theta}{(Y_{m-1,0}(\boldsymbol i))_{H_m}}}
\]
and the analogous expression indexed by $\boldsymbol i'$ using
\cref{eqn:svp-all-characters}. 
Averaging over $\boldsymbol\delta$ and applying
\cref{eqn:svp-leaf-frequency-identity,lem:svp-offset-character-orthogonality} makes every term with
$\boldsymbol\alpha\ne\boldsymbol\alpha'$ vanish. Hence only the terms
with $\boldsymbol\alpha=\boldsymbol\alpha'$ remain.

Fix $\boldsymbol\alpha$ and consider the term with
$\boldsymbol\alpha'=\boldsymbol\alpha$. For $r\in[0,p)$, let $
z_{\boldsymbol\alpha,r}
:=
U^Tk_{\boldsymbol\alpha,\theta,r}$
and define
\[
\psi_{\boldsymbol\alpha,r}(X)
:=
\omega_q^{\ip{z_{\boldsymbol\alpha,r}}{[X]_q}}
\exp\left(
-\frac{2\pi i\varepsilon_r}{q}\ip{X}{y}
\right).
\]
Then, it holds that
\[
\E_X[\psi_{\boldsymbol\alpha,r}(X)]
=f_y(\varepsilon_rz_{\boldsymbol\alpha,r}),\qquad
\E_X[2\pi iX\psi_{\boldsymbol\alpha,0}(X)]
=G_y(z_{\boldsymbol\alpha,0}).
\]
Here the first identity uses the evenness of $F_{1/s}$. If
$r\in[1,p)\setminus I$, the two samples are independent and hence
give $|f_y(\varepsilon_rz_{\boldsymbol\alpha,r})|^2$. If
$r\in I\setminus\{0\}$, the same sample occurs in both copies and
$|\psi_{\boldsymbol\alpha,r}(X)|^2=1$; thus both its character and
its target phase cancel. At $r=0$, the resulting factor is
$\norm{G_y(z_{\boldsymbol\alpha,0})}^2$ when $0\notin I$, and
$\E_{X\sim D_{\cL^*,s}}[\norm{2\pi X}^2]$ when $0\in I$.
Consequently,
\begin{align}
\E_{\boldsymbol\delta,X}\left[
(\mathcal K^y_{\boldsymbol i}(\theta))^*
\mathcal K^y_{\boldsymbol i'}(\theta)\mid U\right]
=
\begin{cases}
T_{m-1,I,y}(\theta),&0\notin I,\\
\E_{X\sim D_{\cL^*,s}}[\norm{2\pi X}^2]
V_{m-1,I,y}(\theta),&0\in I.
\end{cases}
\label{eqn:cvp-overlap-joint-moment}
\end{align}
Splitting the expectation according to
$\norm X\le C s\sqrt n$ and its complement, and applying
\cref{lem:dgs-tail}, gives
\[
\E_{X\sim D_{\cL^*,s}}[\norm{2\pi X}^2]
\le O(ns^2)
\le
\left(\frac{s^2\lambda}{q}\right)^22^{o(n)}.
\]
Therefore, \cref{lem:cvp-omitted-factor-masses} bounds the absolute
value of \cref{eqn:cvp-overlap-joint-moment} by the right-hand side
of \cref{eqn:cvp-overlap-bound}.

Applying the same character orthogonality to two independent copies
gives
\begin{align}
\E_{\boldsymbol\delta}\left[
\norm{A_{y,\boldsymbol\delta}(\theta)}^2\mid U\right]
=T_{m-1,y}(\theta)
\le
\sum_{\kappa\in R_{m-1}}T_{m-1,y}(\kappa)
\le
\left(\frac{s^2\lambda}{q}\right)^22^{o(n)},
\label{eqn:cvp-overlap-mean-product}
\end{align}
where $R_{m-1}=H_m$ and we used
\cref{eqn:cvp-total-mass-bounds}. Expanding the two subtractions in
the expectation in the lemma shows that it is the difference between
\cref{eqn:cvp-overlap-joint-moment,eqn:cvp-overlap-mean-product}.
Since $I\ne\varnothing$, the triangle inequality proves
\cref{eqn:cvp-overlap-bound}, after absorbing a factor of two into
$2^{o(n)}$.
\end{proof}

\begin{proof}[Proof of \cref{lem:cvp-finite-list-second-moment}]
This is the same counting argument as in the proof of
\cref{lem:svp-finite-list-second-moment}. Work on the event in
\cref{lem:cvp-shifted-mass-bounds}, on which
\cref{lem:cvp-overlap-bound} holds. For a fixed nonempty
$I\subseteq[0,p)$, at most $N^{2p-|I|}$ ordered pairs
$(\boldsymbol i,\boldsymbol i')$ have overlap set $I$. Expanding the
squared norm of the estimator and applying
\cref{eqn:cvp-overlap-bound} give
\begin{align*}
&\E_{\boldsymbol\delta,X}\left[
\norm{\widehat A_{y,\boldsymbol\delta}(\theta)
-A_{y,\boldsymbol\delta}(\theta)}^2\mid U\right]\notag\\
&\quad\le
\frac1{N^{2p}}
\sum_{\varnothing\ne I\subseteq[0,p)}
N^{2p-|I|}
\left(\frac{s^2\lambda}{q}\right)^2
Q^{|I|-1}2^{o(n)}\notag\\
&\quad=
\left(\frac{s^2\lambda}{q}\right)^2Q^{-1}2^{o(n)}
\left[\left(1+\frac QN\right)^p-1\right]
\le
\left(\frac{s^2\lambda}{q}\right)^2Q^{-1}2^{o(n)},
\end{align*}
where the last inequality follows from $N=Qp^4$. The event is
uniform in $y$ and $\theta$, proving the lemma.
\end{proof}